%% file: main.tex
\documentclass[acmsmall,dvipsnames,x11names]{acmart}
\renewcommand\footnotetextcopyrightpermission[1]{} % removes footnote with conference information in first column
\pdfoutput=1
\acmJournal{PACMPL}
\acmVolume{1}
\acmNumber{POPL} % CONF = POPL or ICFP or OOPSLA
\acmYear{2020}
\acmMonth{10}
\startPage{1}

\setcopyright{none}
\input{macros}

\usepackage{hyperref}
\usepackage{subcaption}                        
\usepackage{adjustbox}
\usepackage{amsmath,amsfonts}
\usepackage{marvosym}

\usepackage{cleveref}
\usepackage{graphicx}
\usepackage{nameref}
\usepackage{stmaryrd}
\usepackage{xcolor}
\usepackage{xfrac}
\usepackage{xspace}
\usepackage{cancel}
\usepackage{csquotes}
\usepackage{nicefrac}

\usepackage{listings,multicol}

\usepackage{tensor}
\usepackage{tikz}
\usetikzlibrary{patterns,decorations.pathreplacing,arrows,arrows.meta,decorations.pathmorphing,positioning,fit,trees,shapes,shadows,automata,calc}
\usepackage{pgfplots}

\usepackage{caption}
\usepackage{proof}
\usepackage{mathtools}

\allowdisplaybreaks

\begin{document}

%% Title information
\title{Relatively Complete Verification of Probabilistic Programs} %(Submission + Appendix)}         %% [Short Title] is optional;
\subtitle{An Expressive Language for Expectation-based Reasoning}                     %% \subtitle is optional
%\subtitlenote{with subtitle note}       %% \subtitlenote is optional;
                                        %% can be repeated if necessary;
                                        %% contents suppressed with 'anonymous'

%% Author information
%% Contents and number of authors suppressed with 'anonymous'.
%% Each author should be introduced by \author, followed by
%% \authornote (optional), \orcid (optional), \affiliation, and
%% \email.
%% An author may have multiple affiliations and/or emails; repeat the
%% appropriate command.
%% Many elements are not rendered, but should be provided for metadata
%% extraction tools.

%% Author with single affiliation.
\author[Batz]{Kevin Batz}
\authornote{Batz and Katoen are supported by the ERC AdG 787914 FRAPPANT.}
%\authornote{with author1 note}          %% \authornote is optional;
                                        %% can be repeated if necessary
%\orcid{nnnn-nnnn-nnnn-nnnn}             %% \orcid is optional
\affiliation{
  %\position{Position1}
  %\department{Software Modeling and Verification Group}              %% \department is recommended
  \institution{RWTH Aachen University, Germany}            %% \institution is required
  %\streetaddress{Street1 Address1}
  %\city{City1}
  %\state{State1}
  %\postcode{Post-Code1}
  \country{Germany}                    %% \country is recommended
}
\email{kevin.batz@cs.rwth-aachen.de}          %% \email is recommended

\author[Kaminski]{Benjamin Lucien Kaminski}
%\authornote{with author1 note}          %% \authornote is optional;
                                        %% can be repeated if necessary
%\orcid{nnnn-nnnn-nnnn-nnnn}             %% \orcid is optional
\affiliation{
  %\position{Position1}
  %\department{Software Modeling and Verification Group}              %% \department is recommended
  \institution{University College London, United Kingdom}            %% \institution is required
  %\streetaddress{Street1 Address1}
  %\city{City1}
  %\state{State1}
  %\postcode{Post-Code1}
  \country{United Kingdom}                    %% \country is recommended
}
\email{b.kaminski@ucl.ac.uk}          %% \email is recommended

\author[Katoen]{Joost-Pieter Katoen}
\authornotemark[1]
%\authornote{with author1 note}          %% \authornote is optional;
                                        %% can be repeated if necessary
%\orcid{nnnn-nnnn-nnnn-nnnn}             %% \orcid is optional
\affiliation{
  %\position{Position1}
  %\department{Software Modeling and Verification Group}              %% \department is recommended
  \institution{RWTH Aachen University, Germany}            %% \institution is required
  %\streetaddress{Street1 Address1}
  %\city{City1}
  %\state{State1}
  %\postcode{Post-Code1}
  \country{Germany}                    %% \country is recommended
}
\email{katoen@cs.rwth-aachen.de}          %% \email is recommended

\author[Matheja]{Christoph Matheja}
%\authornote{with author1 note}          %% \authornote is optional;
                                        %% can be repeated if necessary
%\orcid{nnnn-nnnn-nnnn-nnnn}             %% \orcid is optional
\affiliation{
  %\position{Position1}
  %\department{Software Modeling and Verification Group}              %% \department is recommended
  \institution{ETH Z\"urich, Switzerland}            %% \institution is required
  %\streetaddress{Street1 Address1}
  %\city{City1}
  %\state{State1}
  %\postcode{Post-Code1}
  \country{Switzerland}                    %% \country is recommended
}
\email{cmatheja@inf.ethz.ch}          %% \email is recommended

%% Abstract
%% Note: \begin{abstract}...\end{abstract} environment must come
%% before \maketitle command
\begin{abstract}
        We study \emph{a syntax for specifying quantitative \enquote{assertions}}---functions mapping program states to numbers---for probabilistic program verification.
        We prove that our syntax is expressive in the following sense: Given any probabilistic program $\cc$, if a function $\FF$ is expressible in our syntax, then the function mapping each initial state $\pstate$ to the expected value of $\FF$ evaluated in the final states reached after termination of $\cc$ on $\sigma$ (also called the weakest preexpectation $\wp{\cc}{\FF}$) is also expressible in our syntax.
        
        As a consequence, we obtain a \emph{relatively complete verification system}  for reasoning about expected values and probabilities in the sense of Cook:
        %Apart from a single reasoning step about the inequality of two functions expressed as syntactic expressions in our language, given $\FF$, $\FG$, and $\cc$, we can check whether $\FG \preceq \wp{\cc}{\FF}$.
        Apart from proving a single inequality between two functions given by syntactic expressions in our language, 
        given $\FF$, $\FG$, and $\cc$, we can check whether $\FG \preceq \wp{\cc}{\FF}$. \\ \\
        \noindent
        \emph{\today: This is a revised version, correcting technical issues in the proofs of Theorem \ref{thm:dedekind_nf} and \ref{thm:prod_exp}.}
\end{abstract}

%% Keywords
%% comma separated list
%\keywords{probabilistic programs, randomized algorithms, formal verification, quantitative verification, quantitative reasoning}

%% \maketitle
%% Note: \maketitle command must come after title commands, author
%% commands, abstract environment, Computing Classification System
%% environment and commands, and keywords command.
\maketitle

\input{introduction}

\input{extensional}

\input{towards}

\input{syntax}

\input{semantics}
\input{proof_straight_line}

\input{proof_outline}

%\input{multiplication}

\input{first_order_arithmetic}

% Bibliography

\input{normalforms}

\input{sums_prod_goedel}
\input{expressiveness}

\input{negatives}

\input{applications}

\input{conclusion}

\subsubsection*{Acknowledgements.}
We thank David N.\ Jansen for his feedback on an older version of this paper.

\bibliography{literature}

\appendix
\newpage

\section{Appendix to Section~\ref{sec:normalforms} (The Dedekind Normal Form)}

\input{appendix_normalforms}

%%\section{Appendix to Section~\ref{sec:product_of_exp} (Products of Expectations)}
%%\input{appendix_products}

\section{Appendix to Section~\ref{sec:embedding} (G\"odelization for Syntactic Expectations)}

\input{appendix_embedding}

\section{Appendix to Section~\ref{sec:sums_prod_via_goedel} (Sums, Products, and Infinite Series of Syntactic Expectations)}

\input{appendix_sums_prod_via_goedel}

\section{Appendix to Section~\ref{sec:expressiveness} (Expressiveness of our Language)}
\label{app:expressiveness}
\input{appendix_expressiveness}

\end{document}

%% file: macros.tex
\newcommand{\ie}{i.e.,\ }
\newcommand{\eg}{e.g.,\ }

\newcommand{\Sup}{\reflectbox{\textnormal{\textsf{\fontfamily{phv}\selectfont S}}}\hspace{.2ex}}
\newcommand{\Inf}{\raisebox{.6\depth}{\rotatebox{-30}{\textnormal{\textsf{\fontfamily{phv}\selectfont \reflectbox{J}}}}\hspace{-.1ex}}}
\newcommand{\Quant}{\reflectbox{\textnormal{\textsf{\fontfamily{phv}\selectfont Q}}}\hspace{.2ex}}
\newcommand{\SupV}[1]{\Sup #1\colon}
\newcommand{\InfV}[1]{\Inf #1\colon}
\newcommand{\qprefix}[1]{\sfsymbol{Prefix} \left( #1 \right) \xspace}
\newcommand{\qprefixnoarg}{\sfsymbol{Prefix}\xspace}
\newcommand{\qprefixnoarginvert}{\overline{\qprefixnoarg}\xspace}

\newcommand{\interpret}{\ensuremath{\mathfrak{I}}\xspace}
\newcommand{\pstate}{\pstatea}
\newcommand{\pstatea}{\ensuremath{\sigma}\xspace}
\newcommand{\pstateb}{\ensuremath{\tau}\xspace}

\newcommand{\relPrime}[2]{#1 \bot #2}

\newcommand{\TT}{\TTa\xspace}
\newcommand{\TTa}{\ensuremath{a}\xspace}
\newcommand{\TTb}{\ensuremath{b}\xspace}
\newcommand{\TTc}{\ensuremath{c}\xspace}

\newcommand{\BB}{\ensuremath{\BBa}\xspace}
\newcommand{\BBa}{\ensuremath{\varphi}\xspace}
\newcommand{\BBb}{\ensuremath{\psi}\xspace}
\newcommand{\BBc}{\ensuremath{\xi}\xspace}

\newcommand{\FF}{\ensuremath{f}\xspace}
\newcommand{\FG}{\ensuremath{g}\xspace}
\newcommand{\FH}{\ensuremath{h}\xspace}

\newcommand{\RR}{\ensuremath{\RRa}\xspace}
\newcommand{\RRa}{\ensuremath{r}\xspace}
\newcommand{\RRb}{\ensuremath{s}\xspace}
\newcommand{\RRc}{\ensuremath{t}\xspace}

\newcommand{\XX}{\ensuremath{x}\xspace}
\newcommand{\XY}{\ensuremath{y}\xspace}
\newcommand{\XZ}{\ensuremath{z}\xspace}

\newcommand{\VV}{\ensuremath{v}\xspace}
\newcommand{\VW}{\ensuremath{w}\xspace}
\newcommand{\VU}{\ensuremath{u}\xspace}

\newcommand{\VVcut}{\VV_{\sfsymbol{Cut}}}

\newcommand{\toExp}[1]{\iverson{#1}}

\newcommand{\PP}{\ensuremath{P}\xspace}
\newcommand{\isNat}{\ensuremath{N}\xspace}

\newcommand{\dcut}[1]{\ensuremath{\sfsymbol{Cut} \left( #1 \right)} \xspace}

\newcommand{\dcutzero}[1]{\ensuremath{\sfsymbol{\underline{Cut}} \left( #1 \right)} \xspace}

\newcommand{\DDvar}[2]{\ensuremath{D[#1,  #2 ]}\xspace}
\newcommand{\DD}[1]{\DDvar{\VVcut}{#1}}
\newcommand{\dexpvar}[2]{\ensuremath{\sfsymbol{Dedekind} [#1,  #2 ]} \xspace}
\newcommand{\dexp}[1]{\dexpvar{\VVcut}{#1}}

\newcommand{\gnum}{num}

\newcommand{\seqelemsymbol}{\sfsymbol{Elem}}
\newcommand{\seqelem}[3]{\ensuremath{\seqelemsymbol\left( #1, #2 ,#3 \right)}\xspace}
\newcommand{\rseqelemsymbol}{\sfsymbol{RElem}}

\newcommand{\rseqelem}[3]{\ensuremath{\rseqelemsymbol\left( #1, #2 ,#3 \right)}\xspace}
\newcommand{\seqnum}[1]{\ensuremath{\langle #1 \rangle}\xspace}
\newcommand{\stateseqnum}[1]{\ensuremath{\langle (#1) \rangle}\xspace}

\newcommand{\stateseq}[3]{\ensuremath{\sfsymbol{StateSequence}_{\varseq{#1}}\left(#2, #3 \right)}\xspace}

\newcommand{\sequencesymbol}{\sfsymbol{Sequence}}
\newcommand{\sequence}[2]{\sequencesymbol \left(#1, #2\right)}

\newcommand{\rsequencesymbol}{\sfsymbol{RSequence}}
\newcommand{\rsequence}[2]{\rsequencesymbol \left(#1, #2\right)}

\newcommand{\encodesstate}[2]{\sfsymbol{EncodesState}_{\varseq{#1}} \left( #2 \right)}

\newcommand{\gPair}{\sfsymbol{Pair}\xspace}

\newcommand{\facexp}[1]{\ensuremath{\sfsymbol{Fac} \left( #1\right)} \xspace}

\newcommand{\harmexp}[1]{\ensuremath{\sfsymbol{Harmonic} \left( #1\right)} \xspace}

\newcommand{\gproductsymbol}{\sfsymbol{Product}}
\newcommand{\gproductvar}[3]{\ensuremath{\gproductsymbol \left[ #1, #2, #3 \right]}\xspace}
\newcommand{\gproduct}[2]{\gproductvar{\vprod}{#1}{#2}}

\newcommand{\vprod}{\VV_{\textnormal{prod}} \xspace}

\newcommand{\gsumsymbol}{\sfsymbol{Sum}}

\newcommand{\gsumvar}[3]{\ensuremath{\gsumsymbol \left[ #1, #2, #3 \right]}\xspace}

\newcommand{\gsum}[2]{\gsumvar{\vsum}{#1}{#2}}

\newcommand{\gapply}[3]{\sfsymbol{Subst}_{\varseq{#1}} \left[ #2, #3  \right]}

\newcommand{\gsubst}[3]{\sfsymbol{Subst}_{\varseq{#1}} \left[ #2,#3\right]}

\newcommand{\vsum}{\VV_{\textnormal{sum}} \xspace}

 \newcommand{\pathexpsymbol}{\sfsymbol{Path}}
 
 \newcommand{\pathexppost}[3]{\pathexpsymbol\left[ #1 \right] \left( #2,#3 \right)}
 
\newcommand{\pathexp}[2]{\pathexppost{\FF}{#1}{#2}}

\newcommand{\QSL}{\sfsymbol{QSL}\xspace}
\newcommand{\SL}{\sfsymbol{SL}\xspace}

\newcommand{\sfsymbol}[1]{\textsf{\upshape {#1}}}

\newcommand{\ttsymbol}[1]{\texttt{\upshape {#1}}}
\newcommand{\wpsymbol}{\sfsymbol{wp}}

\renewcommand{\wp}[2]{\wpsymbol\llbracket#1\rrbracket\left(#2\right)}

\newcommand{\wpC}[1]{\wpsymbol\llbracket#1\rrbracket}

\newcommand{\wlpsymbol}{\sfsymbol{wlp}}

\newcommand{\wlp}[2]{\wlpsymbol\llbracket#1\rrbracket\left(#2\right)}

\newcommand{\conditionalPair}[2]{{\let\oldarraystretch\arraystretch}\renewcommand{\arraystretch}{1}~\holter{~\raisebox{.5ex}{${#1}$}~}{~\raisebox{.125ex}{${#2}$}~}~\renewcommand{\arraystretch}{\oldarraystretch}}

\newcommand{\cc}{\ensuremath{C}} % programs
\newcommand{\guard}{\BB} % conditial guard
\newcommand{\ee}{\ensuremath{e}} % expressions

\newcommand{\pp}{\ensuremath{p}} % probabilities
\newcommand{\ff}{\ensuremath{X}} % expectations
\newcommand{\fg}{\ensuremath{Y}}
\newcommand{\fh}{\ensuremath{Z}}

\newcommand{\SKIP}{\ttsymbol{skip}}

\newcommand{\AssignSymbol}{\mathrel{\textnormal{\texttt{:=}}}}
\newcommand{\ASSIGN}[2]{\ensuremath{#1 \AssignSymbol #2}}

\newcommand{\AVAILLOC}[1]{\PosNats}
\usepackage{actuarialangle}

\newcommand{\COMPOSE}[2]{\ensuremath{{#1}{\,;}~ {#2}}}
\newcommand{\PCHOICE}[3]{\ensuremath{\left\{\, {#1} \,\right\}\mathrel{\left[\,#2\,\right]}\left\{\, {#3} \,\right\}}}

\newcommand{\IFSYMBOL}{\ensuremath{\textnormal{\texttt{if}}}}
\newcommand{\IF}[1]{\ensuremath{\IFSYMBOL\,\left(\, {#1} \,\right)\,\{}}
\newcommand{\ELSESYMBOL}{\ensuremath{\textnormal{\texttt{else}}}}
\newcommand{\ELSE}{\ensuremath{\}\,\ELSESYMBOL\,\{}}
\newcommand{\ITE}[3]{\ensuremath{\IFSYMBOL\,\left(\, {#1} \,\right)\,\left\{\, {#2} \,\right\}\,\ELSESYMBOL\,\left\{\, {#3} \,\right\}}}

\newcommand{\WHILESYMBOL}{\ensuremath{\textnormal{\texttt{while}}}}
\newcommand{\WHILE}[1]{\ensuremath{\WHILESYMBOL \left(\, {#1} \,\right)\left\{\right.}}

\newcommand{\WHILEDO}[2]{\ensuremath{\WHILESYMBOL \left(\, {#1} \,\right)\left\{\, {#2} \,\right\}}}

\newcommand{\pgcl}{\textnormal{\sfsymbol{pGCL}}\xspace}   % Programs
\newcommand{\Vars}{\ensuremath{\mathsf{Vars}}\xspace}   % Variables
\newcommand{\Terms}{\ensuremath{\mathsf{AExpr}}\xspace}   % Terms
\newcommand{\Bools}{\ensuremath{\mathsf{Bool}}\xspace}   % Boolean expressions
\newcommand{\Nats}{\ensuremath{\mathbb{N}}\xspace}
\newcommand{\PosNats}{\ensuremath{\mathbb{N}_{>0}}\xspace}
\newcommand{\Ints}{\ensuremath{\mathbb{Z}}\xspace}
\newcommand{\Rats}{\ensuremath{\mathbb{Q}}\xspace}

\newcommand{\PosReals}{\mathbb{R}_{\geq 0}}
\newcommand{\PosRats}{\mathbb{Q}_{\geq 0}}
\newcommand{\PosRealsInf}{\mathbb{R}_{\geq 0}^\infty}

\newcommand{\E}{\mathbb{E}}
\newcommand{\SyntE}{\ensuremath{\mathsf{Exp}}\xspace}

\newcommand{\FOArith}[1]{\mathbf{A}_{#1}}
\newcommand{\FOArithPosRats}{\FOArith{\PosRats}}
\newcommand{\FOArithNats}{\FOArith{\Nats}}

\newcommand{\monus}{\mathbin{\dot-}}
\newcommand{\abs}[1]{|#1|}

\newcommand{\exprod}{{}\odot{}}

\newcommand{\FV}[1]{\ensuremath{\sfsymbol{FV}\left(#1\right)}}

\newcommand{\toFOPosRats}[1]{#1_{\PosRats}}

\newcommand{\iverson}[1]{\left[ {#1} \right]}

\newcommand{\subst}[2]{\left[ {#1} \middle/ {#2}\right]}
\newcommand{\statesubst}[2]{\left[ {#1} \mapsto {#2}\right]}

\newcommand{\charwpsym}{\Phi}
\newcommand{\charwp}[3]{\charwpsym_{#3}} %\charfun{\wpsymbol}{#1}{#2}{#3}}
\newcommand{\charwpn}[4]{\charwpsym^{#4}_{#3}} %\charfun{\wpsymbol}{#1}{#2}{#3}}

\newcommand{\eval}[1]{\ensuremath{\left\llbracket {#1} \right\rrbracket}}
\newcommand{\sem}[3]{\ensuremath{\tensor*[^{#2}_{}]{\left\llbracket {#1} \right\rrbracket}{}}}
\newcommand{\semleft}[3]{\ensuremath{\tensor*[^{#2}_{}]{\llbracket {#1} }{}}}
\newcommand{\semright}{\rrbracket}

\newcommand{\States}{\Sigma}
\newcommand{\To}{\rightarrow}
\newcommand{\true}{\mathsf{true}}
\newcommand{\false}{\mathsf{false}}
\newcommand{\mydot}{\text{{\Large\textbf{.}}~}}

\newcommand{\varseq}[1]{\ensuremath{\mathbf{#1}}\xspace}
\newcommand{\equivstatesrel}[1]{\sim_{#1}}

\newcommand{\equivstates}[3]{#1 \equivstatesrel{#2} #3}

\newcommand{\partitionedstates}[1]{\States_{\varseq{#1}} }
\newcommand{\statepred}[2]{\iverson{#1}_{#2}}

\newcommand{\qiff}{\quad\textnormal{iff}\quad}
\newcommand{\qqiff}{\qquad\textnormal{iff}\qquad}

\newcommand{\qqand}{\qquad\textnormal{and}\qquad}

\newcommand{\ppreceq}{~{}\preceq{}~}

\newcommand{\defeq}{~{}\triangleq{}~}
\newcommand{\ddefeq}{~{}\defeq{}~}
\newcommand{\eeq}{~{}={}~}
\newcommand{\eequiv}{~{}\equiv{}~}
\newcommand{\pplus}{~{}+{}~}
\newcommand{\ccdot}{~{}\cdot{}~}

\newcommand{\qqlongrightarrow}{\qquad{}\longrightarrow{}\qquad}
\newcommand{\mmid}{~{}|{}~}
\newcommand{\qmid}{\quad{}|{}\quad}

\newcommand{\lleq}{~{}\leq{}~}
\newcommand{\LL}{~{}<{}~}

\newcommand{\iin}{~{}\in{}~}

\newcommand{\setcomp}[2]{\left\{\, {#1} ~\middle|~ {#2} \,\right\}}
\newcommand{\blue}[1]{\textcolor{DodgerBlue3}{#1}}

\newcommand{\lfp}{\ensuremath{\textnormal{\sfsymbol{lfp}}~}}

%% file: introduction.tex
% !TEX root = ./main.tex

\section{Introduction}\label{sec:introduction}

Probabilistic programs are ordinary programs whose execution may depend on the outcome of random experiments, such as
sampling from primitive probability distributions or branching on the outcome of a coin flip.
Consequently, running a probabilistic program (repeatedly) on a single input generally gives not a single output but a \emph{probability distribution} over outputs.

Introducing randomization into computations is an important tool for the design and analysis
of \emph{efficient algorithms}~\cite{Motwani99}.
However, increasing efficiency by randomization often comes at the price of introducing a non-zero probability of producing incorrect outputs.
Furthermore, even though a program may be efficient \emph{in expectation}, individual executions may exhibit a long---even infinite---run time~\cite{DBLP:conf/rta/BournezG05,DBLP:journals/jacm/KaminskiKMO18}.

Reasoning about these probabilistic phenomena is hard.
For instance, deciding termination of probabilistic programs has been shown to be strictly more complex than for ordinary
programs~\cite{DBLP:conf/mfcs/KaminskiK15,acta19}.
Nonetheless, probabilistic program verification is an active research area.
After seminal work on probabilistic program semantics by \citet{Kozen1979,Kozen1981}, many different techniques have been developed, see \cite{DBLP:conf/popl/HartSP82} for an early example.
Modern approaches include, amongst others, 
martingale-based techniques~\cite{DBLP:conf/cav/ChakarovS13,DBLP:conf/popl/ChatterjeeFNH16,DBLP:conf/cav/ChatterjeeFG16,DBLP:conf/popl/ChatterjeeNZ17,DBLP:conf/aplas/HuangFC18,DBLP:conf/vmcai/FuC19} and
weakest-precondition-style calculi~\cite{benni_diss,McIverM05,DBLP:conf/pldi/NgoC018,DBLP:journals/jacm/KaminskiKMO18,DBLP:journals/pacmpl/BatzKKMN19}.
The former can be phrased in terms of the latter, and all of the aforementioned techniques can be understood as instances or extensions of Kozen's probabilistic propositional dynamic logic (\textsf{PPDL})~\cite{Kozen1983,Kozen1985}.

\paragraph{Probabilistic program verification, extensionally}

There are two perspectives for reasoning about programs: the \emph{extensional} and the \emph{intensional}.
Whereas intensional approaches provide a syntax, \ie a formal language, for  assertions, extensional approaches admit arbitrary assertions and dispense with considerations about syntax altogether---they treat assertions as purely mathematical entities.

A standard technique for probabilistic program verification that takes the extensional approach is the \emph{weakest preexpectation \textnormal{($\wpsymbol$)} calculus} of~\citet{McIverM05}---itself an instance of Kozen's \textsf{PPDL}~\cite{Kozen1983,Kozen1985}.
Given a probabilistic program~$\cc$ and \emph{some function} $\FF$ (called the \emph{postexpectation}), mapping (final) states to numbers, the weakest preexpectation $\wp{\cc}{\FF}$ is a mapping from (initial) states to numbers, such that%
\begin{align*}
	\wp{\cc}{\FF}(\sigma) \eeq \begin{array}{l}
		\textnormal{Expected value of $\FF$, measured in final states reached}\\
		\textnormal{after termination of $\cc$ on initial state $\sigma$}~.
		\end{array}
\end{align*}%
For probabilistic programs with \emph{discrete} probabilistic choices, the $\wpsymbol$ calculus can be defined for \emph{arbitrary} real-valued postexpectations $\FF$~\cite{benni_diss,McIverM05}.
%For programs sampling from \emph{continuous} distributions, one has to take care of measurability of $\FF$~\cite{Kozen1983,Kozen1985}, but nevertheless:
%It is not necessary to provide any \emph{syntax} for these functions.
%They are purely mathematical entities.

\paragraph{Probabilistic program verification, intensionally}
While the extensional approach often yields elegant formalisms, it is unsuitable for developing practical verification
tools, which ultimately rely on some syntax for assertions.
In particular, we cannot---in general---rely on the property, implicitly assumed in the extensional approach, that
there is no distinction between assertions representing the same mathematical entity: a tool may not realize that $4 \cdot 0.5$
and $\sum_{i=0}^{\infty} \nicefrac{1}{2^i}$ represent the same mathematical entity (the number $2$).

An example of intensional probabilistic program verification is the 
verifier of \citet{DBLP:conf/pldi/NgoC018} which specifies a simple syntax which is extensible by user-specified base and rewrite functions.
%As verification is generally computationally intractable, a first angle of attack is to actually define a syntax, i.e.\ a language, in which reasoning can be carried out.
%Typically, one then restricts to syntactic subclasses of programs and assertions for which verification can indeed be performed.
%A classical example this is Floyd-Hoare logic (or simply Hoare logic)~\cite{floyd1967assigning,Hoare1969} 
%or Dijkstra's weakest precondition calculus \cite{Dijkstra1976} together with the language of first-order arithmetic.

\paragraph{Main contribution}
Given a calculus for program verification and an assertion language, two fundamental questions immediately arise:%
\begin{enumerate}
  \item \emph{Soundness:} \emph{Are only true assertions derivable} in the calculus?
  \item \emph{Completeness}: \emph{Can every true assertion be derived} \underline{and} \emph{is it expressible} in the assertion language?
\end{enumerate}%
While soundness is typically a \emph{must} for any verification system,
completeness is---as noted by \citet{DBLP:journals/fac/AptO19} in their recent survey of 50 years of Hoare logic---a 
\enquote{subtle matter and requires careful analysis}.

In fact, to the best of our knowledge, existing probabilistic program verification techniques 
(including all of the above references amongst many other works) 
either take the extensional approach
or do not aim for completeness.
%Probabilistic program verification techniques typically take the extensional approach 
%intensional aspects are less well understood.
%While extensional probabilistic program analysis is quite well studied (see e.g., almost all 
%of the above references amongst many other works), intensional aspects are less well understood.
In this paper, we take the intensional path and make the following contribution to %intensional 
formal reasoning about probabilistic programs: % verification:
\begin{center}%
\vspace{.5em}%
\fbox{\parbox{0.98\textwidth}{%
	\begin{center}
	We provide a simple formal \emph{language of functions} for probabilistic program verification such that: % for probabilistic program analysis, such that:%
	
	\vspace{.5em}%
	If $\FF$ is syntactically expressible, then $\wp{\cc}{f}$ is syntactically expressible.
	\end{center}
}}%
\vspace{.5em}%
\end{center}
%
%\paragraph{Relative completeness}
A language from which we can draw functions $\FF$ with the above property is called \emph{expressive}.
%Although we cannot hope to achieve completeness for $\wpsymbol$---Hoare logic is already incomplete~\cite{DBLP:journals/fac/AptO19} and is subsumed by $\wpsymbol$---
Having an expressive language renders the $\wpsymbol$ calculus \emph{relatively} complete~\cite{Cook1978SoundnessAC}:
%Apart from \emph{a single reasoning} about the inequality of two functions given as syntactic expressions in our language, we can check whether $\FG \preceq \wp{\cc}{\FF}$.
%the question whether a syntactic function $\FG$ is bounded by $\wp{\cc}{\FF}$
%is reduced to comparing the syntactical functions 
%$\FF$ and $\wp{\cc}{\FF}$.
%Hence, $\wpsymbol$ is complete \emph{modulo} the fact that we need an oracle to compare syntactical functions.
Given functions $\FF$ and $\FG$ in our language and a probabilistic program $\cc$, suppose we want to verify $\FG \preceq \wp{\cc}{\FF}$, where $\preceq$ denotes the point-wise order of functions	mapping states to numbers.
Due to expressiveness,
%Since our language is expressive, 
we can effectively construct in our language a function $\FH$ representing $\wp{\cc}{\FF}$.
Hence, verification is complete \emph{modulo} checking whether the inequality $\FG \preceq \FH$ between two functions in our language holds. % $\wp{\cc}{\FF}$ and $\FG$ represent the same function.
Indeed, Hoare logic is also only complete \emph{modulo} deciding an implication between two formulae in the language of \mbox{first-order arithmetic~\cite{DBLP:journals/fac/AptO19}}.

\paragraph{Challenges and usefulness}
Notice that providing \emph{some} expressive language is rather easy:
A~single{\-}ton language that can only represent the null-function is trivially expressive since, for any program $\cc$, the expected value of $0$ is $0$. That is, $\wp{\cc}{0} = 0$.
The challenge in a quest for an expressive language for probabilistic program verification is hence to find a language that (i) is closed under taking weakest preexpectations and (ii) can express \emph{interesting (quantitative) properties}.

Indeed, our language can: % changed from "...satisfies both properties." We only speak about the latter property in this paragraph.
For instance, it is capable of expressing \emph{termination probabilities} (via~$\wp{\cc}{1}$---the expected value of the constant function 1).
These can be \emph{irrational numbers} like the reciprocal of the golden ratio~$\sfrac{1}{\varphi}$~\cite{OlmedoKKM16}.
In general, termination probabilities carry a high internal degree of complexity~\cite{acta19}.
Our language can also express probabilities over program variables on termination of a program and that can be expressed in terms of $\pi$, $\sqrt 3$ and so forth. These can e.g., be generated by Buffon machines, \ie probabilistic programs that only use Bernoulli experiments~\cite{DBLP:conf/soda/FlajoletPS11}.

Termination probabilities already hint at one of the technical challenges we face:
Even starting from a constant function like $1$, our language has to be able to express mappings from states to highly complex real numbers.
Another challenge we face is that when constructing $\wp{\cc}{\FF}$, due to probabilistic branching in combination with loops, considering single execution traces is not enough:
We have to collect all terminating traces and average over the values of $\FF$ in terminal states.
We attack these challenges via Gödel numbers for rational sequences and encodings of Dedekind cuts.

Aside from termination probabilities, our language is capable of expressing \emph{a wide range of practically relevant functions}, like \emph{polynomials} or \emph{Harmonic numbers}.
Polynomials are a common subclass of ranking functions\footnote{In probabilistic program analysis terminology: ranking supermartingales.} for automated probabilistic termination analysis; harmonic numbers are ubiquitous in expected runtime analysis. 
We present more scenarios covered by our syntax and avenues for future work in Sections~\ref{sec:applications} and~\ref{sec:conclusion}.

Overall, we believe that an expressive \emph{syntax} for probabilistic program verification is what really expedites a search for tractable fragments of both programs and \enquote{assertion} language in the first place.
Studying such fragments may also yield additional insights:
For example, \citet{DBLP:journals/tocl/Kozen00} and \citet{DBLP:journals/isci/KozenT01} studied the propositional fragment of Hoare logic and showed that it is subsumed 
by an extension of KAT---Kleene algebra with tests.

\paragraph{Further related work}
Relative completeness of Hoare logic was shown by \citet{Cook1978SoundnessAC}.
\citet{winskel} and \citet{loeckx1984foundations} proved expressiveness of first-order arithmetic for Dijkstra's weakest precondition calculus.
For \emph{separation logic}~\cite{DBLP:conf/lics/Reynolds02}---a very successful logic for compositional reasoning about \emph{pointer programs}---expressiveness was shown by~\citet{expressiveness_sl_conference,expressiveness_sl}, almost a decade later than the logic was originally developed and started to be used.

Perhaps most directly related to this paper is the work by \citet{DBLP:journals/ijfcs/HartogV02} on a Hoare-like logic for verifying probabilistic programs.
They prove relative completeness (also in the sense of \citet{Cook1978SoundnessAC}) of their logic for \emph{loop-free} probabilistic programs and \emph{restricted postconditions}; they leave expressiveness for loops as an open problem: \enquote{It is not clear whether the probabilistic predicates are sufficiently expressive [\dots] for a given while loop.}

%\blkcommentinline{ LEFTOVER FROM CHRISTOPH: DIDN't QUITE KNOW HOW TO INTEGRATE:
%
%In the \emph{intensional approach}, where the specification language is given by a formal syntax,
%completeness is typically not achievable due to~\cite{goedel, tarski}.
%The best we can hope for is \emph{relative completeness}, i.e.\ given an oracle that discharges whether 
%\medskip
%
%Due to Tarski's undefinability theorem~\cite{tarski1936}, no sfficently 
%\medskip
%
%In the \emph{extensional approach}, where the specification language is purely mathematical rather than being defined by a formal syntax, completeness is typically immediate (cf.\cite{Nielson92}).
%However, deriving a true formula may require may require non-trivial transformations that are not restricted by any syntax:
%Since there is no distinction between the number $2$ and, for instance, the geometric series $\sum_{i=0}^{\infty} \nicefrac{1}{2}$ (which converges to $2$), a derivation that starts with the former but only works with the latter derivation is considered derivable---to the best of our knowledge, all calculi for reasoning about probabilistic programs either take the extensional approach or do not aim for completeness.
%}

\subsubsection*{Organization of the paper}

We give an introduction to \emph{syntax, extensional semantics, and verification systems} for probabilistic programs, in particular the weakest preexpectation calculus, in~\Cref{sec:extensional}.
We formulate the \emph{expressiveness problem} in \Cref{sec:towards-expressiveness}.
We \emph{define the syntax and semantics} of our \emph{expressive language of expectations} in \Cref{sec:syntax}.
We \emph{prove expressiveness of our language for loop-free probabilistic programs} in \Cref{sec:expressiveness:loop-free}.
We then move to proving expressiveness of our language for loops.
We \emph{outline the expressiveness proof for loops} in \Cref{sec:proof-outline} and \emph{do the full technical proof} throughout Sections~\ref{sec:embedding}~--~\ref{sec:expressiveness}.
%As we will evade the issue of \emph{negative numbers}, we make remarks on how to deal with them in \Cref{sec:negatives}.
%Finally, we discuss applications of our language---scenarios in which it can be applied---in \Cref{sec:applications} and conclude in \Cref{sec:conclusion}.
In \Cref{sec:negatives} and \Cref{sec:applications}, we discuss extensions and a few scenarios in which our language could be useful; we conclude in \Cref{sec:conclusion}.

%% file: extensional.tex
% !TEX root = ./main.tex

\section{Probabilistic Programs --- The Extensional Perspective}
\label{sec:extensional}
We briefly recap classical reasoning about probabilistic programs \'{a} la \citet{Kozen1985}, which is agnostic of any 
particular syntax for expressions or formulae---it takes an \emph{extensional} approach.

\subsection{The Probabilistic Guarded Command Language}
\label{sec:pgcl}
We consider the imperative probabilistic programming language $\pgcl$ featuring discrete probabilistic choices---branching on outcomes of coin flips---as well as standard control-flow instructions.

\subsubsection{Syntax}
Formally, a program $\cc$ in $\pgcl$ adheres to the grammar
\begin{align*}
\cc \qqlongrightarrow 	
&  \SKIP 	\tag{effectless program} \\
& \qmid \ASSIGN{\XX}{\TT}	\tag{assignment} \\
& \qmid \COMPOSE{\cc}{\cc}	\tag{sequential composition} \\
& \qmid \PCHOICE{\cc}{p}{\cc}	\tag{probabilistic choice} \\
& \qmid  \ITE{\BB}{\cc}{\cc}	\tag{conditional choice} \\
& \qmid  \WHILEDO{\BB}{\cc}~,	\tag{while loop} 
\end{align*}
where $x$ is taken from a countably infinite set of \emph{variables} $\Vars$,
$\TT$ is an \emph{arithmetic expression} over variables,
$p \in [0,1] \cap \Rats$ is a rational probability, and
$\BB$ is a Boolean expression (also called \emph{guard}) over variables.
For an overview of metavariables $\cc$, $\XX$, $\TT$, $\BB$, \dots, used throughout this paper, see \Cref{tab:metavariables} at the end of this section.

For the moment, we assume that both arithmetic and Boolean expressions are standard expressions without bothering to provide them with a concrete syntax. 
However, we will require them to adhere to a concrete syntax which we provide in Sections~\ref{sec:syntax:terms}~and~\ref{sec:syntax:bool}.
%alongside our syntax for expectations---the quantitative analogon to assertions---
%

\subsubsection{Program States}
\label{sec:program-states}
A program state $\pstate$ maps each variable in $\Vars$ to its value---a positive rational number in $\PosRats$.\footnote{To keep the presentation simple, we consider only \emph{unsigned} variables; we discuss this design choice and an extension to signed variables, which can also evaluate to negative rationals, in Section~\ref{sec:negatives}.}
To ensure that the set of program states is countable,\footnote{Working with probabilistic programs over a countable set of states avoids technical issues related to measurability.} 
we restrict ourselves to states in which at most finitely many variables---intuitively those that appear in a given program---are assigned non-zero values; every state can thus be understood as a finite mapping that only keeps track of assignments to non-zero values.
Formally, the set $\States$ of program states is
\begin{align*}
  \States \eeq \setcomp{ 
    \pstate\colon \Vars \to \PosRats ~
    }{
    \vphantom{\big(} \setcomp{ x \in \Vars }{ \pstate(x) \neq 0} \textnormal{ is finite}
%    \exists V \subseteq \Vars\colon
%        |V| < \infty
%        \wedge \forall x \in (\Vars\setminus V)\colon \pstate(x) = 0
  }~.
\end{align*}
We use metavariables $\pstatea$, $\pstateb$, \dots, for program states, see also \Cref{tab:metavariables}.
We denote by $\sem{\ee}{\pstate}{}$
the evaluation of (arithmetic or Boolean) expression $\ee$ in $\pstate$, i.e.,
the value obtained from evaluating $\ee$ after replacing every variable $x$ in $\ee$ 
by $\pstate(x)$.
We define the semantics of expressions more formally in \Cref{sec:semantics}.

\subsubsection{Forward Semantics}
\label{sec:forward-semantics}
One of the earliest ways to give semantics to a probabilistic program~$\cc$ is by means of \emph{forward-moving measure transformers} \cite{Kozen1979,Kozen1981}.
These transform an initial state~$\pstate$ into a probability distribution $\mu_\cc^{\pstate}$ over final states (\ie a measure on $\States$).
We consider Kozen's semantics the \emph{reference} forward semantics.
More operational semantics are provided in the form of probabilistic transition systems~\cite{GretzKM14,benni_diss}, where programs describe potentially infinite Markov chains whose state spaces comprise of program states, or trace semantics~\cite{DBLP:conf/esop/CousotM12,DBLP:conf/birthday/PierroW16,acta19}, where the traces are sequences of program states and each trace is assigned a certain probability.

In any of these semantics, the probabilistic choice $\PCHOICE{\cc_1}{p}{\cc_2}$ flips a coin with bias $p$ towards heads. 
If the coin yields heads, $\cc_1$ is executed (with probability $p$); otherwise, $\cc_2$.
Moreover, $\SKIP$ does nothing.
$\ASSIGN{x}{\TT}$ assigns the value of expression $\TT$ (evaluated in the current program state) to $x$.
The sequential composition $\COMPOSE{\cc_1}{\cc_2}$ first executes $\cc_1$ and then $\cc_2$.
The conditional choice $\ITE{\BB}{\cc_1}{\cc_2}$ executes $\cc_1$ if the guard $\BB$ is satisfied; otherwise, it executes $\cc_2$.
Finally, the loop $\WHILEDO{\BB}{\cc}$ keeps executing the loop body $\cc$ as long as $\BB$ evaluates to true.

\subsection{Weakest Preexpectations}

%Before we discuss requirements for a syntax that is both useful and allows relatively complete verification of $\pgcl$ programs, let us consider the underlying technique for reasoning about probabilistic programs in general, i.e., independent of any particular syntax.
%
Dually to the forward semantics, probabilistic programs can also be provided with semantics in the form of \emph{backward-moving random variable transformers}, originally due to~\citet{Kozen1983,Kozen1985}.
This paper is set within this dual view, which is a standard setting for probabilistic program verification.

\subsubsection{Expectations}
\label{sec:sem-exp}
\emph{Floyd-Hoare logic}~\cite{Hoare1969,floyd1967assigning} as well as the \emph{weakest precondition calculus} of \citet{Dijkstra1976} employ first-order predicates for reasoning about program correctness.
For probabilistic programs, \citet{Kozen1983,Kozen1985} was the first to generalize from predicates to measurable functions (or random variables). 
Later, \citet{McIverM05} coined the term \emph{expectation}---not to be confused with expected value---for such functions.
In reference to Dijkstra's weakest precondition calculus, their verification system is called the \emph{weakest preexpectation calculus}.

Formally, the set $\E$ of \emph{semantic expectations} is defined as
\begin{align*}
  \E \eeq \setcomp{ \ff }{ \ff\colon \States \to \PosRealsInf }~,
\end{align*}
\ie functions $\ff$ that associate a non-negative \emph{quantity} (or infinity) to each program state.
We use metavariables $\ff$, $\fg$, $\fh$ for semantic expectations.

Expectations form the \emph{assertion \enquote{language}} of the weakest preexpectation calculus.
However, we note that---so far---expectations are \emph{in no way defined syntactically}:
They are just the whole set of functions from $\States$ to $\PosRealsInf$.
It is hence borderline to speak of a \emph{language}.
The goal of this paper is to provide a \emph{syntactically defined subclass of\:~$\E$}---\ie an \emph{actual language}---such that formal reasoning about probabilistic programs can take place completely within this class.

We furthermore note that we work with more general expectations than \citet{McIverM05}, who only allow \emph{bounded} expectations, \ie expectations $\ff$ for which there is a bound $\alpha \in \PosReals$ such \mbox{that $\forall \pstate\colon \ff(\pstate) \leq \alpha$}.
In contrast to McIver and Morgan, our structure $(\E,\, \preceq)$ of \emph{unbounded} expectations forms a \emph{complete lattice} with least element $0$ and greatest element $\infty$, where $\preceq$ lifts the standard ordering  $\leq$
on the (extended) reals to expectations by pointwise application. That is, 
\begin{align*}
	\ff \ppreceq \fg \qquad \text{iff} \qquad \forall \pstate  \in \States \colon\quad \ff(\pstate) \lleq \fg(\pstate)~.	
\end{align*}%
Examples of (bounded) expectations include, for instance, \citet{Iverson1962} brackets $\iverson{\BB}$, which associate to a Boolean expression $\BB$ its indicator function:\footnote{We use $\lambda$-expressions to denote functions; function $\lambda x\mydot f$ applied to $a$ evaluates to $f$ in which $x$ is replaced by $a$.}%
\begin{align*}
	\label{eqn:sem_iverson}
	\iverson{\BB} \eeq \lambda \pstate \mydot 
	\begin{cases}
		1, &\text{if}~\sem{\BB}{\pstate}{} = \true \\
		0, &\text{if}~ \sem{\BB}{\pstate}{}= \false~.
	\end{cases}
\end{align*}%
Iverson brackets embed Boolean predicates into the set of expectations, rendering McIver and Morgan's calculus a conservative extension of Dijkstra's calculus.

Examples of \emph{unbounded} expectations are arithmetic expressions over variables, like%
\begin{align*}
	x + y \eeq \lambda \pstate\mydot \pstate(x) + \pstate(y)~,
\end{align*}%
where we point-wise lifted common operators on the reals, such as $+$, to operators on expectations.
Strictly speaking, \emph{McIver and Morgan's calculus cannot handle expectations like $x + y$} off-the-shelf.

We denote by $\ff\subst{x}{\TT}$ the \enquote{substitution} of variable $x$ by expression $a$ in expectation $\ff$, i.e.,
\begin{align*}
	\ff\subst{x}{\TT} \eeq \lambda \pstate \mydot  \ff\Bigl(\pstate \statesubst{x}{\sem{\TT}{\pstate}{}} \Bigr)~, 
	\qquad\text{where}\qquad 
	\pstate \statesubst{x}{r} \eeq \lambda y \mydot \begin{cases}
			r, & \textnormal{if } y = x, \\
			\pstate(y), & \textnormal{else}.
		\end{cases}
\end{align*}

\subsubsection{Backward Semantics: The Weakest Preexpectation Calculus}

\begin{figure}[t]
	\begin{center}
		\begin{adjustbox}{max width=.8\textwidth}
				\begin{tikzpicture}[node distance=4mm, decoration={snake,pre=lineto,pre length=.5mm,post=lineto,post length=1mm, amplitude=.2mm}]
					\draw[lightgray,use as bounding box,draw=none] (-3.75,-.875) grid (3.25, 3.25);
					\node (sigma) at (0, 3) {\Large$\boldsymbol{~\pstate}$};
					
					\node[gray,inner sep=0pt, outer sep=0pt] (branch) at (-0.225, 1.4) {$\bullet$};
					
					\node (tau1) at (-2, 0) {$\bullet$};
					\node (tau2) at (-.5, .25) {$\bullet$};
					\node (tau3) at (1, .125) {$\bullet$};
					\node[inner sep=0pt] (taudots) at (2.5, 0) {$\ddots$};
					
					\node[below of=tau1] {$\ff(\tau_1)$};
					\node[below of=tau2] {$\ff(\tau_2)$};
					\node[below of=tau3] {$\ff(\tau_3)$};

					\node (Exp) at (-3.1, -0.25) {\Large $\textbf{\textsf{Exp}}\boldsymbol{\Bigl[}$};
					\node at (3.1, -0.25) {\Large $\boldsymbol{\Bigr]}$};

					\node[gray] (C) at (0.5, 1.5) {$\cc$};
					
					\draw[gray,decorate,->,thick] (sigma) -- (tau1);
					
					\draw[gray,decorate,thick] (sigma) -- (branch);
					
					\draw[gray,decorate,->,thick] (branch) -- (tau2);
					\draw[gray,decorate,->,thick] (branch) -- (tau3);
					\draw[gray,decorate,thick] (sigma) -- (taudots);
					
					\draw[decorate,lightgray] (tau1) -- (tau2) -- (tau3) -- (taudots);
					
					\draw[gray] (sigma) edge[|->,bend right=45,above left,thick] node {\Large$\wp{\cc}{\ff}$} (Exp);
					%\draw (B) edge[->, very thick, right] node{\white{$\sfrac{1}{2}$}} (C);
				\end{tikzpicture}
				\qquad\quad
				\begin{tikzpicture}[node distance=4mm, decoration={snake,pre=lineto,pre length=.5mm,post=lineto,post length=1mm, amplitude=.2mm}]
					\draw[lightgray,use as bounding box,draw=none] (-3.75,-.875) grid (3.25, 3.25);
					\node (sigma) at (0, 3) {\Large$\boldsymbol{~\pstate'}$};
					
					\node[gray,inner sep=0pt, outer sep=0pt] (branch) at (.2, 1.25) {$\bullet$};
					
					\node (tau1) at (-2, 0) {$\bullet$};
					\node (tau2) at (-.75, .25) {$\bullet$};
					\node (tau3) at (1, -0.25) {$\bullet$};
					\node[inner sep=0pt] (taudots) at (2.5, 0) {$\ddots$};
					
					\node[below of=tau1] {$\ff(\tau_1')$};
					\node[below of=tau2] {$\ff(\tau_2')$};
					\node[below of=tau3] {$\ff(\tau_3')$};

					\node at (-3.1, -0.25) {\Large $\textbf{\textsf{Exp}}\boldsymbol{\Bigl[}$};
					\node at (3.1, -0.25) {\Large $\boldsymbol{\Bigr]}$};

					\node[gray] (C) at (1, 1) {$\cc$};
					
					\draw[gray,decorate,->,thick] (sigma) -- (tau1);
					
					\draw[gray,decorate,thick] (sigma) -- (branch);
					
					\draw[gray,decorate,->,thick] (branch) -- (tau2);
					\draw[gray,decorate,->,thick] (branch) -- (tau3);
					\draw[gray,decorate,thick] (sigma) -- (taudots);
					
					\draw[decorate,lightgray] (tau1) -- (tau2) -- (tau3) -- (taudots);
					
					\draw[gray] (sigma) edge[|->,bend right=45,above left,thick] node {\Large$\wp{\cc}{\ff}$} (Exp);
				\end{tikzpicture}%
		\end{adjustbox}%
	\end{center}%
	\caption{The weakest preexpectation $\wp{\cc}{\ff}$ maps every initial state $\pstate$ to the expected value of $\ff$, measured with respect to the final distribution over states reached after termination of program $\cc$ on input $\pstate$. 
	$\wpC{\cc}$ is backward-moving in the sense that it transforms an $\ff\colon \States \To \PosRealsInf$, evaluated in final states after termination of $\cc$, into $\wp{\cc}{\ff}\colon \States \To \PosRealsInf$, evaluated in initial states before execution of $\cc$.}%
	\label{fig:prog-execution}%
\end{figure}%
Suppose we are interested in the expected value of the quantity (expectation) $\ff$ after termination of $\cc$.
In analogy to Dijkstra, $\ff$ is called the \emph{postexpectation} and the sought-after expected value is called the \emph{weakest preexpectation} of $\cc$ with respect to \emph{postexpectation} 
$\ff$, denoted $\wp{\cc}{\ff}$~\cite{McIverM05}.
As the expected value of $\ff$ generally depends on the initial state $\pstate$ on which $\cc$ is executed, the \emph{weakest preexpectation} $\wp{\cc}{\ff}$ is itself also a map of type $\E$, mapping an initial program state $\pstate$ to the \emph{expected value} of $\ff$ (measured in the final states)
after successful termination of $\cc$ on $\pstate$, see \Cref{fig:prog-execution}. %
The weakest preexpectation calculus is a backward semantics in the sense that it transforms a postexpectation $\ff \in \E$, evaluated in final states after termination of $\cc$, into a preexpectation $\wp{\cc}{\ff} \in \E$, evaluated in initial states before execution of $\cc$.

Between forward-moving measure transformers and backward-moving expectation transformers, there exists the following duality established by Kozen:%
\begin{theorem}[Kozen Duality \textnormal{[\citeyear{Kozen1983,Kozen1985}]}]%
\label{thm:kozen-duality}%
	If $\mu_{\cc}^{\pstate}$ is the distribution over final states obtained by running $\cc$ on initial state $\pstate$, then for any postexpectation $\ff$,%
\begin{align*}
	\sum_{\tau \in \States} \mu_{\cc}^{\pstate}(\tau) \cdot \ff(\tau) \eeq \wp{\cc}{\ff}(\pstate)~.
\end{align*}
\end{theorem}%
\noindent{}%
In particular, if $\ff = \iverson{\BB}$, then $\wp{\cc}{\ff}(\pstate)$ is the \emph{probability} that
running $\cc$ on $\pstate$ terminates in a final state satisfying $\BB$---thus generalizing Dijkstra's weakest preconditions.

As with standard weakest preconditions, weakest preexpectations are not determined monolithically for the whole program $C$ as characterized above.
Rather, they are determined \emph{compositionally} using a backward-moving \emph{expectation transformer}
\begin{align*}
\wpsymbol \colon \pgcl \to (\E \to \E)
\end{align*}
which is defined recursively on the structure of $\cc$ according to the rules in Figure~\ref{table:wp}.%
\input{fig_wp}
Most of these rules are standard: 
$\wpC{\SKIP}$ is the identity as $\SKIP$ does not modify the program state.
For the assignment $\ASSIGN{x}{\TT}$, \mbox{$\wp{\ASSIGN{x}{\TT}}{\ff}$} substitutes in $\ff$ the assignment's left-hand side $x$ by its
right-hand side $\TT$.
For sequential composition, $\wp{\COMPOSE{\cc_1}{\cc_2}}{\ff}$ first determines the weakest preexpectation $\wp{\cc_2}{\ff}$
which is then fed into $\wpC{\cc_1}$ as a postexpectation.
For both the probabilistic choice $\PCHOICE{\cc_1}{p}{\cc_2}$ and the conditional choice $\ITE{\BB}{\cc_1}{\cc_2}$, the
weakest preexpectation with respect to $\ff$ yields a convex sum $p \cdot \wp{\cc_1}{\ff} + (1-p) \cdot \wp{\cc_2}{\ff}$.
In the former case, the weights are given by the probability~$p$.
In the latter case, they are determined by the guard $\BB$, \ie we have  $p = \iverson{\BB}$ and $1 - \iverson{\BB} = \iverson{\neg \BB}$.

The weakest preexpectation of a loop is given by the least fixed point of its unrollings, \ie 
\begin{align*}
	\wp{\WHILEDO{\BB}{\cc'}}{\ff} \eeq \lfp \fg \mydot \charwp{\BB}{\cc'}{\ff}(\fg)~,
\end{align*}%
where the \emph{characteristic function} 
$\charwp{\BB}{\cc'}{\ff}$ of $\WHILEDO{\BB}{\cc'}$ with respect to $\ff \in \E$ is defined as
\begin{align*}
\charwp{\BB}{\cc'}{\ff}\colon \quad \E \to \E, \quad \fg ~{}\mapsto{}~ \iverson{\neg \guard} \cdot \ff + \iverson{\BB} \cdot \wp{\cc'}{\fg}~.
\end{align*}
Since $(\E, \leq)$ is a complete lattice and $\charwp{\BB}{\cc'}{\ff}$ is monotone, fixed points 
exist due to the Knaster-Tarski fixed point theorem; we take the least fixed point because we reason about total correctness.

Throughout this paper, we exploit that $\charwp{\BB}{\cc'}{\ff}$ is, in fact, 
Scott-continuous (cf.~\cite{OlmedoKKM16}).
Kleene's theorem then allows us to approximate the least fixed point iteratively:
\begin{lemma}[\textnormal{\citet{kleene1952introduction}}]
	\label{lem:kleene_for_wp}
	We have 
	%Let $\cc = \WHILEDO{\BB}{\cc'}$ be a loop and let $\ff$ be an expectation. We have
	%
	\begin{align*}
	\wp{\WHILEDO{\BB}{\cc'}}{ \ff}
	\eeq
	\lfp \fg \mydot \charwp{\BB}{\cc'}{\ff}(\fg)
	\eeq
	\sup_{n \in \Nats} \charwpn{\BB}{\cc'}{\ff}{n}(0)~,
	\end{align*}
    where $0 = \lambda \pstate\mydot 0$ is the constant-zero expectation
    and $\charwpn{\BB}{\cc'}{\ff}{n}(\fg)$ denotes the $n$-fold application of 
    $\charwp{\BB}{\cc'}{\ff}$ to $\fg$.%, i.e.,\
%    \begin{align*}
%        \charwpn{\BB}{\cc'}{\ff}{0}(0) \eeq 0
%        \qquad \text{and} \qquad
%        \charwpn{\BB}{\cc'}{\ff}{n+1}(0)  \eeq \charwp{\BB}{\cc'}{\ff}(\charwpn{\BB}{\cc'}{\ff}{n}(0) )~.
%    \end{align*}
\end{lemma}
%
%%%%%\noindent{}%
%%%%%Before we take a closer look at examples motivating our search for a suitable syntax, 
%%%%%we note that the weakest preexpectation is \emph{sound} in the sense that the rules in Figure~\ref{table:wp} indeed compute the desired expected value~\cite{GretzKM14}. That is, 
%%%%%\begin{align*}
%%%%%    \wp{\WHILEDO{\BB}{\cc'}}{\ff}(\pstate)
%%%%%    \eeq 
%%%%%    \sum_{\pstate, \ldots, \pstate_n \in \mathrm{Paths}} \mathrm{Pr}(\pstate, \ldots, \pstate_n) \cdot \ff(\pstate_n)~,
%%%%%\end{align*}
%%%%%where $\mathrm{Paths}$ is the set of all execution paths $\pstate, \ldots, \pstate_n$ of program $\WHILEDO{\BB}{\cc'}$ that (i) start in initial state $\pstate$ and (ii) terminate in some state $\pstate_n$; 
%%%%%$\mathrm{Pr}(\pstate, \ldots, \pstate_n)$ denotes the probability that the given path $\pstate, \ldots, \pstate_n$ is executed.
%%%%%%
%
\begin{table}[t]
	\caption{Metavariables used throughout this paper.}%
	\label{tab:metavariables}%
	\renewcommand{\arraystretch}{1.25}%
	\begin{tabular}{l@{\qquad}l@{\qquad}l@{\qquad}l}
		\hline\hline
		\textbf{Entities}			& \textbf{Metavariables}  				& \textbf{Domain}		& \textbf{Defined in}			\\
		\hline
		Natural numbers 		& $n,\, i,\, j,\, k$ 				& $\Nats$				&						\\
		Positive rationals 		& $\RRa,\, \RRb,\, \RRc$ 				& $\PosRats$			&						\\
		Positive extended reals 	& $\alpha,\, \beta,\, \gamma$			& $\PosRealsInf$		&						\\
		Rational probabilities 	& $p,\, q$ 							& $[0,\, 1] \cap \Rats$	&						\\[.75em]
		Variables				& $\XX,\, \XY,\, \XZ,\, \VV,\, \VW,\, \VU, \gnum$	& \Vars				& \Cref{sec:pgcl}	 		\\
%%%%%		Logical variables		& $\VV,\, \VW,\, \VU$			& \LVars				& \Cref{sec:syntax:terms}		\\[.75em]
		%
		Arithmetic expressions 	& $\TTa,\, \TTb$ 				& \Terms				& \Cref{sec:syntax:terms}		\\
		Boolean expressions 	& $\BBa,\, \BBb,\, \BBc$	 			& \Bools				& \Cref{sec:syntax:bool}		\\[.75em]
		Syntactic expectations 	& $\FF,\, \FG,\, \FH$ 					& \SyntE				& \Cref{sec:syntax:exp}		\\
		Semantic expectations 	& $\ff,\, \fg,\, \fh$ 					& $\E$				& \Cref{sec:sem-exp}		\\[.75em]
		Programs 				& $C$ 							& \pgcl 				& \Cref{sec:pgcl}			\\
		Program states 		& $\sigma,\, \tau$ 					& $\States$ 			& \Cref{sec:program-states}	\\
%%%%%		Interpretations 			& $\interpret$ 				& 		 			& \Cref{sec:semantics}		\\
		\hline\hline
	\end{tabular}%
	\renewcommand{\arraystretch}{1}%
	%\vspace*{.5em}%
\end{table}%
%

%% file: fig_wp.tex
\begin{figure}[t]

	\renewcommand{\arraystretch}{1.5}

\begin{tabular}{@{\hspace{1em}}l@{\hspace{2em}}l}
	\hline\hline
	$\boldsymbol{\cc}$			& $\boldsymbol{\textbf{\textsf{wp}}\,\left \llbracket \cc\right\rrbracket  \left(\ff \right)}$ \\
	\hline
	$\SKIP$					& $\ff$ 																					\\
	$\ASSIGN{x}{\TT}$			& $\ff\subst{x}{\TT}$ \\
	$\COMPOSE{\cc_1}{\cc_2}$		& $\wp{\cc_1}{\vphantom{\big(}\wp{\cc_2}{\ff}}$ \\
	$\PCHOICE{\cc_1}{\pp}{\cc_2}$		& $\pp \cdot \wp{\cc_1}{\ff} + (1- \pp) \cdot \wp{\cc_2}{\ff}$ \\
	$\ITE{\guard}{\cc_1}{\cc_2}$		& $\iverson{\guard} \cdot \wp{\cc_1}{\ff} + \iverson{\neg \guard} \cdot \wp{\cc_2}{\ff}$ \\
	$\WHILEDO{\guard}{\cc'}$		& $\lfp \fg\mydot \iverson{\neg \guard} \cdot \ff + \iverson{\guard} \cdot \wp{\cc'}{\fg}$ \\
	\hline\hline
\end{tabular}
\caption{Rules defining the weakest preexpectation of program $\cc$ with respect to postexpectation $\ff$.}
\label{table:wp}
\end{figure}

%% file: towards.tex
% !TEX root = ./main.tex

\section{Towards an Expressive Language for Expectations}
\label{sec:towards-expressiveness}
As long as we take the extensional approach to program verification, \ie we admit all expectations in $\E$, 
reasoning about expected values of $\pgcl$ programs is \emph{complete}: 
For every program $\cc$ and postexpectation $\ff$, it is, in principle, possible to find an expectation $\wp{\cc}{\ff} \in \E$ which---by the above soundness property---coincides with the
expected value of $\ff$ after termination of $\cc$.

The main goal of this paper is to enable (relatively) complete verification of probabilistic programs by taking an \emph{intensional} approach. 
That is, we use the same verification technique described in Section~\ref{sec:extensional} (\ie the weakest preexpectation calculus) but%
\begin{center}
	fix a set $\SyntE$ of syntactic expectations $\FF$.
\end{center}%
We use metavariables $\FF$, $\FG$, $\FH$, \dots, for syntactic expectations, as opposed to \mbox{$\ff$, $\fg$, $\fh$, \dots,} for semantic expectations in $\E$, see also \Cref{tab:metavariables}. 
While $\FF$ itself is merely a syntactic entity to begin with, we denote by $\eval{\FF}$ the corresponding semantic expectation in $\E$.
Having a syntactic set of expectations at hand immediately raises the question of \emph{expressiveness}: 
\begin{center}
	For $\FF \in \mathbf{E}$, is the weakest preexpectation $\wp{\cc}{\eval{\FF}}$ again expressible in $\mathbf{E}$?
\end{center}%
\begin{definition}[Expressiveness of Expectations]
\label{def:expressiveness}
    The set $\SyntE$ of syntactic expectations is \emph{expressive}
    iff
    for all programs $\cc$ and all $\FF \in \SyntE$
    there exists a syntactic expectation $\FG \in \SyntE$,
    such that%
	\begin{align*}
		\wp{\cc}{\eval{\FF}} = \eval{\FG}~. \tag*{$\triangle$}
	\end{align*}%
\end{definition}
\noindent%
Notice that constructing \emph{some} expressive set of syntactic expectations is straightforward.
For example, the set $\SyntE = \{ 0 \}$, which consists of a single expectation $0$---interpreted as the constant expectation $\eval{0} = \lambda \pstate\mydot 0$---is expressive: $\wp{\cc}{\eval{0}} = \eval{0}$ holds for every $\cc$ by strictness of $\wpsymbol$.\footnote{$\wpsymbol$ being strict means that $\wp{\cc}{0} = 0$ for every $\cc$, see~\cite{benni_diss}.}

The main challenge is thus to find a syntactic set $\SyntE$ that (i) can be proven expressive \emph{and} (ii)~covers interesting properties---at the very least, it should cover all Boolean expressions $\BB$ (to reason about probabilities) and all arithmetic expressions $\TT$ (to reason about expected values).
%To guide our search towards such a syntax, this section proposes a few baseline requirements---expectations that we deem essential and that have to be included in any useful set of syntactic expectations.
%After that, we perform a case study on what forms of expectations are typically found in the probabilistic program verification literature. The results serve as a sanity check on the relevance of the concrete syntax presented in the next section.
%%
%\subsection{Essentials for Reasoning about Probabilities}
%%
%%
%\subsection{Essentials for Reasoning about Expected Values}
%%
%%
%\subsection{Case Study: How are Expectations used in the Literature?}
%%
%%
%\cmcommentinline{Remainder of this section:}
%%
%\begin{itemize}
%        \item at the very least, $\mathbf{E}$ should cover common guards $\BB$
%        \item this automatically gives us some interesting things, namely $\iverson{\true}$ and 
%              characteristic assertions; essentially all distribution transformer
%        \item for expected values, it seems reasonable to also 
%        \item what about it does not cover expected values beyond probabilities though?
%\end{itemize}

%

%% file: syntax.tex
% !TEX root = ./main.tex

\section{Syntactic Expectations}
\label{sec:syntax}

We now describe the syntax and semantics for a set $\SyntE$ of syntactic expectations which we will (in the subsequent sections) prove to be expressive and which can be used to express interesting properties such as, amongst others, the expected value of a variable $x$, the probability to terminate, the probability to terminate in a set described by a first-order arithmetic predicate $\BB$, etc.

\subsection{Syntax of Arithmetic Expressions}
\label{sec:syntax:terms}

We first describe a \emph{syntax for arithmetic expressions}, which form \emph{precisely the right-hand-sides of \underline{assi}g\underline{nments} that we allow in \pgcl programs}.
Naturally, the syntax of arithmetical expressions will reoccur in our syntax of expectations.
Formally, the set \Terms of arithmetic expressions is given by%
\begin{align*}
	\TT \qqlongrightarrow 	
	& \RR \iin \PosRats 	\tag{non-negative rationals} \\
	& \qmid  \XX \iin \Vars 	\tag{$\PosRats$-valued variables} \\
	%
%%%%%	& \qmid  \VV \iin \LVars 	\tag{$\PosRats$-valued \enquote{logical} variables} \\
	%
	& \qmid  \TT + \TT 		\tag{addition} \\
	& \qmid  \TT \cdot \TT~, 	\tag{multiplication} \ \\
	& \qmid \TT \monus \TT~, \tag{subtraction truncated at 0 (\enquote{monus})}
\end{align*}
where $\Vars$ is a \emph{countable} set of $\PosRats$-valued variables.
%%%%% and $\LVars$ is a countable set of $\PosRats$-valued logical variables. 
We use metavariables $\RRa,\, \RRb,\, \RRc$ for non-negative rationals, $\XX,\, \XY,\, \XZ,\, \VV,\, \VW,\, \VU$ for variables, and $\TTa,\, \TTb,\, \TTc$ for arithmetic expressions, see also \autoref{tab:metavariables}.%
%%%%Given a variable $\varSymb$ and arithmetic expressions $\TT, \TT'$, the arithmetic expression $\TT\subst{\varSymb}{\TT'}$ obtained from substituting
%%%%every occurrence of $\varSymb$ in $\TT$ by $\TT'$ is defined recursively as follows:
%%%%%
%%%%\begin{align*}
%%%%r\subst{\varSymb}{\TT'} &\ddefeq r \\
%%%%%
%%%%x\subst{\varSymb}{\TT'} & \ddefeq  \begin{cases}
%%%%\TT', & \text{if}~x = \varSymb  \\
%%%%%
%%%%x, & \text{otherwise}~
%%%%\end{cases} \\
%%%%%
%%%%v\subst{\varSymb}{\TT'} & \ddefeq  \begin{cases}
%%%%\TT', & \text{if}~v = \varSymb  \\
%%%%%
%%%%x, & \text{otherwise}~
%%%%\end{cases} \\
%%%%%
%%%%\left( \TTa + \TTb \right)\subst{\varSymb}{\TT'} & \ddefeq \TTa\subst{\varSymb}{\TT'} + \TTb\subst{\varSymb}{\TT'} \\
%%%%%
%%%%\left( \TTa \cdot \TTb \right)\subst{\varSymb}{\TT'} & \ddefeq \TTa\subst{\varSymb}{\TT'} \cdot \TTb\subst{\varSymb}{\TT'} \\
%%%%\end{align*}%

%%%%%The arithmetic expressions in which \emph{no logical variables occur} form precisely the set of arithmetic expressions which we allow as right-hand-sides of assignments in \pgcl programs.
%%%%%The logical variables play their role \enquote{merely} in making our language of expectations expressive.

\subsection{Syntax of Boolean Expressions}
\label{sec:syntax:bool}

We next describe a \emph{syntax for Boolean expressions} over $\Terms$, which form \emph{precisely the g\underline{uards} that we allow in \pgcl programs} (for conditional choices and while loops).
Again, the syntax of Boolean expressions will also naturally reoccur in our syntax of expectations.
Formally, the set \Bools of Boolean expressions is given by%
\begin{align*}
	\BB \qqlongrightarrow 	
	& \TT < \TT 			\tag{strict inequality of arithmetic expressions} \\
	& \qmid  \BB \wedge \BB 	\tag{conjunction} \\
	& \qmid  \neg \BB~. 		\tag{negation}
\end{align*}%
We use metavariables $\BBa,\, \BBb,\, \BBc$ for Boolean expressions, see also \autoref{tab:metavariables}.
%%%%%The Boolean expressions in which \emph{no logical variables occur} (within the arithmetic expressions) form precisely the set of Boolean expressions which we allow as guards of conditional choices and loops in \pgcl programs.

The following expressions are syntactic sugar with their standard interpretation and semantics:
\begin{align*}
    \false~,
    \qquad 
    \true~,
    \qquad
    \BBa \vee \BBb~,
    \qquad
    \BBa \longrightarrow \BBb~,
    \qquad 
    \TTa = \TTb~,
    \qqand
    \TTa \leq \TTb~.
\end{align*}%
%
%%
%\begin{align*}
%    \false \ddefeq & 
%    0 < 0 
%    \tag{false} \\
%    %
%    \true \ddefeq &
%    \neg \false 
%    \tag{true} \\
%    %
%    \BBa \vee \BBb \ddefeq & 
%    \neg (\neg \BBa \wedge \neg \BBb) 
%    \tag{disjunction} \\
%    %
%    \BBa \longrightarrow \BBb \ddefeq & 
%     (\neg \BBa \vee \BBb) 
%    \tag{implication} \\
%    %
%    \TTa = \TTb \ddefeq & 
%    \neg (\TTa < \TTb) \wedge \neg (\TTb < \TTa) 
%    \tag{equality of arithmetic expressions} \\
%	%
%    \TTa \leq \TTb \ddefeq & 
%	\TTa = \TTb \vee \TTa < \TTb 
%    \tag{inequality of arithmetic expressions}
%\end{align*}%
%%
%%%%Given a variable $\varSymb$, a Boolean expression $\BB$ and an arithmetic expression $\TT$, the Boolean expression obtained from
%%%%substituting every occurrence of $\varSymb$ in $\BB$ by $\TT$ is defined recursively as follows:%
%%%%%
%%%%\begin{align*}
%%%%   \left( \TTa < \TTb \right)\subst{\varSymb}{\TT} \ddefeq& \TTa\subst{\varSymb}{\TT} < \TTb \subst{\varSymb}{\TT} \\
%%%%   %
%%%%   \left( \BBa \wedge \BBb \right)\subst{\varSymb}{\TT} \ddefeq& \BBa\subst{\varSymb}{\TT} \wedge \BBb\subst{\varSymb}{\TT} \\
%%%%   %
%%%%   \left(\neg \BB \right)\subst{\varSymb}{\TT} \ddefeq& \neg \left(\BB \subst{\varSymb}{\TT} \right)
%%%%%\end{align*}%
%
%
%

\subsection{Syntax of Expectations}
\label{sec:syntax:exp}

We now describe the syntax of a set of \emph{expressive expectations} which can be used as both pre- and postexpectations for the verification of probabilistic programs.
Formally, the set \SyntE of \emph{syntactic expectations} is given by%
\begin{align*}
	\FF \qqlongrightarrow 	
	& \TT 		\tag{arithmetic expressions} \\
	& \qmid \iverson{\BB} \cdot \FF ~\mid~   \FF \cdot \iverson{\BB}  \tag{guarding} \\
	& \qmid  \FF + \FF 		\tag{addition} \\
	& \qmid \TT \cdot \FF ~\mid~   \FF \cdot \TT  \tag{scaling by arithmetic expressions} \\
	& \qmid \SupV{\XX} \FF  	\tag{supremum over $\XX$} \\
	& \qmid \InfV{\XX} \FF~.  	\tag{infimum over $\XX$}
	%
	%
	%& \qmid \iverson{b} 	\tag{iverson brackets, i.e.\ indicator functions} \\
	%
	%& \qmid  f \cdot f 	\tag{multiplication} \\
	%
	%& \qmid  t \cdot f 	\tag{scaling} \\
	%
	%& \qmid  \iverson{b} \cdot f 	\tag{guarded multiplication} \\
\end{align*}
As mentioned before, we use metavariables $\FF,\, \FG,\, \FH$ for syntactic expectations, see also \autoref{tab:metavariables}.
Let us go over the different possibilities of syntactic expectations according to the above grammar.

\paragraph{Arithmetic expressions}
These form the base case and it is immediate that they are needed for an expressive language.
Assume, for instance, that we want to know the \enquote{expected} (in fact: certain) value of variable $x$---itself an arithmetic expression by definition---after executing $\ASSIGN{x}{\TT}$.
Then this is given by $\wp{\ASSIGN{x}{\TT}}{x} = \TT$---again an arithmetic expression.
As $\TT$ could have been \emph{any} arithmetic expression, we at least need all arithmetic expressions in an \mbox{expressive expectation language}.

\paragraph{Guarding and addition.}
Both guarding---multiplication with a predicate---and addition are used for expressing weakest preexpectations of conditional choices and loops.
As we have, for instance,
\begin{align*}
	\wp{\ITE{\BB}{\cc_1}{\cc_2}}{f}  \eeq \iverson{\BB} \cdot \wp{\cc_1}{f} + \iverson{\neg\BB} \cdot \wp{\cc_1}{f}~,
\end{align*}%
it is evident that guarding and addition is convenient, if not necessary, for being expressive.

\paragraph{Scaling by arithmetic expressions.}
One could ask why we restrict to multiplications of arithmetic expressions and expectations and do not simply allow for multiplication of two arbitrary expectations $\FF \cdot \FG$.
We will defer this discussion to \Cref{sec:note-on-f-times-f}.
For now, it suffices to say that we can express all multiplications we need without running into trouble with quantifiers which would happen otherwise.

\paragraph{Suprema and infima.}
The supremum and infimum constructs $\SupV{\XX} \FF$ and $\InfV{\XX} \FF$ take over the role of the $\exists$ and $\forall$ quantifiers of first-order logic.
We use them to \emph{bind} variables $\XX$. 
The \mbox{$\Sup$ and $\Inf$}~quantifiers are necessary to make our expectation language expressive in the same was as, for instance, at least the $\exists$ quantifier is necessary to make first-order logic expressive for weakest preconditions of non-probabilistic programs.
\medskip 

As is standard, we additionally admit %in our syntax 
\emph{parentheses} for clarifying the order of precedence in syntactic expectations.
%To mostly avoid parentheses, however, 
To keep the amount of parentheses to a minimum,
we assume that $\cdot$ has precedence over $+$ and that the quantifiers~\Sup and \Inf have the \emph{least} precedence.

The set of \emph{free variables} $\FV{\FF} \subseteq \Vars$ % \cup \LVars$ of some syntactic expectation $\FF$ 
is the set of all variables that occur syntactically in $\FF$
and that are not in the scope of some \Sup or \Inf quantifier. 
%Since quantification over program variables is not allowed,
%$\FV{\FF}$ contains at least all program variables occurring in $\FF$. 
We write $\FF(\XX_1,\, \ldots,\, \XX_n)$ to indicate
that \emph{at most} the variables $\XX_1,\, \ldots,\, \XX_n$ occur freely in $\FF$.

Given a syntactic expectation $\FF$, a variable $\XX \in \FV{\FF}$, and an arithmetic expression $\TT$, we denote by $\FF\subst{\XX}{\TT}$ the \emph{syntactic replacement} of every occurrence of $\XX$ in $\FF$ by $\TT$.
Given a syntactic expectation of the form $\FF(\ldots, \XX_i, \ldots)$, we often write $\FF(\ldots,\TT,\ldots)$ instead
of the more cumbersome $\FF(\ldots, \XX_i, \ldots)\subst{\XX_i}{\TT}$.

%% file: semantics.tex
\subsection{Semantics of Expressions and Expectations} 
\label{sec:semantics}

The semantics of arithmetic and Boolean expressions is standard---see \Cref{fig:semantics:expressions}.
%We now give a Tarski-style semantics for syntactic expectations.
%
\input{fig-semantics-expressions}
For a program state~$\sigma$, 
%%%%%and interpretation $\interpret$,
we define
\begin{align*}
\sigma\statesubst{\XX}{\RR} \ddefeq \lambda \XY \mydot 
\begin{cases}
  \RR, &\text{if}~\XY=\XX \\
  \sigma(\XY), &\text{otherwise.}
\end{cases}
%%%%%\qquad 
%%%%%\text{and}
%%%%%\qquad
%%%%%\interpret\statesubst{\VV}{\RR} \ddefeq \lambda \VW \mydot
%%%%%\begin{cases}
%%%%%   \RR ,& \text{if}~\VW=\VV \\
%%%%%   %
%%%%%   \interpret(\VW), & \text{otherwise}~.
%%%%%\end{cases}
\end{align*}
The semantics $\sem{\FF}{\pstate}{\interpret}$ of an expectation \FF under state \pstate 
%%%%%and interpretation \interpret 
is an \emph{extended positive real} (\ie a positive real number or $\infty$) defined inductively as follows:
\begin{align*}
	\sem{\TT}{\pstate}{\interpret}  &\ddefeq \sem{\TT}{\pstate}{\interpret}~\footnotemark\\[.125em]
	\sem{\iverson{\BB} \cdot \FF}{\pstate}{\interpret} \ddefeq \sem{\FF \cdot \iverson{\BB} }{\pstate}{\interpret}&\ddefeq \begin{cases}
            \sem{\FF}{\pstate}{\interpret}, & \textnormal{ if }~ \sem{\BB}{\sigma}{\interpret} = \true \\
		0, & \textnormal{ else }
	\end{cases}  \\[.75em]
	\sem{\FF + \FG}{\pstate}{\interpret} &\ddefeq \sem{\FF}{\pstate}{\interpret} \pplus \sem{\FG}{\pstate}{\interpret}\\[.75em]
	\sem{\TT \cdot \FF}{\pstate}{\interpret} \ddefeq \sem{\FF \cdot \TT}{\pstate}{\interpret} &\ddefeq \sem{\TT}{\pstate}{\interpret} \ccdot \sem{\FF}{\pstate}{\interpret}
    \\[.75em]
	\sem{\SupV{\XX} \FF}{\pstate}{\interpret} & \ddefeq \sup~\setcomp{\sem{\FF}{\pstate\statesubst{\XX}{\RR}}{\interpret\statesubst{\VV}{\RR}}}{\RR \in \PosRats}\\[.125em]
	\sem{\InfV{\XX} \FF}{\pstate}{\interpret} &\ddefeq \inf\hspace{1.1ex}\setcomp{\sem{\FF}{\pstate\statesubst{\XX}{\RR}}{\interpret\statesubst{\VV}{\RR}}}{\RR \in \PosRats}
	%
	%
%	\sem{\Min{f}{g}}{\sigma}{\interpret} &\qeq \Min{\sem{f}{\sigma}{\interpret}}{\sem{g}{\sigma}{\interpret}} \\
%	%
%	\sem{\Max{f}{g}}{\sigma}{\interpret} &\qeq \Max{\sem{f}{\sigma}{\interpret}}{\sem{g}{\sigma}{\interpret}} \\[.75em]
	%
\end{align*}%
\footnotetext{Here, on the left-hand-side $\sem{\:\cdot\:}{\pstate}{\interpret}$ denotes the semantics of expectations, whereas on the right-hand-side $\sem{\:\cdot\:}{\pstate}{\interpret}$ denotes the semantics of arithmetic expressions.}%
We assume that $0 \cdot \infty = \infty \cdot 0 = 0$.
Most of the above are self-explanatory.
The most involved definitions are the ones for quantifiers.
The interpretation of the $\SupV{\XX} \FF$ quantification interprets~$\FF$ under all possible values of the bounded variable $\XX$ and then returns the supremum of all these values.
Analogously, $\InfV{\XX} \FF$ returns the infimum. 
Notice that---even though all variables evaluate to rationals---both the supremum and the infimum are taken over a set of reals.
Hence, an expectation $\FF$ involving $\Sup$ or $\Inf$ possibly evaluates to an \emph{irrational} number. 
For example, the expectation
%Consider, e.g., the expectation
\[
   \FF \eeq \Sup \XX \colon \iverson{x \cdot x < 2} \cdot x~,
\]
evaluates to $\sqrt{2} \not\in\PosRats$ under every state $\pstate$.

The supremum of $\emptyset$ is $0$.
Dually, the infimum of $\emptyset$ is $\infty$. 
The supremum of an unbounded set is~$\infty$.
We also note that our semantics can generate $\infty$ only by using a $\Sup$ quantifier.%
% and moreover that%
%\begin{align*}
%	\sem{a}{\pstate}{\interpret} = 0
%	\qand
%	\sem{f}{\pstate}{\interpret} = \infty
%	\qqimplies 
%	\sem{a \cdot f}{\pstate}{\interpret} = 0~.
%\end{align*}
%Put simply, our semantics gives $0 \cdot \infty = 0$ which is standard in measure theory and $\wpsymbol$ reasoning.

As a shorthand for turning syntactic expectations into semantic ones, we define%
\begin{align*}
	\eval{\FF} \ddefeq \lambda \pstate\mydot \sem{\FF}{\pstate}{}~.
\end{align*}%

\subsection{Equivalence and Ordering of Expectations}

For two expectations $\FF$ and $\FG$, we write $\FF = \FG$ only if they are \emph{syntactically equal}.
On the other hand, we say that two expectations $\FF$ and $\FG$ are \emph{semantically equivalent}, denoted $\FF \equiv \FG$, if their semantics under every state 
%and interpretation 
is equal, \ie 
\begin{align*}
	\FF \eequiv \FG \qqiff \eval{\FF} \eeq \eval{\FG}~.
\end{align*}%
Similarly to the partial order $\preceq$ on semantical expectations in $\E$, we define a (semantical) partial order~$\preceq$ on syntactic expectations in $\SyntE$ by%
\begin{align*}
	\FF \ppreceq \FG \qqiff \eval{\FF} \ppreceq \eval{\FG}~.
\end{align*}%

\subsection{A Note on Forbidding $\boldsymbol{\FF \cdot \FG}$ in our Syntax}
\label{sec:note-on-f-times-f}

Analogously to classical logic, a syntactic expectation $\FF$ is in \emph{prenex normal form}, if it is of the form
\begin{align*}
	f \eeq \Quant_1 \XX_1 \ldots \Quant_k \XX_k  \colon \FG~,
\end{align*}
where $\Quant_i \in \{\Sup, \Inf\}$ and where $\FG$ is quantifier-free.
Being able to transform any syntactic expectation into prenex normal form while preserving its semantics will be essential to our expressiveness proof.
In particular, we require that there is an algorithm that brings arbitrary syntactic expectations into prenex normal form, without inspecting their semantics.

The problem with allowing $\FF \cdot \FG$ arises in the context of the $0 \cdot \infty = 0$ phenomenon.
Suppose for the moment that we allow for $\FF \cdot \FG$ syntactically and define
\begin{align*}
   \sem{\FF \cdot \FG}{\sigma}{\interpret}
   \ddefeq
    \sem{\FF}{\sigma}{\interpret} \cdot \sem{\FG}{\sigma}{\interpret}
\end{align*}%
semantically, where $0 \cdot \infty = \infty \cdot 0 = 0$. 
Because of commutativity of multiplication, the above is an absolutely natural definition.
This also immediately gives us that $\sem{\FF \cdot \FG}{\sigma}{\interpret} = \sem{\FG \cdot \FF}{\sigma}{\interpret}$.

We now show that we encounter a problem when trying to transform expectations into prenex normal form. For that, consider the two expectations
\begin{align*}
	\FF \eeq \Inf \XX \colon \frac{1}{\XX +1}
	\qqand
	\quad  \FG \eeq \Sup \XY \colon \XY~.
\end{align*}
Notice that we slightly abuse notation since, strictly speaking, $\frac{1}{\XX +1}$ is not allowed by our syntax. We can however express it as $\Sup \XZ \colon \iverson{\XZ \cdot (\XX +1) = 1}\cdot \XZ $. 
Clearly, we have $\sem{\FF}{\sigma}{} = 0$ and $\sem{\FG}{\sigma}{} = \infty$
for all $\sigma$, \ie both $\FF$ and $\FG$ are constant expectations.

Let us now consider the product of $\FF$ and $\FG$.
For all $\sigma$, its semantics is given by
\begin{align*}
	\sem{\FF \cdot \FG}{\sigma}{\interpret} \eeq \sem{\FF}{\sigma}{\interpret} \cdot \sem{\FG}{\sigma}{\interpret} \eeq 0 \cdot \infty \eeq 0 \eeq \infty \cdot 0 \eeq \sem{\FG}{\sigma}{\interpret} \cdot \sem{\FF}{\sigma}{\interpret} \eeq \sem{\FG \cdot \FF}{\sigma}{\interpret}~.
\end{align*}%
Now consider the following:
\begin{align*}
	\blue{\sem{\FF \cdot \FG}{\sigma}{\interpret}} 
	\eeq & \sem{\left(\Inf \XX \colon \frac{1}{\XX +1} \right) \cdot \bigl(\Sup \XY \colon \XY\bigr)}{\sigma}{\interpret}\\
	\eeq &\sem{\Inf \XX \colon \Sup \XY \colon \frac{1}{\XX +1} \cdot \XY}
    {\sigma}{\interpret} \tag{by prenexing}\\
    \eeq & \inf \setcomp{
        \sup \setcomp{\frac{1}{\RRa +1} \cdot \RRb}
        {\RRb \in \PosRats}
    }
    {\RRa \in \PosRats} \\
    \eeq & \inf \setcomp{
    	\infty
    }
    {\RRa \in \PosRats} \\
    \blue{\eeq} & \blue{\infty} \\
    \blue{~{}\neq{}~}& \blue{0} \\
    \eeq& \sup \setcomp{
    	0
    }
    {\RRb \in \PosRats} \\
    \eeq& \sup \setcomp{
    	\inf \setcomp{\frac{1}{\RRa +1} \cdot \RRb}
    	{\RRa \in \PosRats}
    }
    {\RRb \in \PosRats} \\
    \eeq& \sup \setcomp{
    	\inf \setcomp{\RRb \cdot \frac{1}{\RRa +1} }
    	{\RRa \in \PosRats}
    }
    {\RRb \in \PosRats} \tag{by commutativity of $\cdot$ in $\PosRealsInf$}\\
    \eeq & \sem{\Sup\XY \colon \Inf \XX \colon \XY \cdot  \frac{1}{\XX +1}}
    {\sigma}{\interpret} \\
    \eeq & \sem{\bigl(\Sup \XY \colon \XY\bigr) \cdot \left(\Inf \XX \colon \frac{1}{\XX +1} \right) }{\sigma}{\interpret} \tag{by un-prenexing}\\
    \blue{\eeq} & \blue{\sem{\FG \cdot \FF}{\sigma}{\interpret}}
\end{align*}
We see that $\Sup\XY \colon \Inf \XX \colon \frac{1}{\XX +1} \cdot \XY$ is a sound prenex normal form of $\FG \cdot \FF$ whereas $\Inf \XX \colon \Sup \XY \colon \frac{1}{\XX +1} \cdot \XY$ apparently is not a sound prenex normal form of $\FF \cdot \FG$.
A fact that seems even more off-putting is that---even though $f \equiv 0$---the above argument would not have worked for $f = 0$.

To summarize, we deem the above considerations enough grounds to forbid $f \cdot g$ altogether, in particular since the rescaling $a \cdot f$ suffices in order for our syntactic expectations to be expressive. 
We also note that we will later provide a syntactic, but much more complicated, way to write down arbitrary products between syntactic expectations, see~\Cref{thm:prod_exp}.

%% file: fig-semantics-expressions.tex
\begin{table}[t]
	\renewcommand{\arraystretch}{1.5}

    \begin{tabular}{@{\hspace{1em}}l@{\hspace{2em}}l|@{\hspace{1em}}l@{\hspace{2em}}l}
	\hline\hline
    $\boldsymbol{\TT}$ & $\boldsymbol{\sem{\TT}{\pstate}{}}$ & $\boldsymbol{\BB}$ & $\boldsymbol{\sem{\BB}{\pstate}{} = \true}\quad\textbf{iff}$ \\
	\hline
    $\RR\quad(\in \PosRats)$ & $\RR$ & $\TTa < \TTb$ &  $\sem{\TTa}{\pstate}{} < \sem{\TTb}{\pstate}{}$ \\
    $\XX\quad(\in \Vars)$ & $\pstate(\XX)$ & $\BBb \wedge \BBc$ & $\sem{\BBb}{\pstate}{\interpret} = \true = \sem{\BBc}{\pstate}{\interpret}$   \\
    $\TTb + \TTc$ & $\sem{\TTb}{\pstate}{} + \sem{\TTc}{\pstate}{}$ & $\neg \BBb$ & $\sem{\BBb}{\pstate}{} = \false$ \\
    $\TTb \cdot \TTc$ & $\sem{\TTb}{\pstate}{} \cdot \sem{\TTc}{\pstate}{}$ & \\
    $\TTb \monus \TTc$ & $\begin{cases}
    \sem{\TTb}{\pstate}{} - \sem{\TTc}{\pstate}{}, & \text{if}~\sem{\TTb}{\pstate}{} \geq \sem{\TTc}{\pstate}{} \\
    0\,, &\text{else}
    \end{cases}
    $& \\
	\hline\hline
\end{tabular}
\caption{The semantics of arithmetic expressions $\TT$ and Boolean expressions $\BB$.}
\label{fig:semantics:expressions}
\end{table}

%% file: proof_straight_line.tex
\section{Expressiveness for Loop-free Programs}
\label{sec:expressiveness:loop-free}

\noindent{}%
Before we deal with loops, we now show that our set $\SyntE$ of syntactic expectations is
\emph{expressive for all \underline{loo}p\underline{-}f\underline{ree} $\pgcl$ programs}.
%, which we collect in the set $\pgcl_{\textrm{loop-free}}$.
Proving expressiveness for loops is \emph{way more involved} and will be addressed 
\mbox{separately in the remaining sections}.
\begin{lemma}%
\label{thm:expressive-loop-free}%
	$\SyntE$ is expressive (see \textnormal{\Cref{def:expressiveness}}) for all loop-free $\pgcl$ programs $\cc$, \ie 
	for all $\FF \in \SyntE$ there exists a syntactic expectation $\FG \in \SyntE$, such that
	\begin{align*}
%		\forall \textnormal{loop-free }\cc~ 
%		\forall \FF \in \SyntE~ 
%		\exists \FG \in \SyntE
%		\colon \qquad
		\wp{\cc}{\eval{\FF}} \eeq \eval{\FG}~.
	\end{align*}%
\end{lemma}%
\noindent{}%
For proving this expressiveness lemma (and also for the case of loops), we need the following technical lemma about substitution of variables by values in our semantics:
\begin{lemma}%
\label{lem:substitution}%
	For all $\pstate$, $\FF$, and $\TT$,
	\begin{align*}
	   \sem{\FF\subst\XX{\TTa}}{\sigma}{\interpret}
	   \eeq
	   \sem{\FF}{\sigma \statesubst\XX{\sem{\TTa}{\sigma}{\interpret}}}{\interpret}
	   \qquad\textnormal{or equivalently}\qquad
	   \eval{\FF\subst\XX{\TTa}}
	   \eeq
	   \eval{\FF}\subst\XX{\TTa}
	   %
%%%%%	   \qquad
%%%%%	   \textnormal{and}
%%%%%	   \qquad
%%%%%	   \sem{\FF\subst{\VV}{\TTb}}{\sigma}{\interpret}
%%%%%	   \eeq
%%%%%	   \sem{\FF}{\sigma}{
%%%%%	   	\interpret\statesubst{\VV}{\sem{\TTb}{\sigma}{\interpret}}}~.
	\end{align*}
\end{lemma}%
\begin{proof}
	By induction on the structure of $\FF$.
\end{proof}%
\noindent{}%
Intuitively, Lemma~\ref{lem:substitution} states that syntactically replacing variable $\XX$ by an arithmetical expression $a$ in expectation $\FF$ amounts to interpreting $\FF$ in states where the variable~$\XX$ has been substituted by the evaluation of $\TTa$ under that state.
%%%%%The analogous explanation goes for syntactic substitution of variables by an arbitrary arithmetic expression $\TTb$.

\begin{proof}[Proof of \textnormal{\Cref{thm:expressive-loop-free}}]
	Let $\FF \in \SyntE$ be arbitrary.
	The proof goes by induction on the structure of loop-free programs $\cc$.
	It is somewhat standard, but it demonstrates nicely that our syntactic constructs are actually needed, so we provide it here.
	We start with the atomic programs:
	\paragraph{The effectless program $\SKIP$}
	We have $\wp{\SKIP}{\eval{\FF}} = \eval{\FF}$ and $\FF \in \SyntE$ by assumption.
	
	\paragraph{The assignment $\ASSIGN{x}{\TT}$} 
	We have 
	\begin{align*}
		\wp{\ASSIGN{x}{\TT}}{\eval{\FF}} 
		& \eeq \eval{\FF}\subst{x}{\TT} 		\\
		& \eeq \eval{\FF\subst{x}{\TT}} 		\tag{by \Cref{lem:substitution}}
	\end{align*}%
	and $\FF\subst{x}{\TT} \in \SyntE$ since $\FF\subst{x}{\TT}$ is obtained from $\FF$ by a syntactical replacement.

	\paragraph{Induction Hypothesis}
	For arbitrary loop-free $\cc_1$ and $\cc_2$, there exist syntactic expectations $\FG_1,\FG_2 \in \SyntE$, such that%
	\begin{align*}
		\wp{\cc_1}{\eval{\FF}} \eeq \eval{\FG_1}
		\qqand
		\wp{\cc_2}{\eval{\FF}} \eeq \eval{\FG_2}~.
    \end{align*}
    We then proceed with the compound loop-free programs:

    \paragraph{The probabilistic choice $\PCHOICE{\cc_1}{p}{\cc_2}$} 
    We have%
    \begin{align*}
       &\wp{\PCHOICE{\cc_1}{p}{\cc_2}}{\eval{\FF}} \\
       & \eeq  p \cdot \wp{\cc_1}{\eval{\FF}} + (1-p) \cdot \wp{\cc_2}{\eval{\FF}}
       \tag{by definition of $\wpsymbol$} \\
       & \eeq p \cdot \eval{\FG_1} + (1-p) \cdot \eval{\FG_2}
       \tag{by I.H.\ on $\cc_1$ and $\cc_2$} \\
       & \eeq \eval{p \cdot \FG_1 + (1-p) \cdot \FG_2}
       \tag{pointwise addition and multiplication}
    \end{align*}
    and $p \cdot \FG_1 + (1-p) \cdot \FG_2 \in \SyntE$, see \Cref{sec:syntax:exp}.

    \paragraph{The conditional choice $\ITE{\BB}{\cc_1}{\cc_2}$} 
    We have%
    \begin{align*}
       & \wp{\ITE{\BB}{\cc_1}{\cc_2}}{\eval{\FF}} \\
       & \eeq  \iverson{\BB} \cdot \wp{\cc_1}{\eval{\FF}}  + \iverson{\neg\BB}\cdot \wp{\cc_2}{\eval{\FF}} 
       \tag{by definition of $\wpsymbol$} \\
       & \eeq \iverson{\BB} \cdot \eval{\FG_1} + \iverson{\neg\BB} \cdot \eval{\FG_2}
       \tag{by I.H.\ on $\cc_1$ and $\cc_2$} \\
       & \eeq \eval{\iverson{\BB} \cdot  \FG_1 + \iverson{\neg\BB} \cdot  \FG_2}
       \tag{pointwise addition and multiplication}
    \end{align*}
    and $\iverson{\BB} \cdot  \FG_1 + \iverson{\neg\BB} \cdot  \FG_2 \in \SyntE$, see \Cref{sec:syntax:exp}.
    \medskip
    
    Hence, $\SyntE$ is expressive for loop-free programs.
\end{proof}

%% file: proof_outline.tex
% !TEX root = ./main.tex

\section{Expressiveness for Loopy Programs --- Overview}
\label{sec:proof-outline}

\newcommand{\syntsum}[3]{\sfsymbol{Sum}_{#1,#2}\left( #3 \right)}
\newcommand{\cciter}{\cc_{\textrm{iter}}}

Before we get to the proof itself, we outline the main challenges---and the steps we took to address them---of 
proving expressiveness of our syntactic expectations $\SyntE$ for $\pgcl$ programs including loops; 
the technical details of the involved encodings and auxiliary results are considered throughout Sections~\ref{sec:embedding}~--~\ref{sec:expressiveness}.
This section is intended to support navigation through the individual components of the expressiveness proof; as such, we provide various references to follow-up sections.

\subsection{Setup}
%%%%%Our overall goal is to prove that $\SyntE$ is expressive, i.e.\ for all $\pgcl$ programs $C$ and all syntactic postexpectations $\FF \in \SyntE$, there exists a syntactic weakest preexpectation $\FG \in \SyntE$ such that
%%%%%\[ \wp{\cc}{ \eval{\FF} } \eeq \eval{\FG}~. \]

As in the loop-free case considered in Section~\ref{sec:expressiveness:loop-free}, we prove expressiveness of $\SyntE$ for all $\pgcl$ programs (including loopy ones) by induction
on the program structure; all cases except loops are completely analogous to the proof of Lemma~\ref{thm:expressive-loop-free}.
Our remaining proof obligation thus boils down to proving that, for every loop $\cc = \WHILEDO{\BB}{\cc'}$,
\begin{align*}
  \forall\, \FF \in \SyntE ~ \exists\, \FG \in \SyntE\colon \quad 
  \wp{\,\WHILEDO{\BB}{\cc'}}{ \eval{\FF} } \eeq \eval{\FG}~, \tag{$\dagger$}
\end{align*}
where we already know by the I.H. that the same property holds for the loop body $\cc'$, i.e.,%
\begin{align}
	\forall\, \FF' \in \SyntE ~ \exists\, \FG' \in \SyntE\colon \quad 
	\wp{\cc'}{ \eval{\FF'} } \eeq \eval{\FG'}~.
	\label{eq:loops:ih}
\end{align}%

\paragraph{Remark (A Simplification for this Overview)}\label{sec:overview:simplification}
Just for this overview section, we assume that the set $\Vars$ of all variables is \emph{finite} instead of countable.
This is a convenient simplification to avoid a few purely technical details such that we can focus on the actual ideas of the proof.
We do \emph{not} make this assumption in follow-up sections.
Rather, our construction will ensure that only the finite set of ``relevant'' variables---those that appear in 
the program or the postcondition under consideration---is taken into account.\hfill$\triangle$

\subsection{Basic Idea: Exploiting the Kozen Duality}
\label{sec:kozen_duality}
%
%Exploiting Kozen duality, 
We first move to an alternative characterization of the weakest preexpectation of loops 
whose components are simpler to capture with syntactic expectations.
In particular, we will be able to apply our induction hypothesis (\ref{eq:loops:ih}) to some of these components.

Recall the Kozen duality between forward moving measure transformers and backward moving expectation transformers (see \Cref{thm:kozen-duality} and Figure~\ref{fig:prog-execution} in Section~\ref{sec:extensional}):%
\begin{align*}
    \wp{\cc}{\ff}
    \eeq 
    \lambda \pstate_0\mydot 
    \sum_{\tau \in \States} \ff(\tau) \cdot \mu_{\cc}^{\pstate_0}(\tau)~,
\end{align*}%
where $\mu_{\cc}^{\pstate_0}$ is the probability distribution over final states obtained by running $\cc$ on initial state~$\pstate_0$. 
Adapting the above equality to our concrete case in which $\cc$ is a loop and $\ff = \eval{\FF}$, we obtain
\begin{align*}
    \wp{\,\WHILEDO{\BB}{\cc'}}{\eval{\FF}}
    \eeq 
    \lambda \pstate_0\mydot 
    \sum_{\tau \in \States} \, \eval{\iverson{\neg \BB} \cdot \FF}(\tau) \cdot \mu_{\WHILEDO{\BB}{\cc'}}^{\pstate_0}(\tau)~,
\end{align*}
where we strengthened the postexpectation $\FF$ to $\iverson{\neg\BB} \cdot \FF$ to account for the fact that the loop guard $\BB$ is violated in every final state, see~\cite[Corollary 4.6, p.~85]{benni_diss}.
The main idea is---instead of viewing the whole distribution $\mu_{\WHILEDO{\BB}{\cc'}}^{\sigma_0}$ in a single \enquote{big step}---to take a more operational \enquote{small-step} view: 
we consider the intermediate states reached after each guarded loop iteration, which corresponds to executing the program%
\begin{align*}
	\cciter \eeq \ITE{\BB}{\cc'}{\SKIP}~.
\end{align*}%
We then sum over all terminating \emph{execution paths}---finite sequences of states $\pstate_0, \ldots \pstate_{k-1}$ with initial state $\pstate_0$ and final state $\pstate_{k-1} = \tau$---instead of a single final state $\tau$. 
The probability of an execution path is then given by the product of the probability  $\mu_{\cciter}^{\pstate_{i}}(\pstate_{i+1})$ of each intermediate step, i.e., the probability of reaching the state $\pstate_{i+1}$ from the previous state $\pstate_{i}$:
\begin{align}
    \wp{\,\WHILEDO{\BB}{\cc'}}{\eval{\FF}} \eeq 
    \lambda \pstate_0\mydot 
    \sup_{k \in \Nats}
    \sum_{\sigma_0, \ldots, \sigma_{k-1} \in \States} 
    \eval{\iverson{\neg \BB} \cdot \FF}(\pstate_{k-1}) \cdot
    \prod_{i = 0}^{k-2}
    \mu_{\cciter}^{\pstate_i}(\pstate_{i+1})~.
    \label{eq:loops:small-step}
\end{align}
Notice that the above sum (without the $\sup$) considers all execution paths of a fixed length $k$; we take the supremum over all natural numbers
$k$ to account for all terminating execution paths.

%
%Notice that the first state of every execution path is the initial state $\pstate_0$ which is supplied as an argument of the expectation.
%
Next, we aim to apply the induction hypothesis (\ref{eq:loops:ih}) to the probability $\mu_{\cciter}^{\pstate_i}(\pstate_{i+1})$ of each step
such that we can write it as a syntactic expectation.
To this end, we need to characterize $\mu_{\cciter}^{\pstate_i}(\pstate_{i+1})$ in terms of weakest preexpectations.
We employ a syntactic expectation $\statepred{\pstate}{}$---called the \emph{characteristic assertion}~\cite{winskel} of state $\pstate$---that captures the values assigned to variables by state $\pstate$:\footnote{Recall from our remark on simplification that $\Vars$ is finite.}
\begin{align*}
  \statepred{\pstate}{} \eeq \iverson{\bigwedge_{x \in \Vars} x = \pstate(x)}~.
\end{align*}
By Kozen duality (\Cref{thm:kozen-duality}), 
the probability of reaching state $\pstate_{i+1}$ from $\pstate_{i}$ in one guarded loop iteration~$\cc_{\textrm{iter}}$ is then given by
\begin{align*}
  \mu_{\cciter}^{\pstate_i}(\pstate_{i+1})
  \eeq  
  \wp{\cciter}{\eval{\statepred{\pstate_{i+1}}{}}}(\pstate_i)~.
  %\eeq
  %\wp{\ITE{\BB}{\cc'}{\SKIP}}{\eval{\statepred{\pstate_{i+1}}{}}}(\pstate_i)~.
\end{align*}
By the same reasoning as for conditional choices in Lemma~\ref{thm:expressive-loop-free}
and the induction hypothesis (\ref{eq:loops:ih}), 
there exists a syntactic expectation $\FG_{\cciter}^{\pstate_{i+1}} \in \SyntE$ such that
\begin{align*}
  \mu_{\cciter}^{\pstate_i}(\pstate_{i+1}) 
  \eeq
  \wp{\cciter}{\eval{\statepred{\pstate_{i+1}}{}}}(\pstate_i)
  \eeq
   \eval{\FG_{\cciter}^{\pstate_{i+1}}}(\pstate_{i})~.
\end{align*}
Plugging the above equality into our \enquote{small-step} characterization of loops (\ref{eq:loops:small-step}) 
then yields the following characterization of $\eval{g}$ in ($\dagger$):
\begin{align}
    \wp{\,\WHILEDO{\BB}{\cc'}}{\eval{\FF}} \eeq 
    \lambda \pstate_0\mydot 
    ~
    \underbrace{\sup_{k \in \Nats}
    ~
    \underbrace{\sum_{\sigma_0, \ldots, \sigma_{k-1} \in \States}
    ~
    \underbrace{
    \Bigl\llbracket\,
    	\underbrace{
		\iverson{\neg \BB} \cdot \FF\!
    	}_{\in \SyntE}
     \;\Bigr\rrbracket(\pstate_{k-1})
        ~\cdot~ \underbrace{\prod_{i = 0}^{k-2}
    ~
    \Bigl\llbracket~
    \underbrace{
            \FG_{\cciter}^{\pstate_{i+1}}}_{{}\in \SyntE}
    ~\Bigr\rrbracket(\pstate_{i})
    }_{\mathclap{\textrm{non-constant product expressible in $\SyntE$?}}}
    }_{\textrm{simple product expressible in $\SyntE$?}}
    }_{\textrm{non-constant sum over paths of length $k$ expressible in $\SyntE$? }}
    }_{\Sup k\colon \ldots\, {}\in \SyntE}
    \label{eq:loops:main}
\end{align}
A formal proof of the above characterization is provided alongside \Cref{thm:wp_loop_as_sum}.
\subsection{Encoding Loops as Syntactic Expectations}
Let us now revisit the individual components of the expectation (\ref{eq:loops:main}) above 
and discuss how to encode them as
syntactic expectations in $\SyntE$, moving through the braces from bottom to top:

\subsubsection{The supremum $\sup_{k \in \Nats}$}
\label{sec:outline:supremum}
The supremum ensures that terminating execution paths \emph{of arbitrary length} are accounted for; it is supported in $\SyntE$
by the $\Sup$ quantifier.
If we already know a syntactic expectation~$\FG_{\textrm{sum}}(k) \in \SyntE$ for the entire sum that follows, we hence obtain
an encoding 
%$\FG \in \SyntE$ 
of the whole expectation, namely%
\begin{align*}
	\Sup k\colon \FG_{\textrm{sum}}(k,\ldots) ~{}\in{}~\SyntE~.
\end{align*}%

\subsubsection{The non-constant sum $\sum_{\sigma_0, \ldots, \sigma_{k-1} \in \States}$}
\label{sec:outline:sum}
This sum \emph{cannot} directly be written as a syntactic expectation: % for two reasons:
First, it sums over \emph{execution paths} whereas all variables and constants in syntactic expectations are evaluated to rational numbers.
Second, its number of summands depends on the length $k$ of execution paths whereas $\SyntE$ only supports sums with a constant number of summands.

To deal with the first issue, there is a standard solution in proofs of expressiveness~(cf.~\cite{Loeckx87,winskel,expressiveness_sl,expressiveness_sl_conference}): We employ \emph{G\"odelization} to encode both program states and finite sequences of program states as natural numbers in syntactic expectations. 
The details are found in \Cref{sec:embedding}. In particular:
\begin{itemize}
  \item We show that $\SyntE$ subsumes first-order arithmetic over the natural numbers.
  \item We adapt the approach of \citet{goedel_beta} to encode sequences of both natural numbers and non-negative rationals as G\"odel numbers in our language $\SyntE$.
  \item We define a predicate (in $\SyntE$) \stateseq{}{u}{\VV} that
        is satisfied iff $u$ is the G\"odel number of a sequence of states of length $\VV-1$.
\end{itemize}
\noindent
To deal with the second issue (the sum having a variable number of summands), we also rely on the ability to encode sequences as G\"odel numbers in $\SyntE$---the details are found in Section~\ref{sec:sums_prod_via_goedel}. Roughly speaking, we encode the sum as follows:
\begin{itemize}
  \item We define a syntactic expectation $\FH(\vsum,\ldots)$ that serves as a map from $\vsum$ to individual summands, i.e.,
          $\FH\subst{\vsum}{i}$ yields the $i$-th summand.
  \item We construct a syntactic expectation $\gsum{\FH}{\VV}$ for partial sums, summing up the first $\VV$ summands defined by the syntactic expectation $\FH$---see \Cref{thm:sum_exp} for details.
%More formally, we will show in Theorem~\ref{thm:sum_exp} that
%\[ 
%    \forall \pstate \in \States ~
%    \forall \FH \in \SyntE\colon \qquad
%    \sem{\gsum{\FH}{\VV}}{\sigma}{\interpret} \eeq \sum_{i=0}^{\pstate(\VV)} \sem{\FH\subst{\vsum}{i}}{\sigma}{}~. 
%\]
\end{itemize}
\subsubsection{The product $\eval{\iverson{\neg \BB} \cdot \FF} \cdot \ldots$}
\label{sec:outline:binary-product}
This product is not directly expressible in $\SyntE$ as arbitrary products between syntactic expectations are not allowed.
They are, however, expressible in our language. 
We define a product operation $\FH_1 \exprod \FH_2$ and prove its correctness in 
\Cref{cor:unrestricted_product}.

\subsubsection{The non-constant product $\prod_{i = 0}^{k-2} \eval{\FG_{\cciter}^{\pstate_{i+1}}}(\pstate_{i})$}
\label{sec:outline:product}
This product consists of $k-1$ factors; its encoding requires a similar approach as for non-constant sums.
That is, we define a syntactic expectation $\gproduct{\FH}{\VV}$ that multiplies the first $\VV$
\emph{factors} defined by the syntactic expectation $\FH(\vprod)$. Details are provided in Theorem~\ref{thm:prod_exp}.
%More formally, we will show in
%Theorem~\ref{thm:prod_exp}, \Cref{sec:prod_via_goedel}, that
%\begin{align*}
%    \forall \pstate \in \States ~
%    \forall \FH \in \SyntE\colon \qquad
%    \sem{\gproduct{\FH}{\VV}}{\sigma}{}
%    \eeq 
%    \prod_{j=0}^{\pstate(\VV)} \sem{\FH\subst{\vprod}{j}}{\pstate}{}~.
%\end{align*}
%
%
\subsubsection{The expectations $\eval{\iverson{\neg \BB} \cdot \FF}$ and $\eval{\FG_{\cciter}^{\pstate_{i+1}}}$}
Both are syntactic expectations by construction.

\subsection{The Expressiveness Proof}
It remains to glue together the constructions for the individual components of the expectation (\ref{eq:loops:main}), 
which characterizes the weakest preexpectation of loops.
We present the full construction, a proof of its correctness, and an example of the resulting syntactic expectation in \Cref{sec:expressiveness}.
%

%% file: first_order_arithmetic.tex
% !TEX root = ./main.tex

\section{G\"odelization for Syntactic Expectations}
\label{sec:embedding}

\input{figure_translation}

We embed the (standard model of) first-order arithmetic over both the rational and the natural numbers
in our language $\SyntE$---thereby addressing the first issue raised in \Cref{sec:outline:supremum}. % of the proof outline.
Consequently, $\SyntE$ conservatively extends the standard assertion language of Floyd-Hoare logic (cf.\ \cite{winskel, loeckx1984foundations, Cook1978SoundnessAC}), enabling us to encode finite sequences of both rationals and naturals in $\SyntE$
by means of G\"odelization~\cite{goedel_beta}.

Recall from \autoref{tab:metavariables} that we use, \eg metavariables $\BBa$, $\BBb$ for Boolean expressions, $\pstate$ for program states, and so on and we will omit providing the types in order to unclutter the presentation.
\subsection{Embedding First-Order Arithmetic in $\boldsymbol{\SyntE}$}
We denote by $\FOArithPosRats$ the set of formulas $\PP$ in \emph{first-order arithmetic} over $\PosRats$, \ie the extension of Boolean expressions $\BB$
(see \Cref{sec:syntax:bool}) by an existential quantifier $\exists x\colon \PP$ and a universal quantifier $\forall x\colon P$ with the usual semantics, \eg
$\sem{\forall x\colon \PP}{\pstate}{} = \true$ 
iff for all $r \in \PosRats$, 
$\sem{\PP}{\pstate\statesubst{x}{r}}{} = \true$.
The set $\FOArithNats$ of formulas $\PP$ in first-order arithmetic over $\Nats$ is defined analogously by
restricting ourselves to
(1) states\footnote{Program states serve here the role of \emph{interpretations} in classical first-order logic.} $\pstate\colon \Vars \to \Nats$
and (2) constants in $\Nats$ rather than $\PosRats$.
For simplicity, \emph{we assume without loss of generality that all formulas $\PP$ are in prenex normalform}, \ie $\PP$ is a Boolean expression comprising of a block of quantifiers followed by a quantifier-free formula.
%
%\begin{definition}[First-order arithmetic]
%    The set $\FOArithPosRats$ of formulas $\PP$ in \emph{first-order arithmetic} over the non-negative rationals $\PosRats$ is given by the grammar
%	\begin{align*}
%	\PP \qqlongrightarrow 	
%	&\BB	\quad \qmid \quad \exists x\colon \PP~ \quad \qmid \quad \forall x\colon \PP~,
%	\end{align*}
%    where $\BB$ is a Boolean expression.
%    For every state $\pstate$, we extend the semantics $\sem{\BB}{\pstate}{}$ of Boolean expressions (see \Cref{sec:syntax:bool})
%    to formulas $\PP$ in $\FOArithPosRats$ in the usual fashion:
%    \begin{align*}
%        \sem{\exists x\colon \PP}{\pstate}{} \eeq &
%        \begin{cases}
%            \true, 
%            & \text{if, for some } r \in \PosRats, \sem{\PP}{\pstate\statesubst{x}{r}}{} = \true \\
%            \false, 
%            & \text{else}
%        \end{cases}
%        \\
%        \sem{\forall x\colon \PP}{\pstate}{} \eeq &
%        \begin{cases}
%            \true, 
%            & \text{if, for all} r \in \PosRats, \sem{\PP}{\pstate\statesubst{x}{r}}{} = \true \\
%            \false, 
%            & \text{else}
%        \end{cases}
%    \end{align*}
%\end{definition}
%
Recall that program states originally evaluate variables to \emph{rationals}.
Since our expressiveness proof requires encoding sequences of \emph{naturals}, 
it is crucial that we can assert that a variable evaluates to a natural.
To this end, we adapt a result by \citet{robinson_define_z}:
%
%Two formulas $\PP, \PP'$ are \emph{equivalent}, denoted $\PP \equiv \PP' \in \FOArithPosRats \cup \FOArithNats$, if their truth values coincide for all program states $\sigma$ and all suitable interpretations $\interpret$, i.e.
%%
%\[
%	\text{for all $\sigma$ and all suitable}~\interpret~ \colon  \PP \equiv \PP' \qquad \text{iff} \qquad \sem{\PP}{\sigma}{\interpret} = \sem{\PP'}{\sigma}{\interpret}~.
%\]
%
%
%
\begin{lemma}\label{lem:nats_definable}
	$\Nats$ is definable in $\FOArithPosRats$, i.e.~there exists a formula
	$\isNat(x) \in \FOArithPosRats$, such that for all $\sigma$,
	\begin{align*}
        	\sem{\isNat(x)}{\pstate}{} \eeq \true \qqiff \pstate(x) \iin \Nats~.
	\end{align*}%
\end{lemma}%
%
%\begin{proof}
%%	We adapt a result by Robinson~\cite{robinson_define_z}. 
%See Appendix~\ref{proof:nats_definable} for details.
%\end{proof}
%
\noindent
We use the above assertion $\isNat$ to first embed $\FOArithNats$ in $\FOArithPosRats$.
Thereafter, we embed $\FOArithPosRats$ in 
$\SyntE$.
Embedding a formula $\PP \in \FOArithNats$ in $\FOArithPosRats$ amounts to
(1) asserting $\isNat(x)$ for every $x \in \FV{\PP}$ and (2) guarding every quantified variable $x$ in $\PP$ with $\isNat(x)$, \ie whenever we attempt to evaluate the embedding-formula for non-naturals, we default to $\false$---see~\Cref{table:fonat_to_foposrat} for a formal definition.
\begin{theorem}\label{thm:fo_rats_subsumes_fo_nats}
    Let $\toFOPosRats{\PP} \in \FOArithNats$ be the embedding of 
    $\PP \in \FOArithNats$ as defined in \textnormal{\Cref{table:fonat_to_foposrat}}.
    Then, for all $\pstate$,
	\begin{align*}
    \sem{\toFOPosRats{P}}{\sigma}{}
	\eeq
	\begin{cases}
    \sem{\PP}{\pstate}{}, & \textnormal{if } \pstate(x) \in \Nats \textnormal{ for all } x \in \FV{\PP}, \\
    \false, & \textnormal{otherwise}~.
	\end{cases}
	\end{align*}
	%
%	\sem{\toFOPosRats{P}}{\sigma}{\interpret}
%	%
%	 \eeq
%	 %
%	\begin{cases}
%	%
%	\false, & \text{if there is}~\VV \in \FV{\PP} ~\text{with}~\interpret(\VV) \not \in \Nats  \\
%	%
%	\sem{\PP}{\sigma}{\interpret}, & \text{otherwise}~.
%	\end{cases}
%	\end{align*}
\end{theorem}
%
%
%
%\begin{proof}
%	See Appendix~\ref{proof:fo_rats_subsumes_fo_nats}.
%\end{proof}
%
%
\noindent{}%
Embedding a formula $\PP \in \FOArithPosRats$ into $\SyntE$ amounts to (1) taking its Iverson bracket for every Boolean
expression and (2) substituting the quantifiers $\exists/\forall$ by their quantitative analogs  $\Sup/\Inf$, see \Cref{table:foposrat_to_synte}.
\begin{theorem}\label{thm:exp_subsumes_fo_rats}
    Let $\iverson{\PP} \in \SyntE$ be the embedding of $\PP \in \FOArithPosRats$ 
    as defined in \textnormal{\Cref{table:foposrat_to_synte}}.
    Then, for all $\pstate$,
%	Define the expectation $\iverson{\PP}$ inductively as follows:
%	%
%	\begin{align*}
%	&\text{If} ~ \PP \eeq \BB ~\text{, then}& \iverson{\PP} &\ddefeq \iverson{\BB} \\
%	%
%	&\text{If} ~ \PP \eeq \left( \exists \VV \colon \PP' \right) ~\text{, then}&\iverson{\PP} & \ddefeq \SupV{\VV} \iverson{\PP'} \\
%	%
%	&\text{If} ~ \PP \eeq \left( \forall \VV \colon \PP' \right) ~\text{, then}& \iverson{\PP} & \ddefeq\InfV{\VV}  \iverson{\PP'}
%	\end{align*}
	%
	%
	\begin{align*}
	\sem{\iverson{\PP}}{\sigma}{} \eeq 
	\begin{cases}
	1, &\text{if}~ \sem{\PP}{\sigma}{} = \true \\
	0, & \text{if}~ \sem{\PP}{\sigma}{} = \false~.
	\end{cases}
	\end{align*}
\end{theorem}
\noindent{}%
Given $\PP(\VV_1, \ldots, \VV_n) \in \FOArithPosRats$, we often write $\toExp{\PP(\VV_1, \ldots, \VV_n)}$ instead of
$\toExp{\PP}(\VV_1, \ldots, \VV_n)$.
\subsection{Encoding Sequences of Natural Numbers}
The embedding of $\FOArithNats$ in our language $\SyntE$ of syntactic expectations
gives us access to a classical result by \citet{goedel_beta} for encoding finite sequences of naturals in a \emph{single} natural.
\begin{lemma}[\textnormal{\citet{goedel_beta}}]\label{lem:goedel_beta}
	There is a formula $\seqelem{\VV_1}{\VV_2}{\VV_3} \in \FOArithNats$ (with quantifiers) satisfying:
    For every finite sequence of natural numbers $n_0,\ldots, n_{k-1}$, there is a (G\"odel) number $\gnum \in \Nats$ 
    that encodes it, \ie
	for all $i \in \{0,\ldots, k-1 \}$ and all $m \in \Nats$, it holds that
	\[
	   \seqelem{\gnum}{i}{m} \eequiv \true \qqiff \quad m \eeq n_i~.
	\]
\end{lemma}
%
%\begin{example}
%We define a factorial predicate $\facpred{\VV_1}{\VV_2} \in \FOArithNats$ that evaluates to $\true$ on interpretation $\interpret$ if and only if $\interpret(\VV_2) = \interpret(\VV_1)!$. For that, we make use of the well-known inductive definition of the factorial: For all $n \in \Nats$, we have
%%
%\begin{align*}
%   0! \eeq 1 \qquad \text{and} \qquad 
%   %
%   (n+1)! \eeq  (n+1) \cdot n!~.
%\end{align*}
%%
%Now consider the formula $\facpred{\VV_1}{\VV_2}$ given by
%%
%\begin{align*}
%   \facpred{\VV_1}{\VV_2} \eeq& \exists \gnum \colon \seqelem{\gnum}{0}{1} \wedge \seqelem{\gnum}{\VV_1}{\VV_2}  \\
%   %
%   & \qquad \wedge \big( \forall \VU \colon \forall \VW \colon 
%           (\VU < \VV_1 \wedge \seqelem{\gnum}{\VU}{\VW})  \longrightarrow %
%           \seqelem{\gnum}{\VU+1}{\VW\cdot(\VU+1)} \big) ~.
%\end{align*}
%%
%%
% The subformula right after the $\exists \gnum$ quantifier evaluates to $\true$ on interpretation $\interpret$ if and only if 
%%
%\begin{enumerate}
%	\item $\interpret(\gnum)$ encodes a sequence $n_0, n_1, \ldots$ satisfying
%       %
%       \[
%        %
%            n_0 = 1, ~ n_1 = 1 \cdot n_0, ~ n_2 = 2 \cdot n_1, ~\ldots,  ~n_{\interpret(\VV_1)} = \interpret(\VV_1) \cdot n_{\interpret(\VV_1) -1}~, \text{and}
%        %
%       \]
%       %
%       \item $\interpret(\VV_2)  = n_{\interpret(\VV_1)} = \interpret(\VV_1)!$.
%\end{enumerate}
%%
%Such a sequence exists  iff $\interpret(\VV_2)  = \interpret(\VV_1)!$.
%Hence,  $\sem{\facpred{\VV_1}{\VV_2}}{}{\interpret} = \true$ iff $\interpret(\VV_2) = \interpret(\VV_1)!$.
%\end{example}
%%
%%
\noindent
By Theorem~\ref{thm:exp_subsumes_fo_rats}, we also have an expectation $\iverson{\seqelem{\VV_1}{\VV_2}{\VV_3}}$ expressing $\seqelemsymbol$ in $\SyntE$.
\begin{example}[Factorials via Gödel]
The syntactic expectation below evaluates to the factorial $x!$:
\begin{align*}
   \facexp{x} \eeq& \Sup \VV \colon \Sup \gnum \colon \VV \cdot 
   \big[ \seqelem{\gnum}{0}{1} \wedge \seqelem{\gnum}{x}{\VV}  \\
   &\quad \wedge \forall \VU \colon \forall  \VW  \colon \bigl(\VU<x \wedge \seqelem{\gnum}{\VU}{\VW} \longrightarrow \seqelem{\gnum}{\VU+1}{\VW\cdot(\VU+1)}  \bigr)
    \big]~.
\end{align*}
For every state $\pstate$, the quantifier $\Sup \gnum$  
selects a sequence $n_0,n_1\ldots$ satisfying $n_{\sigma(x)} = \sigma(x)!$. 
The quantifier $\Sup \VV$ then binds $\VV$ to the value $n_{\sigma(x)} = \sigma(x)!$. Finally, 
by multiplying the $\{0,1\}$-valued expectation specifying the sequence by $\VV$, we get that $\sem{\facexp{x}}{\sigma}{\interpret} = \sigma(x)!$.\hfill$\triangle$
\end{example}
\noindent
To assign a \emph{unique} G\"odel number $\gnum$ to a sequence $n_0,\ldots,n_{k-1}$ of length $k$ we employ \emph{minimalization}, \ie we take the least suitable G\"odel number. Formally, we define the formula
\begin{align*}
   &\sequence{\gnum}{\VV} \\
    &\ddefeq
   \left( \forall \VU \colon \VU < \VV \longrightarrow \exists \VW \colon \seqelem{\gnum}{\VU}{\VW} \right) \\
   &\quad\qquad\wedge\big( \forall \gnum' \colon 
   \big(  \forall \VU \colon \VU<\VV \longrightarrow \exists \VW \colon \seqelem{\gnum}{\VU}{\VW} \wedge \seqelem{\gnum'}{\VU}{\VW} \big) \\
   &\quad\qquad\qquad \longrightarrow \gnum' \geq \gnum \big)~.
\end{align*}
%
%A state $\pstate$ satisfies $ \sequence{\gnum}{\VV}$ iff $\gnum$ encodes a sequence of length at least $\pstate(\VV)$ (first conjunct) and if $\gnum$ is the \emph{least} number for encoding the (sub)sequence $n_0,\ldots,n_{\interpret(\VV)-1}$ (second conjunct). 
For every $k$ and every sequence $n_0,\ldots,n_{k-1}$ of length $k$, we then define \emph{the} G\"odel number encoding the sequence $n_0,\ldots,n_{k-1}$ as the unique natural number $\seqnum{n_0,\ldots,n_{k-1}}$ satisfying
\[
    \sequence{\seqnum{n_0,\ldots,n_{k-1}}}{k} ~{}\wedge{}~ \bigwedge\limits_{i=0}^{k-1} \seqelem{\seqnum{n_0,\ldots,n_{k-1}}}{i}{n_i}~.
\]
\subsection{Encoding Sequences of Non-negative Rationals}
Recall that program states in $\pgcl$ map variables to values in $\PosRats$.
To encode sequences of program states, we thus
first lift G\"odel's encoding $\seqelem{\gnum}{i}{n}$ to uniquely encode sequences over $\PosRats$.
The main idea is to represent such a sequence by pairing two sequences over $\Nats$.
\begin{lemma}[Pairing Functions~\textnormal{\cite{cantor1878beitrag}}]\label{lem:pairing}
	There is a formula $\gPair(\VV_1, \VV_2, \VV_3) \in \FOArithNats$ satisfying: For every pair of natural numbers $(n_1, n_2)$, there is \emph{exactly one} natural number $n$ such that
	\[
	\gPair(n, n_1, n_2) \eequiv \true~.
	\]
\end{lemma}
%
%We combine $\seqelemsymbol$ and $\gPair$ to obtain a formula $\gRho \in \FOArithPosRats$
%for sequences of \emph{non-negative rational numbers}.
%
\begin{theorem}
	\label{thm:rho_expectation}
	There is a formula $\rseqelem{\VV_1}{\VV_2}{\VV_3} \in \FOArithPosRats$ satisfying:
	For every finite sequence $\RR_0, \ldots, \RR_{k-1} \subset \PosRats$  there is a 
	G\"odel number $\gnum$, such that for all $i \in \{ 0,\ldots, k-1 \}$ and $s \in \PosRats$,
	\[
	\rseqelem{\gnum}{i}{s} \eequiv \true \quad \text{iff} \quad s \eeq r_i~.
	\]
\end{theorem}
%
%\begin{proof}
%	See Appendix~\ref{proof:rho_expectation}.
%\end{proof}
% 
%
%
\begin{example}[Harmonic Numbers]
	\label{ex:harmonic}
    For every $\pstate$ with $\pstate(x) = k \in \Nats$, the expectation $\harmexp{x} \in \SyntE$ below evaluates to the $k$-th 
    harmonic number $\mathcal{H}(k) \eeq \sum_{i=1}^{k} \frac{1}{i}$.
	\begin{align*}
	\harmexp{x} \eeq& \Sup \VV \colon \Sup \gnum \colon \VV \cdot 
	\big[ \rseqelem{\gnum}{0}{0} \wedge \rseqelem{\gnum}{x}{\VV} \\
	&\quad \wedge \forall \VU \colon \forall \VW \colon (\VU<x \wedge  \rseqelem{\gnum}{\VU}{\VW}) \\ 
	& \qquad \quad \longrightarrow \exists \VW' \colon \VW' \cdot (\VU+1) = 1 \wedge \rseqelem{\gnum}{\VU+1}{\VW +\VW'} \big] ~
	\end{align*}
    Notice that the above Iverson bracket evaluates to $1$ on state $\sigma$ iff $\sigma(\gnum)$ encodes  a sequence $r_0,r_1,\ldots, r_{\sigma(x)}$ such that $\sigma(\VV) = r_{\sigma(x)}$ and
	\[
	r_0 \eeq 0~, ~r_1 = \frac{1}{1} + r_0,~~ ~r_2 = \frac{1}{2} + r_1~, ~\ldots~, r_{\sigma(x)} \eeq \frac{1}{\sigma(x)} + r_{\sigma(x) -1}~.
	\]   
   By Theorem~\ref{thm:fo_rats_subsumes_fo_nats}, we do not need to require that $\sigma(u) \in \Nats$ as $\rseqelem{\gnum}{i}{w}$ is $\false$ if $\sigma(u) \not\in \Nats$.\hfill$\triangle$
\end{example}
\noindent
Analogously to the previous section, we define a predicate $\rsequence{\gnum}{\VV}$ that uses minimalization to 
a \emph{unique} G\"odel number $\gnum$ for every sequence $r_0,\ldots,r_{k-1}$ of length $k$;
the only difference between $\rsequence{\gnum}{\VV}$ and $\sequence{\gnum}{\VV}$ is that every occurrence of
$\seqelem{.}{.}{.}$ is replaced by $\rseqelem{.}{.}{.}$.
%
%\begin{align*}
%&\rsequence{\gnum}{\VV} \\
%%
%\ddefeq&
%\left( \forall \VU \colon \VU < \VV \longrightarrow \exists \VW \colon \rseqelem{\gnum}{\VU}{\VW} \right) \\
%%
%&\wedge\big( \forall \gnum' \colon 
%\big(  \forall \VU \colon \VU<\VV \longrightarrow \exists \VW \colon \rseqelem{\gnum}{\VU}{\VW} \wedge \seqelem{\gnum'}{\VU}{\VW} \big) \\
%&\qquad \longrightarrow \gnum' \geq \gnum \big)~.
%\end{align*}
%
%A state $\pstate$ satisfies $ \sequence{\gnum}{\VV}$ iff $\gnum$ encodes a sequence of length at least $\pstate(\VV)$ (first conjunct) and if $\gnum$ is the \emph{least} number for encoding the (sub)sequence $n_0,\ldots,n_{\interpret(\VV)-1}$ (second conjunct). 
%
Moreover, for every $k$ and every sequence $r_0,\ldots,r_{k-1}$, we define \emph{the} G\"odel number encoding the sequence $r_0,\ldots,r_{k-1}$ as the unique natural number $\seqnum{r_0,\ldots,r_{k-1}}$ satisfying
\[
\rsequence{\seqnum{r_0,\ldots,r_{k-1}}}{k} ~{}\wedge{}~ \bigwedge\limits_{i=0}^{k-1} \rseqelem{\seqnum{r_0,\ldots,r_{k-1}}}{i}{r_i}~.
\]
\subsection{Encoding Sequences of Program States}
\label{sec:encoding_states}
To encode sequences of program states, we first fix a finite set $\varseq{x} = \{x_0,\ldots,x_{k-1}\}$ of \emph{relevant variables}.
Intuitively, \varseq{x} consists of all variables that appear in a given program or a postexpectation. We define an equivalence relation $\equivstatesrel{\varseq{x}}$ on states by 
\[
\equivstates{\sigma_1}{\varseq{x}}{\sigma_2}
\qquad
\text{iff}
\qquad 
\forall x \in \varseq{x}\colon \sigma_1(x) = \sigma_2(x)~.
\]

Every $\gnum$ satisfying $\sequence{\gnum}{k}$ encodes \emph{exactly one} state $\pstate$ (modulo $\equivstatesrel{\varseq{x}}$).
The G\"odel number encoding~$\sigma$ (w.r.t.\ $\varseq{x}$), which we denote by $\seqnum{\sigma}_{\varseq{x}}$, is then the unique number satisfying
\[
\rsequence{\seqnum{\sigma}_{\varseq{x}}}{k} ~{}\wedge{}~\bigwedge\limits_{i=0}^{k-1}
\rseqelem{\seqnum{\sigma}_{\varseq{x}}}{i}{\sigma(x_i)}~.
\]
\noindent
Notice that we implictly fixed an ordering of the variables in $\varseq{x}$ to identify each value stored in~$\pstate$
for a variable in $\varseq{x}$.
The formula
%The fact that $\gnum$ is the G\"odel number of the \emph{current}\blkcomment{Was heißt current state hier?} state $\sigma$ is expressed by
%
\[
\encodesstate{x}{\gnum} \ddefeq \rsequence{num}{k} ~{}\wedge{}~\bigwedge\limits_{i=0}^{k-1}
\rseqelem{\gnum}{i}{x_i}
\]
\noindent
evaluates to $\true$ on state $\sigma$ iff $\sigma(\gnum)$ is the G\"odel number of a state $\sigma'$ with $\equivstates{\sigma}{\varseq{x}}{\sigma'}$.
Now, let $\sigma_0,\ldots,\sigma_{n-1}$ be a sequence of states of length $n$. 
\emph{The} G\"odel number encoding $\sigma_0,\ldots,\sigma_{n-1}$ (w.r.t.\ $\varseq{x}$), which we denote by $\stateseqnum{\sigma_0,\ldots,\sigma_{n-1}}_{\varseq{x}}$, is then the unique number satisfying
\[
\sequence{\stateseqnum{\sigma_0,\ldots,\sigma_{n-1}}_{\varseq{x}}}{n} ~{}\wedge{}~\bigwedge\limits_{i=0}^{n-1}
\seqelem{\stateseqnum{\sigma_0,\ldots,\sigma_{n-1}}_{\varseq{x}}}{i}{\seqnum{\sigma_i}_{\varseq{x}}}~.
\]
We are now in a position to encode sequences of states. 
The formula
\begin{align*}
&\stateseq{x}{\gnum}{\VV}  \\
\eeq& \sequence{\gnum}{\VV} \wedge \left( \exists \VV' \colon \seqelem{\gnum}{0}{\VV'} \wedge \encodesstate{x}{\VV'} \right) \\
&\wedge \forall \VU \colon \forall \VV' \colon \left( ( \VU <\VV \wedge\seqelem{\gnum}{\VU}{\VV'} )\longrightarrow \rsequence{\VV'}{k} \right)
\end{align*}
evaluates to $\true$ on state $\pstate$ iff (1) $\gnum$ is the G\"odel number of some sequence $\pstate_0,\ldots,\sigma_{\pstate(\VV-1)} \in \States$ of states of length $\pstate(\VV)$ and where (2) $\sigma$ and $\sigma_0$ coincide on all variables in $\varseq{x}$, \ie $\equivstates{\sigma}{\varseq{x}}{\sigma_0}$.
Notice that, for every sequence $\sigma_0,\ldots,\sigma_{n-1}$ of states of length $n$, there is \emph{exactly one} $\gnum$ satisfying $\stateseq{x}{\gnum}{n}$.
If clear from the context, we often omit the subscript $\varseq{x}$ and simply write $\seqnum{\sigma}$ (resp.\ $\stateseqnum{\sigma_0,\ldots,\sigma_{n-1}}$) instead of $\seqnum{\sigma}_{\varseq{x}}$ (resp.\ $\stateseqnum{\sigma_0,\ldots,\sigma_{n-1}}_{\varseq{x}}$).

%% file: figure_translation.tex
% !TEX root = ./main.tex

\begin{figure}[t]
	\begin{minipage}{.5\textwidth}
		\centering
	\renewcommand{\arraystretch}{1.5}
	\begin{tabular}{@{\hspace{1em}}l@{\hspace{2em}}l}
		\hline\hline
		$P$			& $\toFOPosRats{P}$ \\
		\hline
		$\BB$					&  
        $\BB \wedge \isNat(x_1) \wedge \ldots \wedge \isNat(x_n)$  \\
		$\exists x \colon \PP' $			& $\exists x \colon \toFOPosRats{\PP'} $ \\
		$\forall x \colon \PP' $		& $\forall x \colon \isNat(x) \longrightarrow \toFOPosRats{\PP'} $ \\
		\hline\hline
	\end{tabular}
    \captionof{figure}{Rules defining the formula formula $\toFOPosRats{\PP} \in \FOArithPosRats$ for a Boolean expression $\BB$ and $\FV{\PP} = \{x_1,\ldots,x_n\}$. }
	\label{table:fonat_to_foposrat}
\end{minipage}
\qquad
\begin{minipage}{.39\textwidth}
	\centering
	\renewcommand{\arraystretch}{1.5}
	\begin{tabular}{@{\hspace{1em}}l@{\hspace{2em}}l}
		\hline\hline
		$P$			& $\iverson{P}$ \\
		\hline
		$\BB$					&  
		$\iverson{\BB}$ 																					\\
		$\exists \VV \colon \PP' $			& $\Sup \VV \colon \iverson{\PP'} $ \\
		$\forall \VV \colon \PP' $		& $\Inf \VV \colon  \iverson{\PP'} $ \\
		\hline\hline
	\end{tabular}
	\captionof{figure}{ Rules for transforming a formula $P \in \FOArithPosRats$ into an expectation $\iverson{\PP} \in \SyntE$.}
	\label{table:foposrat_to_synte}
	\end{minipage}
\end{figure}

%% file: normalforms.tex
% !TEX root = ./main.tex
%
\section{The Dedekind Normal Form}
%\section{Normal Forms}
\label{sec:normalforms}
Before we encode sums and products of non-constant size in $\SyntE$---as required to deal with the challenges in~\Cref{sec:outline:sum,sec:outline:binary-product,sec:outline:product}---we introduce a normal form that gives a convenient handle to 
encode real numbers as syntactic expectations.

As a first step, we transform syntactic expectations into prenex normal form, \ie we rewrite every $\FF \in \SyntE$ into an equivalent syntactic expectation of the form %a semantically equivalent expectation of the form
$\Quant_1 \VV_1 \ldots \Quant_k \VV_k  \colon \FF'$,
where $\Quant_i \in \{\Sup, \Inf\}$ and $\FF'$ is \enquote{quantifier}--free, \ie contains neither \Sup{} nor \Inf.
%
%to reduce the technical effort in the remaining proofs.
%we introduce three (effectively constructible) normalforms for our syntactic expectations.
%Their sole purpose is to reduce the overall technical efforts in the remaining proofs.
%In particular, the \emph{summation normal form} represents every syntactical expectation as a constant-size sum of guarded arithmetic expressions with a quantifier prefix.
%In particular, the \emph{Dedekind normal form} gives a convenient handle to encode real numbers as syntactical expectations. 
%\cmcommentinline{Move everything but the two normal forms to the appendix}
%\subsection{Prenex Normal Form}
%As discussed in \Cref{sec:note-on-f-times-f},
%we first transform syntactic expectations into prenex normal form, i.e., we rewrite every $\FF \in \SyntE$ into an equivalent syntactic expectation of the form %a semantically equivalent expectation of the form
%%
%%
%$\Quant_1 \VV_1 \ldots \Quant_k \VV_k  \colon \FF'$,
%%
%%
%where $\Quant_i \in \{\Sup, \Inf\}$ and $\FF'$ is \enquote{quantifier}--free, i.e.\ contains neither \Sup{} nor \Inf.
%
The following lemma justifies that any expectation can indeed be transformed into an equivalent one in prenex normal form
by iteratively pulling out quantifiers.
In case the quantified logical variable already appears in the expectation the quantifier is pulled over, we rename it by a fresh one first.
\begin{lemma}[Prenex Transformation Rules]
\label{thm:prenex-rules}
 For all \FF, $\FF_1$, $\FF_2 \in \SyntE$, terms $\TT$, and Boolean expressions $\BB$, quantifiers $\Quant \in \{\Sup, \Inf\}$, and fresh logical variables $\VV'$, the following equivalences hold:
\begin{enumerate}
	\item\label{thm:prenex-rules:plus-left} 
	$(\Quant \VV\colon \FF_1) \pplus \FF_2 \eequiv \Quant \VV'\colon \FF_1\subst{\VV}{\VV'} \pplus \FF_2$, 
	\item\label{thm:prenex-rules:plus-right} 
	$\FF_1 \pplus (\Quant \VV\colon \FF_2) \eequiv \Quant \VV'\colon \FF_1 \pplus \FF_2\subst{\VV}{\VV'}$, 
	\item\label{thm:prenex-rules:mult-term} 
	$\TT \ccdot \Quant \VV\colon \FF  \hspace*{1.555em} \eequiv  \Quant \VV'\colon \TT \cdot \FF\subst{\VV}{\VV'}$
	\quad and \quad
	$(\Quant \VV\colon \FF) \ccdot \TT  \eequiv  \Quant \VV'\colon (\FF\subst{\VV}{\VV'} \cdot \TT)$,
	\item\label{thm:prenex-rules:mult-guard} 
	$\iverson{\BB} \ccdot \Quant \VV\colon \FF \hspace*{.49em} \eequiv   \Quant \VV'\colon \iverson{\BB} \cdot \FF\subst{\VV}{\VV'}$
	\quad and \quad
	$(\Quant \VV\colon \FF) \ccdot \iverson{\BB}  \eequiv  \Quant \VV'\colon (\FF\subst{\VV}{\VV'} \cdot \iverson{\BB})$.
\end{enumerate}
\end{lemma}
\noindent
The Dedekind normal form is motivated by the notion of \emph{Dedekind cuts}~\cite{bertrand1863traite}.
We denote by $\dcut{\alpha}$ the Dedekind cut of a real number, \ie the set of all rationals strictly smaller than $\alpha$. 
%
%\[
%	\dcut{a} \eeq \setcomp{r \in \Rats}{r < a}~.
%\]
%
%
In the realm of \emph{all} reals, it is required that a Dedekind cut is neither  the empty set nor the whole set of rationals $\Rats$. 
However, since we operate in the realm of non-negative reals with infinity $\PosRealsInf$, we \emph{do} allow for both empty cuts and $\PosRats$. More formally, we define:
\begin{definition}
	Let $\alpha \in \PosRealsInf$.
	The \emph{Dedekind cut} $\dcut{\alpha} \subseteq \PosRats$ of $\alpha$ is defined as
	\[
		\dcut{\alpha} \ddefeq \setcomp{r \in \PosRats}{ r < \alpha}~.
	\]
	Furthermore, we define $\dcutzero{\alpha} \defeq \dcut{\alpha} \cup \{0\}$.
\end{definition}
\noindent
Dedekind cuts are relevant for our technical development as they allow to describe every real number $\alpha$ as a supremum over a set of rational numbers. 
In particular, the Dedekind cut $\dcut{0}$ of $0$ is the empty set with supremum $0$, and the Dedekind cut $\dcut{\infty}$ of $\infty$ is the set $\PosRats$ with supremum $\infty$.
Formally:
\begin{lemma} For every $\alpha \in \PosRealsInf$, we have $\alpha = \sup \dcut{\alpha}$. 
\end{lemma}
%
%Since the supremum of $\emptyset$ (resp.\ $\PosRats$) in the domain $\PosRealsInf$ is $0$ (resp.\ $\infty$), our definition preserves the property that $a = \sup \dcut{a}$ for every $a \in \PosRealsInf$.
%
%
\begin{theorem}
	\label{thm:dedekind_nf}
    For every $\FF \in \SyntE$, there is a syntactic expectation in prenex normal form
    \[ \dexp{\FF} \eeq  \qprefixnoarg \colon \iverson{\BB}~, \]
    where $ \qprefixnoarg$ is the quantifier prefix,
    $\BB$ is an effectively constructible Boolean expression,
    and the free variable $\VVcut$ is fresh; we call $\dexp{\FF}$ the \emph{Dedekind normal form} of $\FF$.

    Moreover, for all program states $\pstate$, we have
	\[
			\sem{\dexp{\FF}}{\sigma}{\interpret} 
			\eeq
			\begin{cases}
			   1, & \text{if}~~\sigma(\VVcut) < \sem{f}{\sigma}{\interpret} \\
			   0, &\text{otherwise}~.
			\end{cases}
	\]
\end{theorem}
%
%
%\begin{proof}
%	See Appendix~\ref{proof:dedekind_nf}.
%\end{proof}
%
%
%We call $\VVcut$ the \emph{cut variable} of $\dexp{\FF}$.
The Dedekind normal form $\dexp{\FF}$ defines the Dedekind cut of every $\sem{\FF}{\sigma}{\interpret}$, \ie
\[
    \text{for all $\sigma$}\colon \quad \dcut{\sem{\FF}{\sigma}{\interpret}}
    \eeq \setcomp{\RR \in \PosRats}{ \RR = \sigma(\VVcut), \sem{\dexp{\FF}}{\sigma}{\interpret\statesubst{\VVcut}{\RR}} = 1 }~.
\] 
Hence, we can recover $\FF$ from $\dexp{\FF}$:
\begin{lemma}
	\label{lem:dedekind_nf_recover}
	Let $\dexp{\FF}$ be in Dedekind normal form. Then 
	\[
	      \FF \eequiv \Sup \VVcut \colon \dexp{\FF} \cdot \VVcut~. 
	\]
\end{lemma}
%
%\begin{proof}
%	See Appendix \ref{proof:dedekind_nf_recover}.
%\end{proof}

%% file: sums_prod_goedel.tex
\section{Sums, Products, and Infinite Series of Syntactic Expectations}
\label{sec:sums_prod_via_goedel}

This section deals with the syntactic $\gsumsymbol$ and $\gproductsymbol$ expectations as described in \Cref{sec:outline:sum}. Since a syntactic expectation $\FF$ evaluates to a non-negative \emph{extended real}, we rely on a reduction from sums over reals to suprema of sums over \emph{rationals}:

\begin{lemma}
	\label{lem:sum_by_cut}
	For all $\alpha_0,\ldots,\alpha_n \in \PosRealsInf$, we have
	\[
	\sum_{j=0}^n  \alpha_j 
	\eeq
	\sup \setcomp{\sum_{j=0}^n  r_j}{\forall i \in \{0,\ldots,n\} \colon  r_i \in \dcutzero{\alpha_i}}
	\]
\end{lemma}
\begin{theorem}
	\label{thm:sum_exp}
	For every $\FF \in \SyntE$ with free variable $\vsum$,
	there is an effectively constructible expectation 
	$\gsum{\FF}{\VV} \in \SyntE$ such that for all states $\sigma$ with $\sigma(\VV) \in \Nats$, we have
	\[
	\sem{\gsum{\FF}{\VV}}{\sigma}{\interpret}
	\eeq
	\sum_{j=0}^{\sigma(\VV)} \sem{\FF\subst{\vsum}{j}}{\sigma}{\interpret}
	~
	\text{and}
	~
	\sem{\Sup \VV \colon \gsum{\FF}{\VV}}{\sigma}{\interpret}
	\eeq
	\sum_{j=0}^{\infty} \sem{\FF\subst{\vsum}{j}}{\sigma}{\interpret}~.
	\]
\end{theorem}
\begin{proof}
	We sketch the construction of $\gsum{\FF}{\VV}$. 
	\noindent
	\Cref{lem:sum_by_cut} and the Dedekind normal form $\dexp{\FF}$ of $\FF$  (cf. \Cref{thm:dedekind_nf}) give us
	\begin{align}
	&\sum_{j=0}^{\sigma(\VV)}  \sem{\FF\subst{\vsum}{j}}{\sigma}{\interpret}
	\notag\\
	\eeq&
	\sup \setcomp{\sum_{j=0}^{\sigma(\VV)}  r_j}{\forall j \in \{0,\ldots,\sigma(\VV)\} \colon  r_j \in \dcutzero{\sem{\FF\subst{\vsum}{j}}{\sigma}{\interpret}}} 
	\notag\\
	\eeq&
	\sup \setcomp{\sum_{j=0}^{\sigma(\VV)}  r_j}{\forall j \in \{0,\ldots,\sigma(\VV)\} \colon  \sem{\dexpvar{\FF}{r_j}}{\sigma}{\interpret}=1~\text{or}~r_j = 0}~.
	\label{eqn:derive_sum_exp}
	\end{align}
	Writing $\dexp{\FF} =  \qprefixnoarg \colon \iverson{\BB}$ (cf.\ \Cref{thm:dedekind_nf}) and denoting by $ \qprefixnoarginvert$ the quantifier prefix obtained from $\qprefixnoarg$ by flipping all quantifiers, 
	we then construct a syntactic expectation $\FG$ with free variables $\VV$ and $\gnum$ by
	\begin{align*}
	&\Sup \VV' \colon 
	\VV' \cdot \Inf \VU \colon \Inf z \colon \Sup \VVcut \colon \qprefixnoarginvert \colon \\
	& [
	\rseqelem{\gnum}{0}{1} \wedge \rseqelem{\gnum}{\VV+1}{\VV}  \\
	&  \quad \wedge \big( (\VU < \VV+1 \wedge \rseqelem{\gnum}{\VU}{z} \wedge (\iverson{\BB} \subst{\vprod}{\VU} \vee \VVcut = 0) ) \\
	& \qquad \quad  \longrightarrow  \rseqelem{\gnum}{\VU+1}{z + \VVcut}   \big)]~.
	\end{align*}
	For every state $\sigma$ where $\sigma(\gnum)$ is a G\"odel number encoding some sequence
	\begin{align*}
	1,~~~~ 1\cdot r_1, ~~~~1 + r_1 + r_2, ~~~~\ldots ~~~~, 1 +  r_1 + \ldots +  r_{\sigma(\VV)}
	\end{align*}
	with $r_j \in  \dcutzero{\sem{\FF\subst{\vsum}{j}}{\sigma}{}}$
	for all $0 \leq j \leq \sigma(\VV)$, expectation $\FG$ evaluates to the last element of the above sequence, \ie an element of the set from Equation (\ref{eqn:derive_sum_exp}). Hence, by \Cref{lem:sum_by_cut}, the supremum over these sequences, i.e, all G\"odel numbers, gives us 
	\[
	\gsum{\FF}{\VV} \eeq \Sup \gnum \colon \FG~.
	\]
	See Appendix \ref{proof:sum_exp} for a detailed proof.
\end{proof}
\noindent
For an arithmetic expression $\TT$, we write  $\gsum{\FF}{\TT}$ instead of $\gsum{\FF}{\VV}\subst{\VV}{\TT}$. 

\begin{example}
	$\gsumsymbol$ provides us with a much more convenient way to construct $\harmexp{x}$ from \Cref{ex:harmonic}. Let $\FF = \nicefrac{1}{\vsum}$ where $\nicefrac{1}{\vsum}$ is a shorthand for $\Sup \VW \colon \VW \cdot \iverson{\VW \cdot \vsum = 1}$. Then, by \Cref{thm:sum_exp}, we have for every $\sigma \in \States$
	\begin{align*}
	  \sem{\gsum{\FF}{x}}{\sigma}{\interpret} 
	  \eeq
	  \sum_{j=0}^{\sigma(x)} \sem{\FF\subst{\vsum}{j}}{\sigma}{\interpret} 
	  \eeq
	  \sum_{j=1}^{\sigma(x)} \frac{1}{j}
	  \eeq
	  \mathcal{H}(\sigma(x))~.
	\end{align*}
\end{example}

The construction of the syntactic $\gproductsymbol$ expectation is completely analogous:
\begin{theorem}
	\label{thm:prod_exp}
	For every $\FF \in \SyntE$ with free variable $\vprod$,
	there is an effectively constructible expectation 
	$\gproduct{\FF}{\VV} \in \SyntE$ such that for every
	 state $\sigma$ with $\sigma(\VV) \in \Nats$, we have
	\[
	     \sem{\gproduct{\FF}{\VV}}{\sigma}{\interpret}
	     \eeq
	     \prod_{j=0}^{\sigma(\VV)} \sem{\FF\subst{\vprod}{j}}{\sigma}{\interpret}~.
	\]
\end{theorem}
%\begin{proof}
%	%
%See Appendix~\ref{proof:prod_exp}.
%\end{proof}
%
%
%
For an arithmetic expression $\TT$, we write  $\gproduct{\FF}{\TT}$ instead of $\gproduct{\FF}{\VV}\subst{\VV}{\TT}$.

An immediate, yet important, consequence of Theorem~\ref{thm:prod_exp} is that, even though syntactically forbidden, \emph{arbitrary products} of syntactic expectations are expressible in $\SyntE$.
Let $\FF, \FG \in \SyntE$, and let $\vprod$ be a fresh variable. We define the \emph{(unrestricted) product} $\FF \exprod \FG$ of $\FF$ and $\FG$ by
	\[
	      \FF \exprod \FG \ddefeq 
	      \gproduct{\iverson{\vprod = 0} \cdot \FF + \iverson{\vprod = 1} \cdot \FG}{1}~.
	\]
\begin{corollary}
	\label{cor:unrestricted_product}
	 Let $f,g \in \SyntE$. For all states $\sigma$, we have
	 \[
	    \sem{\FF \exprod \FG}{\sigma}{\interpret}
	    \eeq
	    \sem{\FF}{\sigma}{\interpret} \cdot \sem{\FG}{\sigma}{\interpret}~.
	 \]
\end{corollary}

%% file: expressiveness.tex
\section{Expressiveness of our Language}
\label{sec:expressiveness}

With the results from the preceding sections at hand, we give a constructive expressiveness proof for our language $\SyntE$.
%for \emph{all} $\pgcl$ programs---in particular for loops.
%
Fix a set of variables $\varseq{x} = \{ x_0,\ldots, x_{n-1}\}$.
We assume a fixed set $\partitionedstates{x} \subseteq \States$ that contains \emph{exactly one} state from each equivalence class of $\equivstatesrel{\varseq{x}}$ (cf.\ \Cref{sec:encoding_states}).
Given a state $\sigma \in \States$, we define the \emph{characteristic expectation} $\statepred{\sigma}{\varseq{x}}$ of $\sigma$ (w.r.t.\ $\varseq{x}$) as
\[
\statepred{\sigma}{\varseq{x}} \ddefeq \iverson{x_0 = \sigma(x_0) \wedge \ldots \wedge x_{n-1} = \sigma(x_{n-1})}~.
\]
The expectation $\statepred{\sigma}{\varseq{x}}$ evaluates to $1$ on state $\sigma'$ if $\equivstates{\sigma}{\varseq{x}}{\sigma'}$, and to $0$ otherwise.  Finally, we denote by $\Vars(\cc)$ the set of all variables that appear in the $\pgcl$ program $\cc$. 

Let us now formalize the characterization of $\wp{\WHILEDO{\BB}{\cc'}}{\eval{\FF}}$ from \Cref{sec:kozen_duality}:
%
%If clear from the context, we often omit the subscribt $\varseq{x}$ and simply write $\statepred{\sigma}{}$. Notice that if $\Vars(\ff) \subseteq \varseq{x}$ for some expectation $\ff$, then $\equivstates{\sigma_1}{\varseq{x}}{\sigma_2}$
%implies $\ff(\sigma_1) = \ff(\sigma_2)$
%
\begin{theorem}
	\label{thm:wp_loop_as_sum}
	Let $\cc = \WHILEDO{\BB}{\cc'}$ be a loop and let $\FF \in \SyntE$. Furthermore, let $\varseq{x}$ be a finite set of variables with $\Vars(\cc) \cup \FV{\FF} \subseteq \varseq{x}$.  We have
	\begin{align*}
	&\wp{\WHILEDO{\BB}{\cc'}}{\eval{\FF}}  \\
	\eeq &\lambda \sigma \mydot \sup_{k \in \Nats}
	\sum_{\sigma_0,\ldots,\sigma_{k-1} \in \partitionedstates{x}}
	\statepred{\sigma_0}{\varseq{x}}(\sigma)
	\cdot (\iverson{\neg \BB} \cdot \eval{\FF})(\sigma_{k-1}) \\
	 &\qquad \qquad \qquad \qquad \qquad \cdot
	\prod\limits_{i=0}^{k-2} \wp{\ITE{\BB}{\cc'}{\SKIP}}{\statepred{\sigma_{i+1}}{\varseq{x}}}(\sigma_i)~.
	\end{align*}
\end{theorem}
\begin{proof}
	See Appendix \ref{proof:wp_prob_times_exp}.
\end{proof}
We are finally in a position to prove expressiveness (cf. \Cref{def:expressiveness}).
\begin{theorem}
	\label{thm:expressive}
	The language $\SyntE$ of syntactic expectations is expressive.
\end{theorem}
\begin{proof}
	By induction on the structure of $\cc$. 
    All cases except loops are completely analogous to the proof of
    Lemma~\ref{thm:expressive-loop-free}.
 Let us thus consider the case $\cc \eeq \WHILEDO{\BB}{\cc_1}$.
    We employ the syntactic $\gsumsymbol$- and $\gproductsymbol$ expectations from Theorems~\ref{thm:sum_exp} and~\ref{thm:prod_exp} to
    construct the series from Theorem~\ref{thm:wp_loop_as_sum} in $\SyntE$, thus expressing $\wp{\WHILEDO{\BB}{\cc_1}}{\sem{\FF}{}{}}$.
    %Now Recall from Theorem~\ref{thm:wp_loop_as_sum} that 
    %
%    	\begin{align*}
%    &\wp{\WHILEDO{\BB}{\cc'}}{\ff}(\sigma)  \\
%    %
%    \eeq & \sup_{k \in \Nats}
%    \sum_{\sigma_0,\ldots,\sigma_{k-1} \in \partitionedstates{x}}
%    \statepred{\sigma_0}{\varseq{x}}(\sigma)
%    \cdot (\iverson{\neg \BB} \cdot \ff)(\sigma_{k-1})\\
%    &\qquad \qquad \qquad \qquad \cdot
%    \prod\limits_{i=0}^{k-2} \wp{\ITE{\BB}{\cc'}{\SKIP}}{\statepred{\sigma_{i+1}}{\varseq{x}}}(\sigma_i)~.
%    \end{align*}
    %
   %We show that this infinite series is expressible in $\SyntE$ 
   %by means of the syntactic $\gsumsymbol$- and $\gproductsymbol$ expectations from %Theorems~\ref{thm:sum_exp} and~\ref{thm:prod_exp}, respectively. 
    
    The products ocurring in Theorem~\ref{thm:wp_loop_as_sum} are expressed
     %using $\gproductsymbol$ 
     by an effectively constructible  syntactic expectation $\pathexp{\VV_1}{\VV_2}$ (where $\VV_1$ and $\VV_2$ are fresh variables) satisfying:
    \begin{enumerate}
    	\item 
    	\label{eqn:pathexp_spec_1}
    If $\sigma(\VV_1) \in \Nats$ with $\sigma(\VV_1) > 0$ and  $\sigma(\VV_2)=\stateseqnum{\sigma_0,\ldots,\sigma_{\sigma(\VV_1)-1}}_{\varseq{x}}$, then 
    \begin{align}
       &\sem{\pathexp{\VV_1}{\VV_2}}{\sigma}{\interpret} \notag\\
      \eeq & (\iverson{\neg \BB} \cdot \sem{\FF}{}{})(\sigma_{\sigma(\VV_1)-1})
       \cdot
       \prod\limits_{i=0}^{\sigma(\VV_1)-2} \wp{\ITE{\BB}{\cc_1}{\SKIP}}{\statepred{\sigma_{i+1}}{\varseq{x}}}(\sigma_i)
    \end{align}
    %\substack{\text{if}~\sigma(\VV_1)  \not\in\Nats  ~\text{or}~\sigma(\VV_1) = 0~.}
    %
    \item
    \label{eqn:pathexp_spec_2}
    If $\sigma(\VV_1)  \not\in\Nats$  or $\sigma(\VV_1) = 0$, then 
         $\sem{\pathexp{\VV_1}{\VV_2}}{\sigma}{\interpret} = 0$.
    \end{enumerate}
    Then, for the syntactic expectation 
    \begin{align*}
       \FH \eeq & \semleft{\Sup length \colon \Sup nums \colon
       	\gsumsymbol\big[\vsum, \iverson{\stateseq{x}{\vsum}{length}}}{}{} \\
       &\qquad \qquad \qquad \qquad\qquad\qquad \qquad 
       \exprod	\pathexp{length}{\vsum}, {nums} \big] \semright~,
    \end{align*}
    we have 
    ~$\wp{\WHILEDO{\BB}{\cc_1}}{\sem{\FF}{}{}} = \eval{\FH}$.
    Here, the quantifier $\Sup length$ in $\FH$ corresponds to the $\sup k$ from Theorem~\ref{thm:wp_loop_as_sum}.  The subsequent $\gsumsymbol$ expectation expresses the sum from Theorem~\ref{thm:wp_loop_as_sum}: Summing over sequences of states of length $length$ is realized by summing over all G\"odel numbers $\gnum$ satisfying $\stateseq{\varseq}{\gnum}{length}$. See \ref{proof:epressiveness_appendix} for a detailed correctness proof.
\end{proof}
\subsection{Example}
    We conclude this section by sketching the construction of a syntactic expectation for a concrete loop.
	Consider the program $\cc$ given by 
	\begin{align*}
	    &\WHILE{c = 1 } \\
	    & \quad \PCHOICE{\ASSIGN{c}{0}}{\nicefrac{1}{2}}{\ASSIGN{c}{1}}; \\
	    &\quad \ASSIGN{x}{x+1} ~ \}
	\end{align*}
	where we denote the loop body by $\cc'$.
    Morever, let $\FF \defeq x \in \SyntE$. 
    Then the syntactic
    expectation $\FH$ expressing $\wp{\WHILEDO{c=1}{\cc'}}{\sem{x}{}{}}$ as sketched in the 
    proof of Theorem~\ref{thm:expressive} is 
    %is constructive, we obtain a syntactic 
	%
	%
	\begin{align*}
	\FH   
	\eeq
	&\semleft{\Sup length \colon \Sup nums \colon
		\gsumsymbol\big[\vsum, \iverson{\stateseq{x}{\vsum}{length}}}{}{} \\
	&\qquad \qquad \qquad \qquad\qquad\qquad \qquad 
	\exprod	\pathexp{length}{\vsum}, {nums} \big] \semright~,
	\end{align*}
	where the syntactic expectation $\pathexppost{\FF}{length}{\VV_2}$ is defined as follows:
	\begin{align*}
	%&\pathexppost{\FF}{length}{\VV_2} \\
	%
	%\eeq& 
	%
    & \iverson{length < 2} \cdot (\Sup \gnum \colon \iverson{\seqelem{\vsum}{length -1}{\gnum}}
	\exprod \gapply{x}{(\iverson{\neg (c=1)} \cdot x)}{\gnum}) \\
	+& \iverson{length \geq 2} \cdot
	(\Sup \gnum \colon \iverson{\seqelem{\vsum}{length -1}{\gnum}}
	\exprod \gapply{x}{(\iverson{\neg (c=1)} \cdot x)}{\gnum}) \\
	&\quad\exprod
	\gproductsymbol \big(\Sup \gnum_1 \colon \Sup \gnum_2 \colon
	\iverson{\seqelem{\vsum}{\vprod}{\gnum_1} \wedge \seqelem{\vsum}{\vprod+1}{\gnum_2}} \\
	&\qquad  \quad 
	\exprod
	\gapply{\varseq{x}}{\gsubst{\varseq{x'}}{g}{\gnum_2}}{\gnum_1}, length - 2 \big)
	\end{align*}
	and where 
	\begin{align*}
	   \FG \eeq \iverson{c=1} \cdot \frac{1}{2} \cdot \left( \iverson{0=c' \wedge x+1=x'} + \iverson{1=c' \wedge x+1=x'} \right) + \iverson{\neg{(c=1)}} \cdot \iverson{c=c' \wedge x=x'}~.
	\end{align*}
	We omit unfolding $\FH$ further. Although our general construction yields rather complex syntactic preexpectations, notice we can express $\wp{\WHILEDO{c=1}{\cc'}}{\sem{x}{}{}}$ much more concisely as 
	\[
	    x \pplus \iverson{c=1} \cdot 2 \quad{}\in{}\quad \SyntE~.
	\]

%% file: negatives.tex
% !TEX root = ./main.tex

\section{On Negative Numbers}\label{sec:negatives}

%The attentive reader may have noticed that we evaded negative numbers throughout the paper:
%with regard to two aspects:
Throughout the paper, we have evaded supporting negative numbers in two aspects:
\begin{enumerate}
	\item In our \emph{verification system}---the weakest preexpectation calculus---we allow expectations, both syntactic and semantic, to map program states to \emph{non-negative} values in $\PosRealsInf$ only.
	\item In our \emph{programming language}, we allow variables to assume \emph{non-negative} values in $\PosRats$ only.
\end{enumerate}%
While the former restriction is fairly standard in the literature on probabilistic programs (cf.~\cite{McIverM05}), considering only unsigned program variables is less common.
%(cf.~\cite{winskel}).
An attentive reader may thus ask whether our completeness results rely on the above restrictions.
In this section, we briefly comment on our reasons for
considering only non-negative numbers.
Moreover, we discuss how one \emph{could} incorporate support for negative numbers in both of the above aspects.

\subsection{Signed Expectations}

%Although the assumption that expectations evaluate to non-negative values is fairly common, 
There exist approaches that support signed expectations, which allow arbitrary reals in their codomain.
%Another angle of attack are singed expectations, which allow negative numbers as part of their codomain.
%The benefit is that one could then allow for signed program variables as well, without having to worry about positivity at all.
%On the downside, signed expectations require significant technical overhead~\cite{DBLP:conf/lics/KaminskiK17} and proof rules for loops become much more involved involved.
However, as working with signed expectations may lead to integrability issues, 
these approaches require a significant technical overhead (cf. \cite{DBLP:conf/lics/KaminskiK17} for details). 
Moreover, proof rules for loops become much more involved.
Calculi like Kozen's PPDL \emph{in principle} allow signed expectations off-the-shelf, but PPDL's induction rule for loops is restricted to non-negative expectations as well~\cite{Kozen1983}.
We thus opted for the more common approach of considering only unsigned expectations.
An alternative is to perform a \emph{Jordan decomposition} on the expectation (\ie decomposing it into positive and negative parts) and then reason individually about the positive and the negative part.
As outlined below, such a decomposition can already be performed on program level \emph{without} changing the verification system.

\subsection{Signed Program Variables}

Omitting negative numbers does \emph{not} affect our results because they can easily be encoded in our (Turing complete) programming language:
%On the other hand, Turing-completeness does not require negative numbers. 
we can emulate signed variables, for instance, by splitting each variable $x$ into two variables $\abs{x}$ and $x_{\mathit{sgn}}$, representing the absolute value of $x$ and its sign ($x_{\mathit{sgn}} = 1$ if $x$ negative, and $x_{\mathit{sgn}} = 0$ otherwise), respectively.
With this convention, 
the program below emulates 
the subtraction assignment $\ASSIGN{z}{x - y}$ using only addition and monus: % as follows:%
\begin{align*}
	& \IF{x_{\mathit{sgn}} = y_{\mathit{sgn}}} \tag*{\textcolor{gray}{\texttt{// calculuate magnitude of $z$}}}\\
	& \qquad \ASSIGN{\abs{z}}{\bigl(\abs{x} \monus \abs{y} \bigr) \pplus \bigl(\abs{y} \monus \abs{x} \bigr)} \\
	& \ELSE\\
	& \qquad \ASSIGN{\abs{z}}{\abs{x} + \abs{y}} \\
	& \COMPOSE{\}}{} \\
	& \IF{\abs{x} > \abs{y}} \tag*{\textcolor{gray}{\texttt{// calculuate sign of $z$}}}\\
	& \qquad \ASSIGN{z_{\mathit{sgn}}}{x_{\mathit{sgn}}} \\
	& \ELSE \\
	& \qquad \IF{\abs{x} = \abs{y}} \\
	& \qquad \qquad \ASSIGN{z_{\mathit{sgn}}}{0} \\
	& \qquad \ELSE \\
	& \qquad \qquad \ASSIGN{z_{\mathit{sgn}}}{1 \monus y_{\mathit{sgn}}} \\
	& \qquad \} \\
	& \}
\end{align*}%
Similar emulations can be performed for addition, multiplication, etc.
For the purpose of proving relative completeness, signed variables are thus syntactic sugar; we omit them for simplicity.

Our main reason for disallowing negative numbers as values of program variables is that we want~$x$ to be a valid (unsigned) expectation.
If $x$ was signed, it would not be a valid expectation as it does not map only to non-negative values.
%(cf. next subsection).
In order to fix this problem to some extent, one would have to \enquote{make $x$ non-negative}, \eg by instead using the expectation $\iverson{x \geq 0} \cdot x$ ($x$ truncated at $0$) or the expectation $|x|$ (absolute value of $x$; not supported (but can be encoded) in our current syntax). However, neither of the above expectations actually represents \enquote{the value of $x$}.

%% file: applications.tex
% !TEX root = ./main.tex

\section{Discussion}\label{sec:applications}
%\section{Applications}\label{sec:applications}

We now discuss a few aspects
%collect various  %scenarios 
in which our expressive language $\SyntE$ of expectations could be useful.

\subsection{Relative Completeness of Probabilistic Program Verification}
An immediate consequence of \Cref{thm:expressive} is that, for all $\pgcl$ programs $\cc$ and all syntactic expectations $\FF, \FG \in \SyntE$, verifying the bounds 
\begin{align*}
        \eval{\FG} \ppreceq \wp{\cc}{\eval{\FF}} 
        \qquad\text{or}\qquad
        \wp{\cc}{\eval{\FF}} \ppreceq \eval{\FG}
\end{align*}
reduces to \emph{checking a single inequality} between two syntactic expectations in $\SyntE$, namely $\FG$ and the \emph{effectively constructible expectation} for $\wp{\cc}{\eval{\FF}}$.
In that sense, the $\wpsymbol$ calculus together with $\SyntE$ form a \emph{relatively complete} (cf.~\cite{Cook1978SoundnessAC}) system for probabilistic program verification.
Given an oracle for discharging inequalities between syntactic expectations, 
every correct inequality of the above form can be derived.
%Checking inequalities of the form $\wp{\cc}{\eval{\FF}} \preceq \eval{\FG}$ is typically needed in invariant-style reasoning for probabilistic programs~\cite{benni_diss}.
%the correctness of every
%on expected values of $\pgcl$ programs can be derived. % with weakest preexpectations (for total correctness).
%Notice that, in the above inequalities, we did not translate $\FF$ into an expectation $\eval{\FF}$ before computing the weakest preexpectation. By \Cref{thm:expressive}, we know how to construct a suitable syntactic expectation $\wp{\cc}{\FF}$ without reverting to a semantical definition.

\subsection{Termination Probabilities}

For each probabilistic program $\cc$, the weakest preexpectation%
\begin{align*}
	\wp{\cc}{1}
\end{align*}%
is a mapping from initial state $\pstate$ to the \emph{probability that $C$ terminates on $\pstate$}.
Since $1 \in \SyntE$, \emph{termination probabilities of any \pgcl program on any input are expressible in our syntax}.

This demonstrates that our syntax is capable of capturing mappings from states to numbers that are \emph{far from trivial} as termination probabilities in general carry a \emph{high degree of internal complexity}~\cite{DBLP:conf/mfcs/KaminskiK15,acta19}.
More concretely, given $C$, $\pstate$, and $\alpha$, deciding whether $C$ terminates on $\pstate$ \emph{at least} with probability $\alpha$ is $\Sigma_1^0$--complete in the arithmetical hierarchy.
%, thus strictly harder than the (non-universal) termination problem for non-probabilistic programs.
Deciding whether $C$ terminates on $\pstate$ \emph{at most} with probability $\alpha$ is even $\Pi_2^0$--complete, thus strictly harder than, \eg the universal termination problem for non-probabilistic programs.

\subsection{Probability to Terminate in Some Postcondition}

For a probabilistic program $\cc$ and a first-order predicate $\iverson{\BB}$, the weakest preexpectation%
\begin{align*}
	\wp{\cc}{\iverson{\varphi}}
\end{align*}%
is a mapping from initial state $\pstate$ to the \emph{probability that $C$ terminates on $\pstate$ in a state $\tau \models \BB$}.
Since $\iverson{\varphi}$~is expressible in $\SyntE$, we have that $\wp{\cc}{\iverson{\varphi}}$ is also expressible in $\SyntE$ by expressivity of $\SyntE$. 
We can thus embed \emph{and generalize Dijkstra's weakest preconditions completely in our system}.

\subsection{Distribution over Final States}

Let $\cc$ be a probabilistic program in which only the variables $\XX_1,\, \ldots,\, \XX_k$ occur.
Moreover, let $\mu_\cc^\pstate$ be the final distribution obtained by executing $\cc$ on input $\pstate$, cf.~\Cref{sec:forward-semantics}.
Then, by the Kozen duality (cf.~\Cref{thm:kozen-duality}), we can express the probability $\mu_\cc^\pstate(\tau)$ of $\cc$ terminating in final state $\tau$ on initial state $\pstate$, where~\mbox{$\tau(\XX_i) = \XX_i'$}, by
\begin{align*}
	\mu_{\cc}^{\pstate}(\tau) \eeq \wp{\cc}{\iverson{\XX_1 = \XX_1' \wedge {\cdots} \wedge \XX_k = \XX_k'}}(\sigma)~.
\end{align*}
Intuitively, we can write the initial values of $\XX_1,\, \ldots,\, \XX_k$ into $\sigma(\XX_1),\, \ldots,\, \sigma(\XX_k)$ and the final values into $\sigma(\XX_1'),\, \ldots,\, \sigma(\XX_k')$.

Since $\iverson{\XX_1 = \XX_1' \wedge {\cdots} \wedge \XX_k = \XX_k'} \in \SyntE$, we have that $\wp{\cc}{\iverson{\XX_1 = \XX_1' \wedge {\cdots} \wedge \XX_k = \XX_k'}}$ is expressible in $\SyntE$ as well.
Hence, \emph{we can express Kozen's measure transformers in our syntax}.

\subsection{Ranking Functions / Supermartingales}

There is a plethora of methods for proving termination of probabilistic programs based on ranking supermartingales~\cite{DBLP:conf/cav/ChakarovS13,DBLP:conf/popl/FioritiH15,DBLP:conf/popl/ChatterjeeFNH16,DBLP:conf/popl/ChatterjeeNZ17,DBLP:conf/aplas/HuangFC18,DBLP:conf/vmcai/FuC19,DBLP:journals/pacmpl/Huang0CG19}.
Ranking supermartingales are similar to ranking functions, but one requires that the value decreases \emph{in expectation}.
Weakest preexpectations are the natural formalism to reason about this.

For algorithmic solutions, ranking supermartingales are often assumed to be, for instance, linear~\cite{DBLP:journals/toplas/ChatterjeeFNH18} or polynomial~\cite{DBLP:conf/cav/ChatterjeeFG16,DBLP:journals/corr/abs-1910-12634,DBLP:conf/pldi/NgoC018}.
This also applies to the allowed shape of templates for loop invariants in works~\cite{DBLP:conf/sas/KatoenMMM10,DBLP:conf/atva/FengZJZX17} on the automated synthesis of probabilistic loop invariants.
\emph{Functions linear or polynomial in the program variables are obviously subsumed by our syntax.}
However, our syntax now enables searching for \emph{wider} tractable classes. % of ranking supermartingales.

\subsection{Harmonic Numbers}

Harmonic numbers are ubiquitous in reasoning about expected values or expected runtimes of randomized algorithms.
They appear, for instance, as the expected runtime of Hoare's randomized quicksort or the coupon collector problem, or as ranking functions for proving almost-sure termination~\cite{benni_diss,OlmedoKKM16,DBLP:journals/jacm/KaminskiKMO18,DBLP:journals/pacmpl/McIverMKK18}.
Harmonic numbers are syntactically expressible in our language as in \Cref{ex:harmonic}, or more conveniently as
\begin{align*}
	H_\XX \eeq \eval{\gsum{\tfrac{1}{\vsum}}{\XX}}~, \qquad \textnormal{where }~ \tfrac{1}{\vsum} \eeq \Sup \XZ \colon \iverson{\XZ \cdot \vsum = 1}\cdot \XZ~.
\end{align*}%
We note that, in termination proofs, the Harmonic numbers do \emph{not} occur as termination probabilities, but rather \emph{in ranking functions} whose expected values after one loop iteration need to be determined.
Our syntax is capable of handling such ranking functions and we could safely add $H_x$ to our syntax.

%% file: conclusion.tex
% !TEX root = ./main.tex

\section{Conclusion and Future Work}\label{sec:conclusion}

We have presented a \emph{language of syntactic expectations} that is \emph{expressive for weakest preexpectations} of probabilistic programs \'{a} la \citet{Kozen1985} and \citet{McIverM05}.
As a consequence, verification of bounds on expected values of functions (expressible in our language) after probabilistic program execution is \emph{relative complete} in the sense of~\citet{Cook1978SoundnessAC}.

We have discussed various scenarios covered by our language, such as reasoning about termination probabilities, thus demonstrating the language's usefulness.

\subsubsection*{Future Work}
We currently do not support probabilistic programs with (binary) \emph{non-deterministic} choices, as do \citet{McIverM05}, and it is not obvious how to incorporate it, given our current encoding.
What seems even more out of reach is handling \emph{unbounded non-determinism}, which would be needed, for instance, to come up with an expressive expectation language for \emph{quantitative separation logic} (\QSL)---an (extensional) verification system for compositional reasoning about probabilistic pointer programs with access to a heap~\cite{DBLP:journals/pacmpl/BatzKKMN19,Matheja20}.

For non-probabilistic heap-manipulating programs, a topic considered by~\citet{expressiveness_sl} are inductive definitions of predicates in classical separation logic (\SL) and proving that  \SL is expressive in this context. 
\QSL also features inductive definitions and it would be interesting to consider expressiveness in this setting.
%For \SL, only recently the question of expressiveness was addressed by~\citet{DBLP:journals/tcs/AmeenT16} in the context of \emph{recursive procedure calls}.
%It would be interesting future work to consider expressiveness of $\SyntE$ for recursive and higher order probabilistic programs~\cite{OlmedoKKM16,DBLP:journals/pacmpl/DahlqvistK20,DBLP:conf/pldi/GehrSV20,DBLP:journals/pacmpl/SatoABGGH19,DBLP:conf/lics/KobayashiLG19,DBLP:conf/lics/StatonYWHK16}.
%
%For \SL,~\citet{expressiveness_sl} address the \emph{complexity of finding loop invariants}, using their expressiveness result.
%In the context of probabilistic loop invariants~\cite{DBLP:conf/qest/GretzKM13,DBLP:conf/sas/KatoenMMM10,benni_diss}, this is an interesting question as well.

Despite its similarity to the $\wpsymbol$ calculus, we did not consider the \emph{expected runtime calculus} ($\textsf{ert}$) by \citet{DBLP:journals/jacm/KaminskiKMO18}.
We strongly conjecture that $\SyntE$ is expressive for expected \mbox{runtimes as well}.

Finally, the \emph{conditional weakest preexpectation} calculus (\textsf{cwp})~\cite{benni_diss,DBLP:journals/toplas/OlmedoGJKKM18} for probabilistic programs with \emph{conditioning} needs weakest \emph{liberal} preexpectations, which generalize Dijkstra's weakest liberal preconditions.
It currently remains open, whether $\wlp{\cc}{f}$ is expressible in $\SyntE$.
There is the duality $\wlp{\cc}{f} = 1 - \wp{\cc}{1{-}f}$, originally due to \citet{Kozen1983}, but it is not immediate how to express $1{-}f$ in $\SyntE$, if $f$ is not a plain arithmetic expression.

%% file: appendix_normalforms.tex
\subsection{Proof of Lemma~\ref{thm:prenex-rules}}

\begin{lemma}
	\label{thm:prenex:aux}
	Let $\alpha \in \PosReals$ and $A,B \subseteq \PosRealsInf$. Then, we have:
	\begin{enumerate}
		\item\label{thm:prenex:aux:sup-mult}
		\makebox[2.8cm]{$\alpha \cdot \sup A$} $\eeq \sup \{ \alpha \cdot a \mmid a \in A \}$,
		\item\label{thm:prenex:aux:inf-mult}
		\makebox[2.8cm]{$\alpha \cdot \inf A$} 
		$\eeq \inf \{ \alpha \cdot a \mmid a \in A \}$,
		\item\label{thm:prenex:aux:sup-plus}
		\makebox[2.8cm]{$(\sup A) + (\sup B)$} 
		$\eeq \sup \{ \beta + \gamma \mmid \beta \in A, \gamma \in B \}$, 
		\item\label{thm:prenex:aux:inf-plus}
		\makebox[2.8cm]{$(\inf A) + (\inf B)$} 
		$\eeq \inf \{ \beta + \gamma \mmid \beta \in A, \gamma \in B \}$, and
		\item\label{thm:prenex:aux:singleton}
		if $A$ is a singleton, i.e.,  $A = \{ \beta \}$, then $\sup A = \inf A = \beta$,
	\end{enumerate}
	where we define $0 \cdot \infty = 0$.
\end{lemma}

\emph{Proof of \Cref{thm:prenex-rules}}.
	\label{proof:prenex-rules}
	By definition of equivalence between expectations, we have 
	\begin{align*}
	\FF_1 \equiv \FF_2 
	\qiff 
	\text{for all $\sigma$}\colon \sem{\FF_1}{\sigma}{\interpret} \eeq \sem{\FF_2}{\sigma}{\interpret}.
	\end{align*}
	Let us fix an arbitrary state $\sigma$.
	
	To prove Lemma~\ref{thm:prenex-rules}~(\ref{thm:prenex-rules:plus-left}) for $\Quant = \Sup$, we proceed as follows:
	\begin{align*}
	& \sem{(\SupV{\VV} \FF_1) \pplus \FF_2}{\sigma}{\interpret} \\
	\eeq & 
	\sup~\setcomp{\sem{\FF_1}{\pstate\statesubst{\VV}{\RR}}{\interpret\statesubst{\VV}{\RR}}}{\RR \in \PosRats} 
	\pplus \sem{\FF_2}{\pstate}{\interpret} 
	\tag{Semantics of expecations} \\
	\eeq & 
	\sup~\setcomp{\sem{\FF_1}{\pstate\statesubst{\VV}{\RR}}{\interpret\statesubst{\VV}{\RR}}}{\RR \in \PosRats} 
	\pplus 
	\sup~\left\{\sem{\FF_2}{\pstate\statesubst{\VV}{\RR}}{\interpret}\right\}
	\tag{Lemma~\ref{thm:prenex:aux}~(\ref{thm:prenex:aux:singleton})} \\
	\eeq & 
	\sup~\setcomp{\sem{\FF_1\subst{\VV}{\VV'}}{\pstate\statesubst{\VV'}{\RR}}{\interpret\statesubst{\VV'}{\RR}} + \sem{\FF_2}{\pstate}{\interpret\statesubst{\VV'}{\RR}}}{\RR \in \PosRats} 
	\tag{Lemma~\ref{thm:prenex:aux}~(\ref{thm:prenex:aux:sup-plus}), $\VV'$ fresh} \\
	\eeq & 
	\sem{\SupV{\VV'} \FF_1\subst{\VV}{\VV'} \pplus \FF_2 }{\sigma}{\interpret}.
	\tag{Semantics of expectation}
	\end{align*}
	The proofs for $\Quant = \Inf$ as well as the proof of Lemma~\ref{thm:prenex-rules}~(\ref{thm:prenex-rules:plus-right}) are completely analogous.
	To prove Lemma~\ref{thm:prenex-rules}~(\ref{thm:prenex-rules:mult-term}) for $\Quant = \Inf$, we proceed as follows:
	\begin{align*}
	& \sem{\TT \cdot \InfV{\VV} \FF}{\sigma}{\interpret} \\
	\eeq & 
	\sem{\TT}{\sigma}{\interpret} \cdot \inf \setcomp{ \sem{\FF}{\sigma\statesubst{\VV}{\RR}}{\interpret\statesubst{\VV}{\RR}} }{ \RR \in \PosRats }
	\tag{Semantics of expectations} \\
	\eeq & 
	\inf \setcomp{\sem{\TT}{\sigma}{\interpret} \cdot \sem{\FF}{\sigma\statesubst{\VV}{\RR}}{\interpret\statesubst{\VV}{\RR}} }{ \RR \in \PosRats }
	\tag{Lemma~\ref{thm:prenex:aux}~(\ref{thm:prenex:aux:inf-mult})} \\
	\eeq & 
	\inf \setcomp{\sem{\TT}{\sigma}{\interpret\statesubst{\VV'}{\RR}} \cdot \sem{\FF\subst{\VV}{\VV'}}{\sigma\statesubst{\VV'}{\RR}}{\interpret\statesubst{\VV'}{\RR}} }{ \RR \in \PosRats }
	\tag{$\VV'$ fresh} \\
	\eeq & 
	\sem{\InfV{\VV'} \TT \cdot \FF\subst{\VV}{\VV'}}{\sigma}{\interpret} 
	\tag{Semantics of expectations}.
	\end{align*}
	The proofs for $\Quant = \Sup$ as well as the proof of Lemma~\ref{thm:prenex-rules}~(\ref{thm:prenex-rules:mult-guard}) are completely analogous.

\subsection{Proof of Theorem~\ref{thm:dedekind_nf}}

First, we prove that every $\FF \in\SyntE$ is equivalent to some expectation in \emph{summation normal form}. For that, we employ an auxiliary result:

\begin{lemma}
	\label{lem:quantifier_free_smf}
	Let $\FF \in \SyntE$ be quantifier-free. Then there exist (1) a natural number $n\geq 1$, (2) Boolean expressions $\BB_1,\ldots,\BB_n$, and (3) terms $\TT_1,\ldots,\TT_n$ such that $\FF$ is equivalent to an expectation \mbox{$\FF'$ given by}
	\[
	\FF' \eeq \sum\limits_{i=1}^{n} \iverson{\BB_i}\cdot \TT_i ~,
	\]
	where the above sum is a shorthand for $\iverson{\BB_1}\cdot \TT_1 + \ldots + \iverson{\BB_n}\cdot \TT_n$.
\end{lemma}

\begin{proof}
	
	By induction on the structure of quantifier-free syntactic expectations. \\ \\
	\noindent
	\emph{Base case $\FF = \TT$.} The expectation $\FF$ is obviously equivalent to 
	\[
	\FF' \eeq \sum\limits_{i=1}^{1} \iverson{\true}\cdot \TT ~.
	\]
	As the induction hypothesis now assume that for some arbitrary, but fixed, quantifier-free syntactic expectations $\FF_1$ and $\FF_2$
	there are expectations $\FF_1'$ and $\FF_2'$ equivalent to $\FF_1$ and $\FF_2$, respectively, given by
	\begin{align*}
	\FF_1' &\eeq \sum\limits_{i=1}^{n} \iverson{\BB_i}\cdot \TT_i~,~\text{and} \\
	\FF_2' &\eeq \sum\limits_{i=1}^{m} \iverson{\BB_i'}\cdot \TT_i' ~.
	\end{align*}
	\noindent
	\emph{The case $\FF = \TT \cdot \FF_1$.} We have
	\begin{align*}
	&\TT \cdot \FF_1 \\
	\eequiv & \TT \cdot \sum\limits_{i=1}^{n} \iverson{\BB_i}\cdot \TT_i
	\tag{by I.H.} \\
	\eequiv &  \sum\limits_{i=1}^{n} \iverson{\BB_i} \cdot \TT \cdot \TT_i
	\tag{$\cdot$ distributes over $+$ in the quantifier-free setting} \\
	\eequiv &  \sum\limits_{i=1}^{n} \iverson{\BB_i} \cdot u_i ~.
	\tag{let $u_i = \TT \cdot \TT_i$}
	\end{align*}
The case $\FF = \FF_1 \cdot \TT$ is completely analogous. \\ \\
	%\kbcommentinline{TODO: Prove distributivity of $\cdot$ over %$+$ in the quantifier-free setting}
	%\cmcommentinline{That should be immediate since it just uses %pointwise distributitivity and the fact that we only talk about %finite sums?}
	%\kbcommentinline{Yes}
	%
	\noindent
	\emph{The case $\FF = \iverson{\BB} \cdot \FF_1$.} We have
	\begin{align*}
	&\iverson{\BB} \cdot \FF_1 \\
	\eequiv & \iverson{\BB} \cdot \sum\limits_{i=1}^{n} \iverson{\BB_i}\cdot \TT_i
	\tag{by I.H.} \\
	\eequiv &  \sum\limits_{i=1}^{n} \iverson{\BB} \cdot  \iverson{\BB_i} \cdot \TT_i
	\tag{$\cdot$ distributes over $+$ in the quantifier-free setting} \\
	\eequiv &  \sum\limits_{i=1}^{n} \iverson{G_i} \cdot  \TT_i ~.
	\tag{let $G_i = \BB \wedge \BB_i$}
	\end{align*}
The case $\FF = \FF_1 \cdot \iverson{\BB}$ is completely analogous. \\ \\
	\noindent
	\emph{The case $\FF = \FF_1 + \FF_2$.} This case is trivial since
	\begin{align*}
	&\FF_1 + \FF_2
	\eequiv  \FF_1' + \FF_2'~,
	\tag{by I.H.}
	\end{align*}
	where $\FF_1' + \FF_2'$ is of the desired form.
	This completes the proof.
\end{proof}

\begin{theorem}[Summation Normal Form]
	\label{thm:summation_nf}
	Every syntactic expectation $\FF$ is equivalent to an expectation $\FF'$ in \emph{summation normal form}, i.e.\
	$\FF'$ is of the form
	\[
	\FF' \eeq \Quant_1 \VV_1 \ldots \Quant_k \VV_k \colon \sum\limits_{i=1}^{n} \iverson{\BB_i}\cdot \TT_i ~.
	\]
\end{theorem}
\begin{proof}
	\label{proof:summation_nf}
	By Lemma~\ref{thm:prenex-rules}, \FF is equivalent to an expectation in prenex normal form, i.e.,
	\[
	\FF \equiv \Quant_1 \VV_1 \ldots \Quant_k \VV_k \colon g~,
	\]
	where $g$ is quantifier-free. 
	By Lemma~\ref{lem:quantifier_free_smf}, $g$ is then equivalent to an expectation of the form
	\[
	\sum\limits_{i=1}^{n} \iverson{\BB_i}\cdot \TT_i~.
	\]
	Hence, \FF is equivalent to the following expectation in summation normal form:
	\[
	\FF' \eeq \Quant_1 \VV_1 \ldots \Quant_k \VV_k \colon \sum\limits_{i=1}^{n} \iverson{\BB_i}\cdot \TT_i~.
	\]
\end{proof}

Towards the construction of the Dedekind normal form of an expectation, we first show that the Dedekind cut of an expectation $\FF$ is definable in $\FOArithPosRats$.
\begin{definition}
	Let $\FF \in \SyntE$ be in summation normal form, say
	\[
	\FF \eeq \Quant_1 \VV_1 \ldots \Quant_k \VV_k \colon \sum\limits_{i=1}^{n} \iverson{\BB_i}\cdot \TT_i~.
	\]
	Now let $\VVcut$ be a fresh variable. We construct a formula $\DD{\FF}$ inductively on $k$ as follows:
	\begin{enumerate}
		\item If $k=0$ and $\FF= \sum\limits_{i=1}^{n} \iverson{\BB_i}\cdot \TT_i$, then
		\[
		\DD{\FF} \eeq \bigwedge_{\big( (B_i, T_i)_{1\leq i \leq n} \big)  \in \bigtimes_{i = 1}^{n} \big\{ (\BB_i, \TT_i), (\neg \BB_i, 0) \big\}}
		\Big( \bigwedge_{i=1}^{n} B_i \longrightarrow \VVcut < \sum_{i=1}^{n} T_i \Big) ~.
		\]
		\item If $k>0$ and $\FF= \Sup \VV \colon \FF'$, then 
		\[
		\DD{\FF} \eeq \exists \VV \colon \DD{\FF'}~.
		\]
		\item If $k>0$ and $\FF= \Inf \VV \colon \FF'$, let $\VV'$ be a fresh variable and define
		\[
		\DD{\FF} \eeq \exists \VV'\colon \VV'>\VVcut \wedge \forall \VV\colon \DD{\FF'}\subst{\VVcut}{\VV'}~.
		\]
	\end{enumerate}
\end{definition}
\begin{lemma}
	\label{lem:dcut_in_forats}
	For every $\FF =\Quant_1 \VV_1 \ldots \Quant_k \VV_k \colon \sum\limits_{i=1}^{n} \iverson{\BB_i}\cdot \TT_i$ in summation normal form and \mbox{all states $\sigma$, }
	\[
	   \sem{\DD{\FF}}{\sigma}{} \eeq \true \qquad \text{iff}\qquad \sigma(\VVcut) \in \dcut{\sem{\FF}{\sigma}{}}~.
	\]
\end{lemma}
\begin{proof}
	By induction on $k$. We fix an arbitrary state $\pstate$. \\ \\
	\noindent
	\emph{Base case $k=0$.} There is \emph{exactly one} $((B_1', T_1'), \ldots, (B_n', T_n')) \in \bigtimes_{i = 1}^{n} \{ (\BB_i, \TT_i), (\neg \BB_i, 0) \}$ such that
	\[
	\sem{\bigwedge_{i=1}^{n} B_i'}{\sigma}{\interpret} \eeq \true~.
	\]
	Hence, we have
	\begin{align*}
		& \sem{f}{\sigma}{\interpret} \\
		\eeq & \sem{\sum\limits_{i=1}^{n} \iverson{\BB_i}\cdot \TT_i}{\sigma}{\interpret}
		\tag{by definition} \\
		\eeq & \sem{\sum\limits_{i=1}^{n} T_i'}{\sigma}{\interpret}~.
		\tag{since $\sem{\iverson{\BB_j}}{\sigma}{\interpret}$ = 0 if $B_j = \neg \BB_j$}
	\end{align*}
	This gives us
	\begin{align*}
		&\sem{\DD{\FF}}{\sigma}{\interpret}  \eeq \true \\
		\text{iff} \quad & \sem{\bigwedge_{((B_1, T_1), \ldots, (B_n, T_n)) \in \bigtimes_{i = 1}^{n} \{ (\BB_i, \TT_i), (\neg \BB_i, 0) \}}
				\big( \bigwedge_{i=1}^{n} B_i \longrightarrow \VVcut < \sum_{i=1}^{n} T_i \big)} {\sigma}{\interpret} \eeq \true 
		\tag{by definition}\\
		\text{iff} \quad & \sem{\bigwedge_{i=1}^{n} B_i' \longrightarrow \VVcut < \sum_{i=1}^{n} T_i' }{\sigma}{\interpret} \eeq \true 
		\tag{by above reasoning}\\
		\text{iff} \quad & \sem{\VVcut < \sum_{i=1}^{n} T_i' }{\sigma}{\interpret} \eeq \true 
		\tag{left-hand side of implication evaluates to $\true$}\\
		\text{iff} \quad & \sigma(\VVcut) \LL \sem{\sum_{i=1}^{n} T_i'}{\sigma}{\interpret} 
		\tag{by definition}\\
		\text{iff} \quad & \sigma(\VVcut) \LL \sem{f}{\sigma}{\interpret}~. 
		\tag{by above reasoning}\\
	\end{align*}

As the induction hypothesis now assume that for some arbitrary, but fixed expectation $\FF'$ in summation normal form and all states $\sigma$, we have 
\[
\sem{\DD{\FF'}}{\sigma}{} \eeq \true \qquad \text{iff}\qquad \sigma(\VVcut) \in \dcut{\sem{\FF'}{\sigma}{}}~.
\]
\emph{Induction step for $\FF \eeq \Sup \VV \colon \FF'$.} We have
\begin{align*}
		& \sem{\DD{\FF}}{\sigma}{\interpret} \eeq  \true \\
		\text{iff}\quad & \sem{\exists \VV \colon \DD{\FF'}}{\sigma}{\interpret} \eeq  \true
		\tag{by definition} \\
		\text{iff}\quad & \text{there is $\RR \in \PosRats$ with}~\sem{\DD{\FF'}}{\sigma\statesubst{\VV}{\RR}}{\interpret} \eeq  \true
		\tag{by definition} \\
		\text{iff}\quad & \text{there is $\RR \in \PosRats$ with}~\sigma\statesubst{\VV}{\RR}(\VVcut) \in \dcut{\sem{\FF'}{\sigma\statesubst{\VV}{\RR}}{}}
		\tag{by I.H.}\\
		\text{iff}\quad & \text{there is $\RR \in \PosRats$ with}~\sigma(\VVcut) < \sem{\FF'}{\sigma\statesubst{\VV}{\RR}}{}
		\tag{by definition and $\VV\neq \VVcut$ by construction}\\
		\text{iff}\quad & \sigma(\VVcut) < \sup \big\{\sem{\FF'}{\sigma\statesubst{\VV}{\RR}}{} ~\mid~ \RR \in \PosRats \big\}
		\tag{see below}\\
		\text{iff}\quad & \sigma(\VVcut) < \sem{\Sup \VV \colon \FF'}{\sigma}{} = \sem{\FF}{\sigma}{}~.
		\tag{by definition}\\
\end{align*}
We justify the last but one step as follows. The \enquote{only if}-direction holds since the supremum evaluates at least to $\sem{\FF'}{\sigma\statesubst{\VV}{\RR}}{}$. For the \emph{if}-direction, assume for a contradiction that
\begin{align}
	\label{eq:proof_dedekind_nf_1}
	\sigma(\VVcut) < \sup \big\{\sem{\FF'}{\sigma\statesubst{\VV}{\RR}}{} ~\mid~ \RR \in \PosRats \big\}
\end{align}
and
\begin{align}
	\label{eq:proof_dedekind_nf_2}
	\text{for all}~\RR \in \PosRats~\text{it holds that}~
	\sigma(\VVcut) \geq \sem{\FF'}{\sigma\statesubst{\VV}{\RR}}{}~.
\end{align}
Inequality~(\ref{eq:proof_dedekind_nf_2}) implies that $\sigma(\VVcut)$ is an upper bound on
$ \big\{\sem{\FF'}{\sigma\statesubst{\VV}{\RR}}{} ~\mid~ \RR \in \PosRats \big\}$. Hence, $\sigma(\VVcut)$ is greater than or equal to the \emph{least} upper bound of this set. This contradicts inequality~(\ref{eq:proof_dedekind_nf_1}). \\ \\
\noindent
\emph{Induction step for $\FF \eeq \Inf \VV \colon \FF'$.} We have
\begin{align*}
	& \sem{\DD{\FF}}{\sigma}{\interpret} \eeq  \true \\
	\text{iff}\quad & \sem{\exists \VV'\colon \VV'>\VVcut \wedge \forall \VV\colon \DD{\FF'}\subst{\VVcut}{\VV'}}{\sigma}{\interpret} \eeq  \true 
	\tag{by definition}\\
	\text{iff}\quad & \text{there is $\RRb \in \PosRats$ with $\RRb > \sigma(\VVcut)$ such that for all $\RR \in \PosRats$,}~ \\
	& \qquad \sem{\DD{\FF'}\subst{\VVcut}{\VV'}}{\sigma\statesubst{\VV'}{\RRb}\statesubst{\VV}{\RR}}{\interpret} \eeq  \true 
	\tag{by definition}\\
	\text{iff}\quad & \text{there is $\RRb \in \PosRats$ with $\RRb > \sigma(\VVcut)$ such that for all $\RR \in \PosRats$,}~ \\
	&\qquad \sigma\statesubst{\VV'}{\RRb}\statesubst{\VV}{\RR}(\VV') \in \dcut{\sem{\FF'}{\sigma\statesubst{\VV'}{\RRb}\statesubst{\VV}{\RR}}{}}
	\tag{by I.H.}\\
	\text{iff}\quad & \text{there is $\RRb \in \PosRats$ with $\RRb > \sigma(\VVcut)$ such that for all $\RR \in \PosRats$,}~  \RRb < \sem{\FF'}{\statesubst{\VV}{\RR}}{}
	\tag{$\VV'$ and $\VVcut$ are fresh}\\
	\text{iff}\quad & \text{there is $\RRb \in \PosRats$ with $\RRb > \sigma(\VVcut)$ such that}~  \RRb \leq \inf \big\{ \sem{\FF'}{\statesubst{\VV}{\RR}}{} ~\mid~\RR \in \PosRats \}
	\tag{$\dagger$, see below} \\
	\text{iff} \quad & \sigma(\VVcut) < \inf \big\{ \sem{\FF'}{\statesubst{\VV}{\RR}}{} ~\mid~\RR \in \PosRats \}
	\tag{\enquote{only if}: by tranisitivity, \enquote{if}: choose $s=\inf \ldots$}
	\\
	\text{iff}\quad &  \sigma(\VVcut) < \sem{\Inf \VV \colon \FF'}{\sigma}{} = \sem{\FF}{\sigma}{}~.
\end{align*}
We justify the step marked with $\dagger$ as follows. For the \enquote{only if}-direction, observe that if $\RRb$ is a strict lower bound on $\big\{ \sem{\FF}{\sigma\statesubst{\VV}{\RR}}{} ~\mid~\RR \in \PosRats \}$, then $\RRb$ is at least the \emph{greatest} lower bound, i.e., the infimum, of this set. For the \enquote{if}-direction, first observe that $\RRb \neq 0$ since $\sigma(\VVcut) \geq 0$. Since $\PosRats$ is dense,  there is an $\RRb'$ with $\sigma(\VVcut) < \RRb' < \RRb$. Hence, we have
\begin{align*}
	 & \text{$\RRb > \sigma(\VVcut)$ and}~  \RRb \leq \inf \big\{ \sem{\FF'}{\statesubst{\VV}{\RR}}{} ~\mid~\RR \in \PosRats \} \\
	 \text{implies} \quad &\text{$\RRb > \sigma(\VVcut)$ and for all $\RR \in \PosRats$,}~  \RRb \leq  \sem{\FF'}{\statesubst{\VV}{\RR}}{} \\
	 \text{implies} \quad &\text{$\RRb' > \sigma(\VVcut)$ and for all $\RR \in \PosRats$,}~  \RRb' <  \sem{\FF'}{\statesubst{\VV}{\RR}}{}~,
\end{align*}
which is what we had to show.

\end{proof}

\begin{definition}[Dedekind Normal Form]
	Let $\FF$ be an expectation in summation normal form. The \emph{Dedekind normal form} $\dexp{\FF}$ of \FF w.r.t. the fresh variable $\VVcut$ is given by:
	\[
	   \dexp{\FF} \eeq \iverson{\DD{\FF}}~,
	\]
	where is $\iverson{\DD{\FF}}$ constructed according to \Cref{table:foposrat_to_synte}, and where we call $\VVcut$ the \emph{cut variable} of $\dexp{\FF}$.
\end{definition}
Soundness of the Dedekind normal form immediately follows from \Cref{thm:exp_subsumes_fo_rats} and \Cref{lem:dcut_in_forats}.
Notice that, by Theorems \ref{thm:exp_subsumes_fo_rats} and \ref{thm:summation_nf}, $\dexp{\FF}$ is equivalent to an expectation of the form $\qprefixnoarg\colon \iverson{\BB}$ for some quantifier prefix $\qprefixnoarg$ and some Boolean expression $\BB$. We thus often write $\dexp{\FF} \eeq \qprefixnoarg\colon \iverson{\BB}$.

\subsection{Proof of Lemma \ref{lem:dedekind_nf_recover}}

\begin{proof}
	\label{proof:dedekind_nf_recover}
		Let $\sigma$ be a state. We have
		\begin{align*}
		& \sem{\Sup \VVcut \colon \dexp{f} \cdot \VVcut}{\sigma}{\interpret} \\
		\eeq & \sup \setcomp{   \sem{\dexp{f} \cdot \VVcut}{\sigma\statesubst{\VVcut}{\RR}}{\interpret\statesubst{\VVcut}{\RR}}   }{\RR \in \PosRats} 
		\tag{by definition}\\
		\eeq & \sup \setcomp{\sem{\VVcut}{\sigma\statesubst{\VVcut}{\RR}}{\interpret\statesubst{\VVcut}{\RR}}}{ \RR \in \dcut{\sem{\FF}{\sigma}{\interpret}}}
		\tag{since $\sem{\dexp{\FF}}{\sigma\statesubst{\VVcut}{\RR}}{\interpret\statesubst{\VVcut}{\RR}} = 0$ if $\RR \not \in \dcut{\sem{\FF}{\sigma}{\interpret}}$ } \\
		\eeq & \sup \setcomp{\RR \in \PosRats}{\RR \in \dcut{\sem{\FF}{\sigma}{\interpret}}} 
		\tag{by definition}\\
		\eeq & \sem{\FF}{\sigma}{\interpret}~.
		\tag{by definition}
		\end{align*}
	
\end{proof}

%% file: appendix_embedding.tex
\subsection{Proof of Lemma~\ref{lem:nats_definable}}

\begin{proof}
	\label{proof:nats_definable}
	
	We employ a result by Robinson~\cite[Section 3]{robinson_define_z}: For $a,b,k \in \Rats$, let
	\begin{align*}
	\Phi(a, b, k) &\ddefeq \exists x,y,z \in \Rats \colon  2 + abk^2 + bz^2 = x^2 + ay^2 \\
	B(a,b) &\ddefeq  \Phi(a,b,0) \wedge \forall m \in \Rats \colon (\Phi(a,b,m) \longrightarrow \Phi(a,b,m+1)) \\
	A(k)  &\ddefeq  \forall a,b \in \Rats \colon 
	B(a,b) \longrightarrow \Phi(a,b,k)~.
	\end{align*}
	Then $k \in \Rats$ is an integer if and only if $A(k)$ holds. Denote by $A'$ (resp. $\Phi', B'$) the formula obtained from
	$A$ (resp. $\Phi, B$) by replacing every occurrence of $\Rats$ by $\PosRats$, i.e.\ we restrict to quantification over $\PosRats$. Note that $\Phi(a',b',k')$ iff $\Phi'(a',b',k')$ for all $a',b',k' \in \PosRats$ since all occurrences of $x,y,z$ in $B$ are squared.
	
	Since $A'(k')$ is expressible in $\FOArithPosRats$, we prove the lemma by showing that for every $k' \in \PosRats$, 
	\[
	k' \in \Nats \qquad \text{if and only if} \qquad A'(k')~.
	\]
	The \enquote{only if} direction is straightforward since
	\[
	\Phi'(a,b,0) \wedge \forall m \in \PosRats \colon (\Phi'(a,b,m) \longrightarrow \Phi'(a,b,m+1))
	\]
	implies $\Phi(a,b,k')$ for all $k' \in  \Nats$. Hence, if $a,b \in \PosRats$ and $B'(a,b)$ holds, then $\Phi'(a,b,k')$ holds, which implies $A'(k')$.  \\ \\
	\noindent
	The \enquote{if} direction is less obvious. We proceed by recapping the crucial parts of Robinson's proof that $A(k)$ implies $k \in \Ints$ for all $k \in \Rats$ on a sufficient level of abstraction. We then show how to employ the same proof to show that $A'(k')$ implies $k' \in \Nats$ for all $k' \in \PosRats$.
	
	Robinson shows that it suffices to derive the following two facts from assumption $A(k)$:
	\begin{align}
	&\Phi(1, p, k)~\text{holds for all primes $p \in P_1$} 
	\label{eqn:proof_nats_definable_1}\\
	&\Phi(q,p,k)~\text{holds for all primes $p \in P_2 $ and all $q \in Q $}~,
	\label{eqn:proof_nats_definable_2}
	\end{align}
	where $P_1, P_2, Q \subseteq \Nats$ are some non-empty sets of primes. We do not give these sets explicitly here since they are not relevant for this proof.
	(\ref{eqn:proof_nats_definable_1}) and (\ref{eqn:proof_nats_definable_2}) in conjunction imply $k \in \Ints$. Now, since $\Phi(a',b',k')$ and $\Phi'(a',b',k')$ are equivalent for all $a',b',k' \in \PosRats$, our proof obligation is to show that for every $k' \in \PosRats$, $A'(k')$ implies:
	\begin{align}
	&\Phi'(1, p, k')~\text{holds for all primes $p \in P_1$}~,\text{and}
	\label{eqn:proof_nats_definable_3}\\
	&\Phi'(q,p,k')~\text{holds for all primes $p \in P_2 $ and all $q \in Q $}~.
	\label{eqn:proof_nats_definable_4}
	\end{align}
	We may then invoke Robinson's result from above to conclude that $k' \in \Nats$.

	To prove~(\ref{eqn:proof_nats_definable_1}), Robinson shows that $B(1, p)$ holds for every $p \in P_1$. Since $B(1,p)$ implies $B'(1,p)$ for all $p \in P_1$ (recall that $p$ is a natural number), we also get $B'(1,p)$. Now, assumption $A(k)$ and the fact that $B(1,p)$ holds imply $\Phi(1,p,k)$. We apply the same reasoning for $\Phi'(1,p,k)$: Assumption $A'(k)$ and the fact that $B'(1,p)$ holds imply $\Phi'(1,p,k')$, which proves~(\ref{eqn:proof_nats_definable_3}).
	
	The proof of~(\ref{eqn:proof_nats_definable_4}) is completely analogous.
	
\end{proof}

\subsection{Proof of Theorem~\ref{thm:fo_rats_subsumes_fo_nats}}

\begin{proof}
	\label{proof:fo_rats_subsumes_fo_nats}
	By induction on the structure of $\PP$ and by using Lemma~\ref{lem:nats_definable}. \\ \\
	\emph{Base case $\PP =\BB$.} If there is $\VV \in \FV{\PP}$ with $\pstate(\VV) \not\in \Nats$, then $\sem{\isNat(\VV)}{\sigma}{\interpret} = \false$ and thus $\sem{\toFOPosRats{\PP}}{\sigma}{\interpret} = \false$. Conversely, if for all $\VV \in \FV{\PP}$ it holds that $\sigma(\VV) \in \Nats$, then $\sigma$ is an interpretation for $\PP$ and obviously $\sem{\toFOPosRats{\PP}}{\sigma}{\interpret} = \sem{\PP}{\sigma}{\interpret}$ since $\sem{\isNat(\VV)}{\sigma}{\interpret} = \true$. \\ \\
	\noindent As the induction hypothesis (I.H.) now assume that the theorem holds for some arbitrary, but fixed, $\PP' \in \FOArithNats$. \\ \\
	\emph{The case $\PP = \exists \VV \colon \PP'$.} First notice that $\FV{\PP} = \FV{\PP'} \setminus \{ \VV \}$.
	Hence, if there is $\VV' \in \FV{\PP}$ with $\sigma(\VV') \not\in \Nats$, then $\sem{\toFOPosRats{\PP}}{\sigma}{\interpret} = \sem{\toFOPosRats{\PP'}}{\sigma}{\interpret}  = \false$ by I.H. Now assume that for all $\VV' \in \FV{\PP}$ it holds that $\sigma(\VV') \in \Nats$, rendering $\sigma$ an interpretation for $\PP$. We have
	%	If $v \not\in \FV{P'}$, then our induction hypothesis yields
	%	%
	%	\begin{align*}
	%	   & \sem{\toFOPosRats{P}}{}{\interpret} \\
	%	   %
	%	   \eeq & \sem{\exists \VV \colon \left( \toFOPosRats{P'} \right)}{}{\interpret}
	%	   \tag{by definition} \\
	%	   %
	%	   \eeq &\sem{\toFOPosRats{P'}}{}{\interpret} 
	%	   \tag{$v \not \in \FV{P'}$} \\
	%	   %
	%	   \eeq &\sem{P'}{}{\interpret}
	%	   \tag{by I.H.}~.
	%	\end{align*}
	%
	%Conversely, if $\VV \in \FV{P'}$, our induction hypothesis implies, we get
	%
	\begin{align*}
	& \sem{\toFOPosRats{\PP}}{\sigma}{\interpret} \eeq \true\\
	\qqiff & \sem{\exists \VV \colon \left( \toFOPosRats{\PP'} \right)}{\sigma}{\interpret} \eeq \true
	\tag{by definition} \\
	\qqiff & \text{there is}~\RR \in \PosRats~\text{with}~ \sem{\toFOPosRats{\PP'}}{\sigma\statesubst{\VV}{\RR}}{\interpret\statesubst{\VV}{\RR}} \eeq \true
	\tag{by definition}\\
	\qqiff& \text{there is}~n \in \Nats~\text{with}~ \sem{\toFOPosRats{\PP'}}{\sigma\statesubst{\VV}{n}}{\interpret\statesubst{\VV}{n}} \eeq \true 
	\tag{If $v \in \FV{\PP'}$, then $\sem{\toFOPosRats{\PP'}}{\sigma\statesubst{\VV}{\RR}}{\interpret\statesubst{\VV}{\RR}} = \true$ only if  $\RR \in \Nats$ by I.H.} \\
	\qqiff& \text{there is}~n \in \Nats~\text{with}~ \sem{\PP'}{\sigma\statesubst{\VV}{n}}{\interpret\statesubst{\VV}{n}} \eeq \true 
	\tag{by I.H. } \\
	\qqiff &\sem{\exists \VV \colon \PP' }{\sigma}{\interpret} \eeq \true
	\tag{by definition}  \\
	\qqiff & \sem{\PP}{\sigma}{\interpret} \eeq \true~.
	\tag{by definition} 
	\end{align*}
	
	%$\sem{\toFOPosRats{P'}}{}{\interpret\statesubst{\VV}{r}} = \false$ for all $r \in \PosRats \setminus %\Nats$ and $\sem{\toFOPosRats{P'}}{}{\interpret\statesubst{\VV}{n}} = %\sem{P'}{}{\interpret\statesubst{\VV}{n}}$ for all $n \in \Nats$, which implies the claim. \\ \\
	%
	\noindent
	\emph{The case $\PP = \forall \VV \colon \PP'$.} First notice that $\FV{\PP} = \FV{\PP'} \setminus \{ \VV \}$.
	Hence, if there is $\VV' \in \FV{\PP}$ with $\sigma(\VV') \not\in \Nats$, then
	\begin{align*}
	& \sem{\toFOPosRats{\PP}}{\sigma}{\interpret} \\
	\eeq & \sem{\forall \VV \colon \left(\toFOPosRats{\PP'} \vee \neg \isNat(\VV) \right)}{\sigma}{\interpret}
	\tag{by definition} \\
	\eeq & \sem{\forall \VV \colon \neg \isNat(\VV)}{\sigma}{\interpret}.
	\tag{$\sem{\toFOPosRats{\PP'}}{\sigma\statesubst{\VV}{\RR}}{\interpret\statesubst{\VV}{\RR}} = \false$ for all $\RR \in \PosRats \setminus \Nats$ by I.H.} \\
	\eeq& \false ~.
	\tag{there is an $\RR \in \PosRats$ with $\RR \not\in \Nats$}
	\end{align*}
	Now assume that for all $\VV' \in \FV{\PP}$ it holds that $\sigma(\VV') \in \Nats$, rendering $\sigma$ an interpretation for $\PP$. We have
	\begin{align*}
	& \sem{\toFOPosRats{\PP}}{\sigma}{\interpret} \eeq \true\\
	\qqiff & \sem{\forall \VV \colon \left( \toFOPosRats{\PP'} \vee \neg \isNat(\VV) \right)}{\sigma}{\interpret} \eeq \true
	\tag{by definition} \\
	\qqiff & \text{for all}~\RR \in\PosRats~\text{we have}~\sem{\toFOPosRats{\PP'} \vee \neg \isNat(\VV)}{\sigma\statesubst{\VV}{\RR}}{\interpret\statesubst{\VV}{n}} \eeq \true \\
	\qqiff& \text{for all}~n\in\Nats~\text{we have}~\sem{\toFOPosRats{\PP'}}{\sigma\statesubst{\VV}{n}}{\interpret\statesubst{\VV}{n}} \eeq \true
	\tag{$\sem{\neg \isNat(\VV)}{\sigma\statesubst{\VV}{\RR}}{\interpret\statesubst{\VV}{\RR}]} = \true$ for all $\RR \in \PosRats\setminus\Nats$} \\
	\qqiff& \text{for all}~n\in\Nats~\text{we have}~\sem{\PP'}{\sigma\statesubst{\VV}{n}}{\interpret\statesubst{\VV}{n}} \eeq \true
	\tag{by I.H.} \\
	\qqiff& \sem{\forall \VV \colon \PP'}{\sigma}{\interpret} \eeq \true
	\tag{by definition} \\
	\qqiff& \sem{\PP}{\sigma}{\interpret} \eeq \true~.
	\tag{by definition} 
	\end{align*}

\end{proof}

\subsection{Proof of Theorem~\ref{thm:exp_subsumes_fo_rats}}

\begin{proof}
	\label{proof:exp_subsumes_fo_rats}
	First notice that $\sem{\iverson{\PP}}{\sigma}{\interpret} \in \{0,1\}$ for all $\PP\in\FOArithPosRats$.
	We now proceed by induction on the structure of $\PP$. \\ \\
	\noindent
	\emph{Base case $\PP = \BB$.} This case follows immediately from the definition of $\sem{\BB}{\sigma}{\interpret}$. \\ \\
	\noindent
	As the induction hypothesis now assume that the theorem holds for some arbitrary, but fixed, $\PP' \in \FOArithPosRats$. \\ \\
	\noindent
	\emph{The case $\PP = \exists \VV \colon \PP'$.} We have
	\begin{align*}
	&\sem{\iverson{\PP}}{\sigma}{\interpret} \eeq 1 \\
	\qiff &  \sem{\SupV{\VV} \iverson{\PP'}}{\sigma}{\interpret} \eeq 1 
	\tag{by definition} \\
	\qiff & \sup~\setcomp{\sem{\iverson{\PP'}}{\pstate\statesubst{\VV}{\RR}}{\interpret\statesubst{\VV}{\RR}}}{\RR \in \PosRats} \eeq 1
	\tag{by definition} \\
	\qiff & \text{there is $\RR \in \PosRats$ with}~\sem{\iverson{\PP'}}{\pstate\statesubst{\VV}{\RR}}{\interpret\statesubst{\VV}{\RR}} \eeq 1
	\tag{$\sem{\iverson{\PP'}}{\pstate\statesubst{\VV}{\RR}}{\interpret\statesubst{\VV}{\RR}} \in \{0,1\}$}\\
	\qiff& \text{there is $\RR \in \PosRats$ with}~ \sem{\PP'}{\pstate\statesubst{\VV}{\RR}}{\interpret\statesubst{\VV}{\RR}} \eeq \true
	\tag{by I.H.} \\
	\qiff&\sem{\exists \VV \colon \PP'}{\pstate}{\interpret} \eeq \true~.
	\tag{by definition}
	\end{align*}
	\emph{The case $\PP = \forall \VV \colon \PP'$.} We have
	\begin{align*}
	&\sem{\iverson{\PP}}{\sigma}{\interpret} \eeq 1 \\
	\qiff &  \sem{\InfV{\VV} \iverson{\PP'}}{\sigma}{\interpret} \eeq 1 
	\tag{by definition} \\
	\qiff & \inf\hspace{1.1ex}\setcomp{\sem{\iverson{\PP'}}{\pstate\statesubst{\VV}{\RR}}{\interpret\statesubst{\VV}{\RR}}}{\RR \in \PosRats} \eeq 1 \\
	\qiff& \text{for all $\RR \in \PosRats$ we have}~ \sem{\iverson{\PP'}}{\pstate\statesubst{\VV}{\RR}}{\interpret\statesubst{\VV}{\RR}} \eeq 1
	\tag{$\sem{\iverson{\PP'}}{\pstate\statesubst{\VV}{\RR}}{\interpret\statesubst{\VV}{\RR}} \in \{0,1\}$} \\
	\qiff &\text{for all $\RR \in \PosRats$ we have}~ \sem{\PP'}{\pstate\statesubst{\VV}{\RR}}{\interpret\statesubst{\VV}{\RR}} \eeq \true
	\tag{by I.H.} \\
	\qiff & \sem{\forall \VV \colon \PP'}{\pstate}{\interpret} \eeq \true 
	\tag{by definition} \\
	\qiff& \sem{\PP}{\pstate}{\interpret} \eeq \true~.
	\tag{by definition}
	\end{align*}
\end{proof}

\subsection{Proof of Theorem~\ref{thm:rho_expectation}}

\begin{proof}
	\label{proof:rho_expectation}
	We define $\rseqelemsymbol$ by
	\begin{align*}
	&\rseqelem{\VV_1}{\VV_2}{\VV_3} \\
	\eeq & \exists n, n_1, n_2 \colon
	\gPair(n, n_1, n_2) \wedge
	\seqelem{\VV_1}{\VV_2}{n} \wedge n_2 \cdot \VV_3 = n_1 
	\wedge \left( \relPrime{n_1}{n_2} \vee (n_1=0 \wedge n_2=1) \right)~,
	\end{align*}
	where $\relPrime{n_1}{n_2}$ denotes \emph{relative primality} of $n_1$ and $n_2$, which is definable in $\FOArithNats$. 
	Let $\RR_0 = \frac{n_{0,1}}{n_{0,2}}, \ldots, \RR_{k-1} = \frac{n_{k-1, 1}}{n_{k-1, 2}}$ such that each $n_{i,j}$ is a natural number satisfying: If $\RR_i = 0$, then $n_{i,1} = 0$ and $n_{i,2} = 1$. If $\RR_i \neq 0$, then $n_{i,1}$ and $n_{i,2}$ are relatively prime. Notice that
	these conditions imply that the pairs $n_{i,0}, n_{i,1}$ are \emph{unique}.
	
	Furthermore, by Lemma~\ref{lem:pairing}, there is a \emph{unique} sequence of natural numbers $n_0,\ldots,n_{k-1}$ with 
	\[
	\gPair(n_i, n_{i, 1}, n_{i,2}) \eequiv \true \quad \text{for all}~i\in\{0,\ldots, k-1\} ~.
	\]
	Finally, by Lemma~\ref{lem:goedel_beta}, there is a natural number $a$ encoding the sequence $n_0,\ldots,n_{k-1}$. This gives us
	\begin{align*}
	&\rseqelem{a}{i}{ r} \eequiv \true \\
	\text{iff} \quad & \exists n, n_1, n_2 \colon
	\gPair(n, n_1, n_2) \wedge
	\seqelem{a}{i}{n} \wedge n_2 \cdot r = n_1
	\wedge \left( \relPrime{n_1}{n_2} \vee (n_1=0 \wedge n_2=1) \right)
	\eequiv \true
	\tag{by definition} \\
	\text{iff} \quad &
	\gPair(n_i, n_{i,1}, n_{i,2}) \wedge
	\seqelem{a}{i}{n_i} \wedge n_{i, 2} \cdot r = n_{i,1} \eequiv \true
	\tag{by above reasoning and uniqueness of the $n$s and the pairs $n_{i,1}, n_{i,2}$} \\
	\text{iff} \quad & r = r_i
	\tag{by construction}~,
	\end{align*}
	which completes the proof.
\end{proof}

%% file: appendix_sums_prod_via_goedel.tex
\subsection{Proof of Theorem~\ref{thm:sum_exp}}

\begin{proof}
	\label{proof:sum_exp}
	
	Write 
	\begin{align*}
	\dexp{\FF} \eeq \qprefix{\FF} \colon \iverson{\BB}~,
	\end{align*}
	with cut variable $\VVcut$ and
	where $\qprefix{\FF} = \Quant_1 \VV_1 \colon \ldots \colon \Quant_n \VV_n$. Furthermore, assume that $\VV, \VV', \gnum, \VU,z$ are fresh logical variables not occurring in $\dexp{\FF}$. Now define
	\begin{align*}
	\gsum{\FF}{\VV}
	\ddefeq &
	\Sup \VV' \colon \Sup \gnum \colon 
	\VV' \cdot \Inf \VU \colon \Inf z \colon \Sup \VVcut \colon \qprefix{\FF} \colon \\
	& [
	\rseqelem{\gnum}{0}{1} \wedge \rseqelem{\gnum}{\VV+1}{\VV}  \\
	&  \quad \wedge \big( (\VU < \VV+1 \wedge \rseqelem{\gnum}{\VU}{z} \wedge (\iverson{\BB} \subst{\vprod}{\VU} \vee \VVcut = 0) ) \\
	& \qquad \quad  \longrightarrow  \rseqelem{\gnum}{\VU+1}{z + \VVcut}   \big)]~.
	\end{align*}
	The reasoning is now analagous to the proof of Theorem~\ref{thm:prod_exp} using Lemma~\ref{lem:sum_by_cut}.
	The equality $\sem{\Sup \VV \colon \gsum{\FF}{\VV}}{\sigma}{\interpret}
	=
	\sum_{j=0}^{\infty} \sem{\FF}{\sigma}{\interpret\statesubst{\vsum}{j}}$ holds since an infinite series evaluates to the supremum of its partial sums.

\end{proof}

\subsection{Proof of Theorem~\ref{thm:prod_exp}}

We employ the following auxiliary result.

\begin{lemma}
	\label{lem:prod_by_cut}
	For all $\alpha_0,\ldots,\alpha_n \in \PosRealsInf$, we have
	\[
	\prod_{j=0}^n  \alpha_j 
	\eeq
	\sup \setcomp{\RR \in \PosRats}{r = r_0 \cdot \ldots \cdot r_n,~\forall 0 \leq i \leq n \colon r_i \in \dcutzero{\alpha_i}}~.
	\]
\end{lemma}
\begin{proof}
	By induction on $n$.
\end{proof}

We now prove Theorem~\ref{thm:prod_exp}.

\begin{proof}
	\label{proof:prod_exp}
	
	Write 
	\begin{align*}
		\dexp{\FF} \eeq \qprefixnoarg \colon \iverson{\BB}
	\end{align*}
	and assume that $\VV, \VV',\gnum, \VU,z$ are fresh variables not occurring in $\dexp{\FF}$. Furthermore, denote by $\qprefixnoarginvert$ the quantifier prefix obtained from $\qprefixnoarg$ by flipping all quantifiers, i.e., by replacing every occurrence of $\Sup$ by $\Inf$ and vice versa.
	Now define 
	\begin{align*}
		\gproduct{\FF}{\VV}
		\ddefeq &
		\Sup \VV' \colon \Sup \gnum \colon 
		\VV' \cdot \Inf \VU \colon \Inf z \colon \Sup \VVcut \colon \qprefixnoarginvert \colon \\
		& [
		\rseqelem{\gnum}{0}{1} \wedge \rseqelem{\gnum}{\VV+1}{\VV'}  \\
		&  \quad \wedge \big( (\VU < \VV+1 \wedge \rseqelem{\gnum}{\VU}{ z} \wedge (\iverson{\BB} \subst{\vprod}{\VU} \vee \VVcut = 0) ) \\
		& \qquad \quad  \longrightarrow  \rseqelem{\gnum}{\VU+1}{z \cdot \VVcut}   \big)]~.
	\end{align*}
	The crux of the proof is to show that the $\{0,1\}$-valued expectation right after $\VV' \cdot \ldots$ evaluates to $1$ on state $\sigma$ iff $\sigma(\gnum)$ encodes a sequence  
	$1, 1\cdot r_0, 1\cdot r_0 \cdot r_1, \ldots, 1\cdot r_0 \cdot \ldots \cdot r_{\sigma(\VV)}$ where  $r_j \in  \dcutzero{\sem{\FF\subst{\vprod}{j}}{\sigma}{\interpret\statesubst{\vprod}{j}}}$
	for all $0 \leq j \leq \sigma(\VV)$ and where $\sigma(\VV') = \prod\limits_{j=0}^{\sigma(\VV)} r_j$.
	This implies
	\begin{align*}
		&\sem{\gproduct{f}{\VV}}{\sigma}{\interpret} \\
		\eeq&\sem{\Sup \VV' \colon \Sup \gnum \colon \VV' \cdot \ldots}{\sigma}{\interpret}
		\tag{by definition} \\
		\eeq&\sup_{r \in \PosRats}
		\sup_{r' \in \PosRats}
		\setcomp{r}{r'~\text{encodes}~1,\ldots, 1\cdot r_0 \cdot \ldots \cdot r_{\sigma(\VV)} ~\text{and}~r = \prod\limits_{j=0}^{\sigma(\VV)} r_j~\text{where}~r_j \in  \dcutzero{\sem{\FF}{\sigma\statesubst{\vprod}{j}}{\interpret\statesubst{\vprod}{j}}} }
		\tag{claim proven below}\\
		\eeq&\sup_{r \in \PosRats}
		\setcomp{r}{r = \prod\limits_{j=0}^{\sigma(\VV)} r_j~\text{where}~r_j \in  \dcutzero{\sem{\FF}{\sigma\statesubst{\vprod}{j}}{\sigma\statesubst{\vprod}{j}}}}~.
		\\
		\eeq&\prod\limits_{j=0}^{\sigma(\VV)} 
		\sem{\FF}{\sigma\statesubst{\vprod}{j}}{\interpret\statesubst{\vprod}{j}}~.
		\tag{by Lemma~\ref{lem:prod_by_cut}}
	\end{align*}
	We have
	\begin{align*}
		&\semleft{\Inf \VU \colon \Inf z \colon \Sup \VVcut \colon \qprefixnoarginvert \colon}{\sigma}{\interpret} \\
		& [
		\rseqelem{\gnum}{0}{1} \wedge \rseqelem{\gnum}{\VV+1}{\VV'}  \\
		&  \quad \wedge \big( (\VU < \VV+1 \wedge \rseqelem{\gnum}{\VU}{z} \wedge (\iverson{\BB} \subst{\vprod}{\VU} \vee \VVcut = 0) ) \\
		& \qquad \quad  \longrightarrow  \rseqelem{\gnum}{\VU+1}{z \cdot \VVcut}   \big)] \semright = 1 \\
		\text{iff}\quad&\text{for all}~r,s \in \PosRats~\text{there is}~t \in \PosRats~\text{with} \semleft{\qprefixnoarginvert \colon}{\sigma}{\interpret} \\
		& [
		\rseqelem{\gnum}{0}{1} \wedge \rseqelem{\gnum}{\VV+1}{\VV'}  \\
		&  \quad \wedge \big( (r < \VV+1 \wedge \rseqelem{\gnum}{r}{s} \wedge (\iverson{\BB} \subst{\vprod}{r}\subst{\VVcut}{t} \vee t = 0) ) \\
		& \qquad \quad  \longrightarrow  \rseqelem{\gnum}{r+1}{s \cdot t}   \big)] \semright = 1 
		\tag{expectation is \{0,1\}-valued}\\
		\text{iff}\quad&\sem{[
			\rseqelem{\gnum}{0}{1} \wedge \rseqelem{\gnum}{\VV+1}{\VV'}]}{\sigma}{\interpret} =1 ~\text{and}~
		\text{for all}~r,s \in \PosRats \\
		&\text{there is}~t \in \PosRats~\text{with} \semleft{\qprefixnoarginvert \colon}{\sigma}{\interpret} \\
		& [ \big( (r < \VV+1 \wedge \rseqelem{\gnum}{r}{s} \wedge (\iverson{\BB} \subst{\vprod}{r}\subst{\VVcut}{t} \vee t = 0) ) \\
		& \qquad \quad  \longrightarrow  \rseqelem{\gnum}{r+1}{s \cdot t}   \big)] \semright \eeq 1
		\tag{quantitative quantifiers correspond to qualitative quantifiers, standard prenexing} \\
		\text{iff}\quad&\sem{[
			\rseqelem{\gnum}{0}{1} \wedge \rseqelem{\gnum}{\VV+1}{\VV'}]}{\sigma}{\interpret} =1 ~\text{and}~
		\text{for all}~r,s \in \PosRats~\text{there is}~t \in \PosRats~\text{with} \\
		&\sem{ [ r < \VV+1 \wedge \rseqelem{\gnum}{r}{s}}{\sigma}{\interpret} =1 ~\text{and}~
		\sem{\qprefixnoarg \colon \iverson{\BB} \subst{\vprod}{r}\subst{\VVcut}{t} \vee t = 0}{\sigma}{\interpret} = 1 \\
		&\quad \text{implies}~\sem{\rseqelem{\gnum}{r+1}{s \cdot t} }{\sigma}{\interpret} \eeq 1 
		\tag{standard prenexing into antecedent of implication, expectation is $\{0,1\}$-valued} \\
		\text{iff}\quad&\sem{[
			\rseqelem{\gnum}{0}{1} \wedge \rseqelem{\gnum}{\VV+1}{\VV}]}{\sigma}{\interpret} =1 ~\text{and}~
		\text{for all}~r,s \in \PosRats~\text{there is}~t \in \PosRats~\text{with} \\
		&\sem{ [ r < \VV+1 \wedge \rseqelem{\gnum}{r}{s}}{\sigma}{\interpret} =1 ~\text{and}~
		t \in \dcut{\sem{\FF}{\sigma\statesubst{\vprod}{r}}{\interpret\statesubst{\vprod}{r}}} \cup \{0\} \\
		&\quad \text{implies}~\sem{\rseqelem{\gnum}{r+1}{s \cdot t} }{\sigma}{\interpret} \eeq 1 
		\tag{by Theorem~\ref{thm:dedekind_nf}}\\
		\text{iff}\quad & \sigma(\gnum)~\text{encodes sequence} ~1, 1\cdot r_0, 1\cdot r_0 \cdot r_1, \ldots, 1\cdot r_0 \cdot \ldots \cdot r_{\sigma(\VV)}~ \\
		&\text{with}~  r_j \in  \dcut{\sem{\FF}{\sigma\statesubst{\vprod}{j}}{\interpret\statesubst{\vprod}{j}}}%\cup \{0\}
		~\text{for all}~0 \leq j \leq \sigma(\VV)~\\
		&\text{and where} ~\sigma(\VV') = \prod\limits_{j=0}^{\sigma(\VV)} r_j~.
	\end{align*}

\end{proof}

%% file: appendix_expressiveness.tex
Given an expectation $\ff \in \E$, we denote by
\[
\Vars(\ff) \eeq \setcomp{x \in \Vars}{\exists \sigma \in \States \colon \exists n,n' \in \Nats  \colon \ff(\sigma\statesubst{x}{n}) \neq \ff(\sigma\statesubst{x}{n'})}
\]
the set of all \enquote{relevant} variables in $\ff$. We restrict to expectations $\ff$ with $|\Vars(\ff)| < \infty$, since $|\Vars(\eval{\FF})| < \infty$ holds for every \emph{syntactic} expectation $\FF$. 
\Cref{thm:wp_loop_as_sum} is a consequence of the $\WHILESYMBOL$-case of the following theorem.
\begin{theorem}
	\label{thm:wp_prob_times_exp}
	Let $\cc$ be a program and $\ff$ be an expectation. Furthermore, let $\varseq{x}$ be a finite set of program variables with $\Vars(\cc) \cup \Vars(\ff) \subseteq \varseq{x}$. We have
	\[
	\wp{\cc}{\ff} \eeq \lambda \sigma_0 \mydot \sum_{\sigma \in \partitionedstates{x}} 
	\wp{\cc}{\statepred{\sigma}{\varseq{x}}}(\sigma_0) \cdot \ff(\sigma)~. 
	\]
\end{theorem}
\begin{proof}
	\label{proof:wp_prob_times_exp}
	By induction on the structure of $\cc$. For a state $\sigma'$, we often abbreviate $\statepred{\sigma'}{\varseq{x}}$ by $\statepred{\sigma'}{}$. \\ \\
	\noindent
	\emph{The case $\cc = \SKIP$.} We have
	\begin{align*}
	&\wp{\SKIP}{\ff}(\sigma_0) \\
	\eeq & \ff(\sigma_0)
	\tag{by definition} \\
	\eeq& \statepred{\sigma_0}{}(\sigma_0) \cdot \ff(\sigma_0) 
	\tag{$\statepred{\sigma_0}{}(\sigma_0) = 1$} \\
	\eeq& \wp{\SKIP}{\statepred{\sigma_0}{}}(\sigma_0) \cdot \ff(\sigma_0) 
	\tag{$\wp{\SKIP}{\statepred{\sigma_0}{}} = \statepred{\sigma_0}{}$} \\
	\eeq& \sum_{\sigma \in \partitionedstates{x}}
	\wp{\SKIP}{\statepred{\sigma}{}}(\sigma_0) \cdot \ff(\sigma) 
	\tag{there is exactly one $\sigma \in \partitionedstates{x}$ with
		$\wp{\SKIP}{\statepred{\sigma}{}}(\sigma_0) = 1$ and for this $\sigma$ we have $\equivstates{\sigma}{\varseq{x}}{\sigma_0}$}~.
	\end{align*}
	\emph{The case $\ASSIGN{\XX}{\TT}$.}  We have
	\begin{align*}
	& \wp{\ASSIGN{\XX}{\TT}}{\ff}(\sigma_0) \\
	\eeq &  \ff(\sigma_0\statesubst{\XX}{\sem{\TT}{\sigma_0}{}}) 
	\tag{by definition} \\
	\eeq & \statepred{\sigma_0\statesubst{\XX}{\sem{\TT}{\sigma_0}{}}{}}{}
	(\sigma_0\statesubst{\XX}{\sem{\TT}{\sigma_0}{}})
	\cdot 
	\ff(\sigma_0\statesubst{\XX}{\sem{\TT}{\sigma_0}{}}) 
	\tag{$\statepred{\sigma_0\statesubst{\XX}{\sem{\TT}{\sigma_0}{}}{}}{}
		(\sigma_0\statesubst{\XX}{\sem{\TT}{\sigma_0}{}}) = 1$} \\
	\eeq& \wp{\ASSIGN{\XX}{\TT}}{
		\statepred{\sigma_0\statesubst{\XX}{\sem{\TT}{\sigma_0}{}}{}}{}
	}(\sigma_0) 
	\cdot 
	\ff(\sigma_0\statesubst{\XX}{\sem{\TT}{\sigma_0}{}})
	\tag{by definition}  \\
	\eeq& \sum_{\sigma \in \partitionedstates{x}}
	\wp{\ASSIGN{\XX}{\TT}}{\statepred{\sigma}{}}(\sigma_0) \cdot \ff(\sigma) ~.
	\tag{$\substack{\text{there is exactly one $\sigma \in \partitionedstates{x}$ with
		$\wp{\ASSIGN{\XX}{\TT}}{\statepred{\sigma}{}}(\sigma_0) = 1$} \\ \text{and for this $\sigma$ we have $\equivstates{\sigma}{\varseq{x}}{\sigma_0\statesubst{\XX}{\sem{\TT}{\sigma_0}{}}}$}}$}
	\end{align*}
	As the induction hypothesis now assume that the theorem holds for some arbitrary, but fixed programs $\cc_1, \cc_2$ and all postexpectations $\ff$. \\ \\
	\noindent
	\emph{The case $\cc = \COMPOSE{\cc_1}{\cc_2}$.} We have
	\begin{align*}
	&\wp{\COMPOSE{\cc_1}{\cc_2}}{\ff}(\sigma_0) \\
	\eeq&\wp{\cc_1}{\wp{\cc_2}{\ff}}(\sigma_0)
	\tag{by definition} \\
	\eeq&\sum_{\sigma \in \partitionedstates{x}} \wp{\cc_1}{\statepred{\sigma}{}}(\sigma_0)
	\cdot \wp{\cc_2}{\ff}(\sigma)
	\tag{I.H.\ on $\cc_1$} \\
	\eeq & \sum_{\sigma \in \partitionedstates{x}} \wp{\cc_1}{\statepred{\sigma}{}}(\sigma_0)
	\cdot \sum_{\sigma' \in \partitionedstates{x}} \wp{\cc_2}{\statepred{\sigma'}{}}(\sigma)
	\cdot \ff(\sigma')
	\tag{I.H.\ on $\cc_2$} \\
	\eeq & \sum_{\sigma \in \partitionedstates{x}}  \sum_{\sigma' \in \partitionedstates{x}} \wp{\cc_1}{\statepred{\sigma}{}}(\sigma_0)
	\cdot\wp{\cc_2}{\statepred{\sigma'}{}}(\sigma)
	\cdot \ff(\sigma')
	\tag{algebra} \\
	\eeq &  \sum_{\sigma' \in \partitionedstates{x}} \sum_{\sigma \in \partitionedstates{x}}  \wp{\cc_1}{\statepred{\sigma}{}}(\sigma_0)
	\cdot\wp{\cc_2}{\statepred{\sigma'}{}}(\sigma)
	\cdot \ff(\sigma')
	\tag{algebra} \\
	\eeq &  \sum_{\sigma' \in \partitionedstates{x}}  \wp{\cc_1}{\wp{\cc_2}{\statepred{\sigma'}{}}}(\sigma_0)
	\cdot \ff(\sigma')
	\tag{I.H.\ on $\cc_1$} \\
	\eeq &\sum_{\sigma' \in \partitionedstates{x}} \wp{\COMPOSE{\cc_1}{\cc_2}}{\statepred{\sigma'}{}}(\sigma_0)
	\cdot \ff(\sigma')~.
	\tag{by definition}
	\end{align*}
	\emph{The case $\cc = \ITE{\BB}{\cc_1}{\cc_2}$.} We distinguish the cases 
	$\iverson{\BB}(\sigma_0) = 1$ and $\iverson{\neg \BB}(\sigma_0) = 1$. For $\iverson{\BB}(\sigma_0) = 1$, we have
	\begin{align*}
	& \wp{\ITE{\BB}{\cc_1}{\cc_2}}{\ff}(\sigma_0) \\
	\eeq & \iverson{\BB}(\sigma_0) \cdot \wp{\cc_1}{\ff}(\sigma_0)
	+
	\iverson{\neg \BB}(\sigma_0) \cdot \wp{\cc_2}{\ff}(\sigma_0)
	\tag{by definition} \\
	\eeq & \iverson{\BB}(\sigma_0) \cdot \wp{\cc_1}{\ff}(\sigma_0)
	\tag{$\iverson{\neg \BB}(\sigma_0)=0$ by assumption} \\
	\eeq& \iverson{\BB}(\sigma_0) \cdot \sum_{\sigma_1 \in \partitionedstates{x}}
	\wp{\cc_1}{\statepred{\sigma_1}{}}(\sigma_0) \cdot \ff(\sigma_1) 
	\tag{I.H.\ on $\cc_1$} \\
	\eeq& \sum_{\sigma_1 \in \partitionedstates{x}}\iverson{\BB}(\sigma_0) \cdot 
	\wp{\cc_1}{\statepred{\sigma_1}{}}(\sigma_0) \cdot \ff(\sigma_1) 
	\tag{algebra} \\
	\eeq& \sum_{\sigma_1 \in \partitionedstates{x}}
	\left(
	\iverson{\BB}(\sigma_0) \cdot 
	\wp{\cc_1}{\statepred{\sigma_1}{}}(\sigma_0)
	+
	\iverson{\neg \BB}(\sigma_0) \cdot \wp{\cc_2}{\statepred{\sigma_1}{}}(\sigma_0)
	\right)
	\cdot \ff(\sigma_1) 
	\tag{$\iverson{\neg \BB}(\sigma_0)=0$ by assumption} \\
	\eeq& \sum_{\sigma_1 \in \partitionedstates{x}}
	\wp{\ITE{\BB}{\cc_1}{\cc_2}}{\sigma_1}(\sigma_0)
	\cdot \ff(\sigma_1) 
	\tag{by definition}~.
	\end{align*}
	The case $\iverson{\neg \BB}(\sigma_0)=1$ is completely analogous. \\ \\
	\noindent
	\emph{The case $\cc = \WHILEDO{\BB}{\cc_1}$.} This case is more involved.
	%\charwp{\BB}{\cc'}{\ff}(\fg)
	First observe that for every $\sigma_1 \in \States$,
	\begin{align}
	\label{eqn:proof_thm_wp_prob_times_exp_1}
	\sup_{k \in \Nats} \charwpn{\BB}{\cc_1}{\statepred{\sigma_1}{}}{k}(0)(\sigma_0)
	\eeq
	\wp{\WHILEDO{\BB}{\cc_1}}{\statepred{\sigma_1}{}}(\sigma_0)~.
	\tag{by Lemma~\ref{lem:kleene_for_wp}}
	\end{align}
	We proceed by induction on $k$ to show that
	\begin{align}
	&\charwpn{\BB}{\cc_1}{\ff}{k}(0)(\sigma_0) \notag \\
	\eeq& 
	\sum_{\sigma_0,\ldots,\sigma_{k-1} \in \partitionedstates{x}}
	(\iverson{\neg \BB} \cdot \ff)(\sigma_{k-1})
	\cdot
	\prod\limits_{i=0}^{k-2} \wp{\ITE{\BB}{\cc_1}{\SKIP}}{\statepred{\sigma_{i+1}}{}}(\sigma_i) \notag \\
	\eeq& \sum_{\sigma_1 \in \partitionedstates{x}}
	\charwpn{\BB}{\cc_1}{\statepred{\sigma_1}{}}{k}(0)(\sigma_0) \cdot \ff(\sigma_1)~.
	\label{eqn:proof_thm_wp_prob_times_exp_2}
	\end{align}
	This implies the claim, since
	\begin{align*}
	&\wp{\WHILEDO{\BB}{\cc_1}}{\ff}(\sigma_0) \\
	\eeq&\sup_{k \in \Nats} \charwpn{\BB}{\cc_1}{\ff}{k}(0)(\sigma_0)
	\tag{by Lemma~\ref{lem:kleene_for_wp}} \\
	\eeq& \sup_{k \in \Nats}\sum_{\sigma_1 \in \partitionedstates{x}}
	\charwpn{\BB}{\cc_1}{\statepred{\sigma_1}{}}{k}(0)(\sigma_0) \cdot \ff(\sigma_1)
	\tag{by Equation~\ref{eqn:proof_thm_wp_prob_times_exp_2}} \\
	\eeq& \sup_{k \in \Nats} \sup_{k' \in \Nats}\sum\limits_{i=0}^{k'}
	\charwpn{\BB}{\cc_1}{\statepred{\sfsymbol{enum}(i)}{}}{k}(0)(\sigma_0) \cdot \ff(\sfsymbol{enum}(i))
	\tag{choose some bijection $\sfsymbol{enum}\colon \Nats \to \partitionedstates{x}$, value of infinite series is supremum of partial sums} \\
	\eeq&\sup_{k' \in \Nats} \sup_{k \in \Nats} \sum\limits_{i=0}^{k'}
	\charwpn{\BB}{\cc_1}{\statepred{\sfsymbol{enum}(i)}{}}{k}(0)(\sigma_0) \cdot \ff(\sfsymbol{enum}(i))
	\tag{swap suprema} \\
	\eeq&\sup_{k' \in \Nats}  \sum\limits_{i=0}^{k'} \sup_{k \in \Nats}
	\charwpn{\BB}{\cc_1}{\statepred{\sfsymbol{enum}(i)}{}}{k}(0)(\sigma_0) \cdot \ff(\sfsymbol{enum}(i))
	\tag{algebra, sum is finite} \\
	\eeq& \sup_{k' \in \Nats}  \sum\limits_{i=0}^{k'} \wp{\WHILEDO{\BB}{\cc_1}}{\statepred{\sfsymbol{enum}(i)}{}}(\sigma_0) \cdot \ff(\sfsymbol{enum}(i)) 
	\tag{by Lemma~\ref{lem:kleene_for_wp}} \\
	\eeq& \sum_{\sigma_1 \in \partitionedstates{x}} \wp{\WHILEDO{\BB}{\cc_1}}{\statepred{\sigma_1}{}}(\sigma_0) \cdot \ff(\sigma_1) 
	\tag{value of infinite series is supremum of partial sums}~.
	\end{align*}
	\noindent
	\emph{Base case $n=0$.} For $\charwpn{\BB}{\cc_1}{\ff}{0}(0)(\sigma_0)$, we have
	\begin{align*}
	&\charwpn{\BB}{\cc_1}{\ff}{0}(0)(\sigma_0) \\
	\eeq & 0 \\
	\eeq  &\sum_{\sigma_0,\ldots,\sigma_{k-1} \in \partitionedstates{x}}
	(\iverson{\neg \BB} \cdot \ff)(\sigma_{k-1})
	\cdot
	\prod\limits_{i=0}^{k-2} \wp{\ITE{\BB}{\cc_1}{\SKIP}}{\statepred{\sigma_{i+1}}{}}(\sigma_i)~.
	\tag{empty sum evaluates to $0$}
	\end{align*}
	For $\sum_{\sigma_1 \in \partitionedstates{x}}
	\charwpn{\BB}{\cc_1}{\statepred{\sigma_1}{}}{0}(0)(\sigma_0) \cdot \ff(\sigma_1)$, we have
	\begin{align*}
	&\sum_{\sigma_1 \in \partitionedstates{x}}
	\charwpn{\BB}{\cc_1}{\statepred{\sigma_1}{}}{0}(0)(\sigma_0) \cdot \ff(\sigma_1) \\
	\eeq & \sum_{\sigma_1 \in \partitionedstates{x}}
	0 \cdot \ff(\sigma_1) \\
	\eeq & 0~.
	\end{align*}
	As the induction hypothesis now assume that Equation~(\ref{eqn:proof_thm_wp_prob_times_exp_2}) holds for some arbitrary, but fixed, $k\in\Nats$. \\ \\
	\noindent
	\emph{Induction Step.} For $\charwpn{\BB}{\cc_1}{\ff}{k+1}(0)(\sigma_0)$, we have
	\begin{align*}
	& \charwpn{\BB}{\cc_1}{\ff}{k+1}(0)(\sigma_0) \\
	\eeq& \charwp{\BB}{\cc_1}{k} \left( \charwpn{\BB}{\cc_1}{\ff}{k}(0)\right)(\sigma_0)
	\tag{by definition} \\
	\eeq&\iverson{\BB}(\sigma_0) \cdot \wp{\cc_1}{\charwpn{\BB}{\cc_1}{\ff}{k}(0)}(\sigma_0)
	+ \iverson{\neg\BB}(\sigma_0) \cdot \ff(\sigma_0)
	\tag{by definition} \\
	\eeq&\iverson{\BB}(\sigma_0)  \\
	& \quad \cdot \wp{\cc_1}{ \lambda \sigma_0 \mydot
		\sum_{\sigma_0,\ldots,\sigma_{k-1} \in \partitionedstates{x}}
		(\iverson{\neg \BB} \cdot \ff)(\sigma_{k-1})
		\cdot
		\prod\limits_{i=0}^{k-2} \wp{\ITE{\BB}{\cc_1}{\SKIP}}{\statepred{\sigma_{i+1}}{}}(\sigma_i)    
	}(\sigma_0) \\
	&\quad + \iverson{\neg\BB}(\sigma_0) \cdot \ff(\sigma_0)
	\tag{I.H. on $k$} \\
	\eeq&\iverson{\BB}(\sigma_0) \cdot
	\sum_{\sigma_1 \in \partitionedstates{x}}
	\wp{\cc_1}{\statepred{\sigma_1}{}}(\sigma _0) \\
	& \quad \cdot 
	\left( \lambda \sigma_0 \mydot
	\sum_{\sigma_0,\ldots,\sigma_{k-1} \in \partitionedstates{x}}
	(\iverson{\neg \BB} \cdot \ff)(\sigma_{k-1})
	\cdot
	\prod\limits_{i=0}^{k-2} \wp{\ITE{\BB}{\cc_1}{\SKIP}}{\statepred{\sigma_{i+1}}{}}(\sigma_i)    
	\right) (\sigma_1) \\
	&\quad + \iverson{\neg\BB}(\sigma_0) \cdot \ff(\sigma_0)
	\tag{I.H. on $\cc_1$} \\
	\eeq&\iverson{\BB}(\sigma_0) \cdot
	\sum_{\sigma_1 \in \partitionedstates{x}}
	\wp{\cc_1}{\statepred{\sigma_1}{}}(\sigma _0) \\
	& \quad \cdot 
	\left( 
	\sum_{\sigma_1,\ldots,\sigma_{k} \in \partitionedstates{x}}
	(\iverson{\neg \BB} \cdot \ff)(\sigma_{k})
	\cdot
	\prod\limits_{i=1}^{k-1} \wp{\ITE{\BB}{\cc_1}{\SKIP}}{\statepred{\sigma_{i+1}}{}}(\sigma_i)    
	\right)  \\
	&\quad + \iverson{\neg\BB}(\sigma_0) \cdot \ff(\sigma_0)
	\tag{applying $\sigma_1$ and index shift} \\
	\eeq&
	\sum_{\sigma_1 \in \partitionedstates{x}} \iverson{\BB}(\sigma_0) \cdot
	\wp{\cc_1}{\statepred{\sigma_1}{}}(\sigma _0) \\
	& \quad \cdot 
	\left( 
	\sum_{\sigma_1,\ldots,\sigma_{k} \in \partitionedstates{x}}
	(\iverson{\neg \BB} \cdot \ff)(\sigma_{k})
	\cdot
	\prod\limits_{i=1}^{k-1} \wp{\ITE{\BB}{\cc_1}{\SKIP}}{\statepred{\sigma_{i+1}}{}}(\sigma_i)    
	\right)  \\
	&\quad + \iverson{\neg\BB}(\sigma_0) \cdot \ff(\sigma_0)
	\tag{algebra} \\
	\eeq&
	\sum_{\sigma_0,\ldots,\sigma_{k} \in \partitionedstates{x}}
	(\iverson{\neg \BB} \cdot \ff)(\sigma_{k})
	\cdot
	\prod\limits_{i=0}^{k-1} \wp{\ITE{\BB}{\cc_1}{\SKIP}}{\statepred{\sigma_{i+1}}{}}(\sigma_i) ~.
	\tag{see below} 
	\end{align*}
	To see that the last step is sound, distinguish the cases $\iverson{\BB}(\sigma_0) =1$
	and $\iverson{\BB}(\sigma_0)=0$. If $\iverson{\BB}(\sigma_0)=1$, then 
	$\wp{\ITE{\BB}{\cc_1}{\SKIP}}{\statepred{\sigma_1}{}}(\sigma_0)
	= \wp{\cc_1}{\statepred{\sigma_1}{}}(\sigma_0)$. Conversely, if $\iverson{\BB}(\sigma_0)=0$, then there is \emph{exactly one}
	sequence of states $\sigma_0, \ldots,\sigma_k = \sigma_0, \ldots, \sigma_0$ such that
	\[
	\prod\limits_{i=0}^{k-1} \wp{\ITE{\BB}{\cc_1}{\SKIP}}{\statepred{\sigma_{i+1}}{}}(\sigma_i)
	\]
	evaluates to $1$. For all other sequences, the above product evaluates to $0$. This gives us 
	\begin{align*}
	&\sum_{\sigma_0,\ldots,\sigma_{k} \in \partitionedstates{x}}
	(\iverson{\neg \BB} \cdot \ff)(\sigma_{k-1})
	\cdot
	\prod\limits_{i=0}^{k-1} \wp{\ITE{\BB}{\cc_1}{\SKIP}}{\statepred{\sigma_{i+1}}{}}(\sigma_i) \\
	\eeq& 
	(\iverson{\neg \BB} \cdot \ff)(\sigma_{0})
	\cdot
	\prod\limits_{i=0}^{k-1} \wp{\ITE{\BB}{\cc_1}{\SKIP}}{\statepred{\sigma_{0}}{}}(\sigma_0)  \\
	\eeq & (\iverson{\neg \BB} \cdot \ff)(\sigma_{0})~,
	\end{align*}
	which completes this case. \\ \\ 
	\noindent
	For $\sum_{\sigma_1 \in \partitionedstates{x}}
	\charwpn{\BB}{\cc_1}{\statepred{\sigma_1}{}}{k+1}(0)(\sigma_0) \cdot \ff(\sigma_1)$, we have
	\begin{align*}
	&\sum_{\sigma_1 \in \partitionedstates{x}}
	\charwpn{\BB}{\cc_1}{\statepred{\sigma_1}{}}{k+1}(0)(\sigma_0) \cdot \ff(\sigma_1) \\
	\eeq&\sum_{\sigma_1 \in \partitionedstates{x}}
	\charwp{\BB}{\cc_1}{\statepred{\sigma_1}{}}(\charwpn{\BB}{\cc_1}{\statepred{\sigma_1}{}}{k}(0))(\sigma_0)
	\cdot \ff(\sigma_1) 
	\tag{by definition}\\
	\eeq&\sum_{\sigma_1 \in \partitionedstates{x}}
	\left(
	\iverson{\BB}(\sigma_0)  
	\cdot
	\wp{\cc_1}{\charwpn{\BB}{\cc_1}{\statepred{\sigma_1}{}}{k}(0)}(\sigma_0)
	+ (\iverson{\neg\BB} \cdot \statepred{\sigma_1}{})(\sigma_0)
	\right)
	\cdot \ff(\sigma_1)
	\tag{by definition} \\
	\eeq&\sum_{\sigma_1 \in \partitionedstates{x}}
	\iverson{\BB}(\sigma_0)  
	\cdot
	\wp{\cc_1}{\charwpn{\BB}{\cc_1}{\statepred{\sigma_1}{}}{k}(0)}(\sigma_0)
	\cdot \ff(\sigma_1) \\
	&
	+ (\iverson{\neg\BB} \cdot \ff)(\sigma_0)
	\tag{$(\iverson{\neg\BB} \cdot \statepred{\sigma_1}{})(\sigma_0) \neq 0$ 
		only if $\equivstates{\sigma_0}{\varseq{x}}{\sigma_1}$} \\
	\eeq&\sum_{\sigma_1 \in \partitionedstates{x}}
	\iverson{\BB}(\sigma_0)  
	\cdot
	\left(\sum_{\sigma \in \partitionedstates{x}}
	\wp{\cc_1}{\statepred{\sigma}{}}(\sigma_0)
	\cdot \charwpn{\BB}{\cc_1}{\statepred{\sigma_1}{}}{k}(0)(\sigma)\right)
	\cdot \ff(\sigma_1) \\
	&
	+ (\iverson{\neg\BB} \cdot \ff)(\sigma_0)
	\tag{I.H.\ on $\cc_1$} \\
	\eeq&\sum_{\sigma_1 \in \partitionedstates{x}} \sum_{\sigma \in \partitionedstates{x}}
	\iverson{\BB}(\sigma_0)  
	\cdot
	\wp{\cc_1}{\statepred{\sigma}{}}(\sigma_0)
	\cdot \charwpn{\BB}{\cc_1}{\statepred{\sigma_1}{}}{k}(0)(\sigma)
	\cdot \ff(\sigma_1) \\
	&
	+ (\iverson{\neg\BB} \cdot \ff)(\sigma_0)
	\tag{algebra} \\
	\eeq&\sum_{\sigma \in \partitionedstates{x}} \sum_{\sigma_1 \in \partitionedstates{x}} 
	\iverson{\BB}(\sigma_0)  
	\cdot
	\wp{\cc_1}{\statepred{\sigma}{}}(\sigma_0)
	\cdot \charwpn{\BB}{\cc_1}{\statepred{\sigma_1}{}}{k}(0)(\sigma)
	\cdot \ff(\sigma_1) \\
	&
	+ (\iverson{\neg\BB} \cdot \ff)(\sigma_0)
	\tag{swap sums} \\
	\eeq&\sum_{\sigma \in \partitionedstates{x}} 
	\iverson{\BB}(\sigma_0)  
	\cdot
	\wp{\cc_1}{\statepred{\sigma}{}}(\sigma_0)
	\cdot \left(
	\sum_{\sigma_1 \in \partitionedstates{x}} \charwpn{\BB}{\cc_1}{\statepred{\sigma_1}{}}{k}(0)(\sigma)
	\cdot \ff(\sigma_1) \right) \\
	&
	+ (\iverson{\neg\BB} \cdot \ff)(\sigma_0)
	\tag{algebra} \\
	\eeq&\sum_{\sigma \in \partitionedstates{x}} 
	\iverson{\BB}(\sigma_0)  
	\cdot
	\wp{\cc_1}{\statepred{\sigma}{}}(\sigma_0) \\
	& \quad \cdot \left(
	\lambda \sigma_0 \mydot \sum_{\sigma_0,\ldots,\sigma_{k-1} \in \partitionedstates{x}}
	(\iverson{\neg \BB} \cdot \ff)(\sigma_{k-1})
	\cdot
	\prod\limits_{i=0}^{k-2} \wp{\ITE{\BB}{\cc_1}{\SKIP}}{\statepred{\sigma_{i+1}}{}}(\sigma_i)
	\right)(\sigma) \\
	&
	+ (\iverson{\neg\BB} \cdot \ff)(\sigma_0)
	\tag{I.H. on $k$} \\ 
	\eeq&\sum_{\sigma_1 \in \partitionedstates{x}} 
	\iverson{\BB}(\sigma_0)  
	\cdot
	\wp{\cc_1}{\statepred{\sigma_1}{}}(\sigma_0) \\
	& \quad \cdot \left(
	\lambda \sigma_0 \mydot \sum_{\sigma_0,\ldots,\sigma_{k-1} \in \partitionedstates{x}}
	(\iverson{\neg \BB} \cdot \ff)(\sigma_{k-1})
	\cdot
	\prod\limits_{i=0}^{k-2} \wp{\ITE{\BB}{\cc_1}{\SKIP}}{\statepred{\sigma_{i+1}}{}}(\sigma_i)
	\right)(\sigma_1) \\
	&
	+ (\iverson{\neg\BB} \cdot \ff)(\sigma_0)
	\tag{rename $\sigma$ by $\sigma_1$} \\ \\
	\eeq&\sum_{\sigma_1 \in \partitionedstates{x}} 
	\iverson{\BB}(\sigma_0)  
	\cdot
	\wp{\cc_1}{\statepred{\sigma_1}{}}(\sigma_0) \\
	& \quad \cdot \left(
	\sum_{\sigma_1,\ldots,\sigma_{k} \in \partitionedstates{x}}
	(\iverson{\neg \BB} \cdot \ff)(\sigma_{k}) 
	\cdot 
	\prod\limits_{i=1}^{k-1} \wp{\ITE{\BB}{\cc_1}{\SKIP}}{\statepred{\sigma_{i+1}}{}}(\sigma_i)
	\right) \\
	&
	+ (\iverson{\neg\BB} \cdot \ff)(\sigma_0)
	\tag{applying $\sigma_1$ and index shift} \\ 
	\eeq&
	\sum_{\sigma_0,\ldots,\sigma_{k} \in \partitionedstates{x}}
	(\iverson{\neg \BB} \cdot \ff)(\sigma_{k})
	\cdot
	\prod\limits_{i=0}^{k-1} \wp{\ITE{\BB}{\cc_1}{\SKIP}}{\statepred{\sigma_{i+1}}{}}(\sigma_i) ~.
	\tag{see reasoning for previous case}
	\end{align*}
	This completes the proof.
\end{proof}

\subsection{Proof of \Cref{thm:expressive}}
\label{proof:epressiveness_appendix}
We employ an auxiliary result.
    \begin{lemma}
	\label{lem:express_statepred_exp}
	Let $\cc \in \pgcl$. Then, for every $\sigma \in \States$, we have
	\begin{align*}
	\wp{\cc}{\statepred{\sigma}{\varseq{x}}} \eeq 
	\wp{\cc}{\iverson{x_0 = x_0' \wedge \ldots \wedge x_{n-1}=x_{n-1}'}}
	\left[ x_0' / \sigma(x_0) , \ldots ,  x_{n-1}' / \sigma(x_{n-1}) \right]~.
	\end{align*}
\end{lemma} 
\begin{proof}
	By induction on $\cc$.
\end{proof}

Since we encode program states $\sigma$ as G\"odel numbers $\seqnum{\sigma}$, we define for every $\FF \in \SyntE$ a syntactic expectation $\gsubst{\varseq{x}}{\FF}{\VV}$ such that $\gsubst{\varseq{x}}{\FF}{\seqnum{\sigma}}$ is equivalent to the syntactic expectation obtained from substituting every $x \in \varseq{x}$ by $\sigma(x)$ in $\FF$. For that, let $\VV_0,\ldots,\VV_{n-1}$ be fresh variables not occurring in $\FF$. Now define
\begin{align*}
&\gsubst{\varseq{x}}{\FF}{\VV} \\
\ddefeq&
\Sup \VV_0 \colon \ldots \colon \Sup \VV_{n-1} \colon
\iverson{\rseqelem{\VV}{0}{\VV_0}
	\wedge \ldots \wedge \rseqelem{\VV}{n-1}{\VV_{n-1}}}
\exprod
\FF
\left[ x_0 / \VV_0 , \ldots ,  x_{n-1} / \VV_{n-1} \right]~.
\end{align*}
\begin{lemma}
	\label{lem:subst_goedel}
	%Let $f \in \SyntE$, $\varseq{x} = \{x_0,\ldots,x_{n-1}\}$, and
	%let $\VV_0,\ldots,\VV_{n-1}$ be fresh logical variables.
	%Then, for every $\gnum = \seqnum{\sigma}_{\varseq{x}}$ with $\sigma \in \States$, we have
	%
	For every $\sigma \in \States$ and $\gnum = \seqnum{\sigma}$, we have
	\begin{align*}
	\FF\left[ x_0 / \sigma(x_0) , \ldots ,  x_n / \sigma(x_n) \right] 
	\eequiv
	\gsubst{x}{\FF}{\gnum}~.
	\end{align*}
\end{lemma}
\begin{proof}
	See Appendix~\ref{proof:subst_goedel}.
\end{proof}
If $\FV{\FF} \subseteq \varseq{x}$, then substituting every $x \in \varseq{x}$ by $\sigma'(x)$ in $\FF$ corresponds to \emph{evaluating} $\eval{\FF}$ in $\sigma'$:
\begin{lemma}
	\label{lem:apply_goedel}
	%Let $\FF \in \SyntEStates$ and $\FV{\FF} \subseteq \varseq{x} = \{x_0,\ldots,x_{n-1}\}$.
	%Then, for all states $\sigma, \sigma'$ and every $\gnum = \seqnum{\sigma'}_{\varseq{x}}$, we have
	If $\FV{\FF} \subseteq \varseq{x}$, then , for all states $\sigma, \sigma'$ and every $\gnum = \seqnum{\sigma'}$, we have
	\[
	\sem{\gapply{x}{\FF}{\gnum}}{\sigma}{\interpret}
	\eeq
	\sem{\FF}{}{}(\sigma')~.
	\]
\end{lemma}
\begin{proof}
	See Appendix~\ref{proof:apply_goedel}.
\end{proof}
In particular, $\sem{\gapply{x}{\FH}{\VV}}{\sigma}{\interpret}$ is independent of $\sigma$ since $\gapply{x}{\FH}{\VV}$ does not contain free program variables. \\ \\
\noindent
\emph{Proof of of \Cref{thm:expressive}.}

We prove that

\begin{align*}
&\wp{\WHILEDO{\BB}{\cc_1}}{\sem{\FF}{}{}}(\sigma)   \\
\eeq & \semleft{\Sup length \colon \Sup nums \colon
	\gsumsymbol\big[\vsum, \iverson{\stateseq{x}{\vsum}{length}}}{\sigma}{} \\
&\qquad \qquad \qquad \qquad\qquad\qquad \qquad 
\exprod	\pathexp{length}{\vsum}, {nums} \big] \semright ~.
\end{align*}
Assuming that $\pathexpsymbol$ satisfies Equations (\ref{eqn:pathexp_spec_1}) and (\ref{eqn:pathexp_spec_2}), we have
\begin{align*}
& \semleft{\Sup length \colon \Sup nums \colon
	\gsumsymbol\big[\vsum, \iverson{\stateseq{x}{\vsum}{length}}}{\sigma}{} \\
&\qquad \qquad \qquad \qquad\qquad\qquad \qquad 
\exprod	\pathexp{length}{\vsum}, {nums} \big] \semright \\
\eeq&  \sup \big\{
\semleft{ \Sup nums \colon
	\gsumsymbol\big[\vsum, \iverson{\stateseq{x}{\vsum}{length}}}{\sigma\statesubst{length}{r}}{} \\
&\qquad \qquad \qquad \qquad\qquad\qquad \qquad 
\exprod	\pathexp{length}{\vsum}, {nums} \big] \semright 
~\mid~ {r \in \PosRats} \big\}
\tag{by definition} \\
\eeq&\sup_{r \in \PosRats}
\sem{\Sup nums \colon
	\gsum{\iverson{\stateseq{x}{\vsum}{r}} \exprod	\pathexp{r}{\vsum}}{nums}}{\sigma}{\interpret}
\tag{rewrite supremum} \\
\eeq&\sup_{r\in \PosRats} \sum\limits_{j=0}^{\infty}
\sem{\iverson{\stateseq{x}{\vsum}{r}} \exprod	\pathexp{r}{\vsum}}{\sigma\statesubst{\vsum}{j}}{\interpret\statesubst{\vsum}{j}}
\tag{by Theorem~\ref{thm:sum_exp}} \\
\eeq & \sup_{r\in \PosRats} \sum\limits_{j=0}^{\infty}
\sem{\iverson{\stateseq{x}{j}{r}} \exprod	\pathexp{r}{j}}{\sigma}{\interpret}
\tag{Lemma \ref{lem:substitution}} \\
\eeq & \sup_{k\in \Nats} \sum_{j=0}^{\infty}
\sem{\iverson{\stateseq{x}{j}{r}} \exprod \pathexp{k}{j}}{\sigma}{\interpret}
\tag{$\pathexp{r}{j}=0$, if $r \not\in \Nats$, replace $r$ by $k$} \\
\eeq & \sup_{k\in \Nats} \sum_{\substack{j=\seqnum{\sigma_0,\ldots,\sigma_{k-1}}
		\\ \sigma_0,\ldots,\sigma_{k-1} \in \partitionedstates{x} \\ \equivstates{\sigma_0}{\varseq{x}}{\sigma}}}
\sem{\pathexp{k}{j}}{\sigma}{\interpret}
\tag{for every sequence $\sigma_0,\ldots,\sigma_{k-1} \in \partitionedstates{x}$ there is exactly one G\"odel number and $\equivstates{\sigma_0}{\varseq{x}}{\sigma}$} \\
\eeq& \sup_{k\in \Nats} \sum_{\substack{j=\seqnum{\sigma_0,\ldots,\sigma_{k-1}}
		\\ \sigma_0,\ldots,\sigma_{k-1} \in \partitionedstates{x}}} \statepred{\sigma_0}{\varseq{x}}(\sigma) \cdot
(\iverson{\neg \BB} \cdot \sem{\FF}{}{})(\sigma_{k-1})
\cdot
\prod\limits_{i=0}^{k-2} \wp{\ITE{\BB}{\cc_1}{\SKIP}}{\statepred{\sigma_{i+1}{}}{}}(\sigma_i)
\tag{by Equation~(\ref{eqn:pathexp_spec_1}) and $\statepred{\sigma_0}{\varseq{x}}(\sigma) = 1$ iff $\equivstates{\sigma_0}{\varseq{x}}{\sigma}$} \\
\eeq& \sup_{k\in \Nats} \sum_{ \sigma_0,\ldots,\sigma_{k-1} \in \partitionedstates{x}}
\statepred{\sigma_0}{\varseq{x}}(\sigma) \cdot
(\iverson{\neg \BB} \cdot \sem{\FF}{}{})(\sigma_{k-1})
\cdot
\prod\limits_{i=0}^{k-2} \wp{\ITE{\BB}{\cc_1}{\SKIP}}{\statepred{\sigma_{i+1}{}}{}}(\sigma_i)
\tag{product does not depend on $j$} \\
\eeq& \wp{\WHILEDO{\BB}{\cc_1}}{\sem{\FF}{}{}}(\sigma)  
\tag{by Theorem~\ref{thm:wp_loop_as_sum}}~.
\end{align*}
It remains to give $\pathexp{\VV_1}{\VV_2}$.
By I.H., there is a $\FG \in \SyntE$ with  
\[
\wp{\ITE{\BB}{\cc_1}{\SKIP}}{\iverson{x_0 = x_0' \wedge \ldots \wedge x_n=x_n'}} \eeq \eval{\FG}~.
\]
Hence, by Lemma~\ref{lem:express_statepred_exp}, we have for every $\sigma \in \States$:
\begin{align}
\label{eqn:g_express_wp_body}
\wp{\ITE{\BB}{\cc_1}{\SKIP}}{\statepred{\sigma}{\varseq{x}}}
\eeq
\eval{\FG\left[ x_0' / \sigma(x_0) , \ldots ,  x_n' / \sigma(x_n) \right]}
\end{align}
Now define
\begin{align*}
&\pathexp{\VV_1}{\VV_2} \\
\ddefeq & 
\iverson{\VV_1 < 2} \cdot (\Sup \gnum \colon \iverson{\seqelem{\VV_2}{\VV_1 -1}{\gnum}}
\exprod \gapply{x}{(\iverson{\neg b} \cdot f)}{\gnum}) \\
+& \iverson{\VV_1 \geq 2} \cdot
(\Sup \gnum \colon \iverson{\seqelem{\VV_2}{\VV_1 -1}{\gnum}}
\exprod \gapply{x}{(\iverson{\neg b} \cdot f)}{\gnum}) \\
&\quad\exprod
\gproductsymbol \big(\Sup \gnum_1 \colon \Sup \gnum_2 \colon
\iverson{\seqelem{\VV_2}{\vprod}{\gnum_1} \wedge \seqelem{\VV_2}{\vprod+1}{\gnum_2}} \\
&\qquad  \quad 
\exprod
\gapply{\varseq{x}}{\gsubst{\varseq{x'}}{g}{\gnum_2}}{\gnum_1}, \VV_1 - 2 \big)
\end{align*}
Here $\Sup \gnum \colon \iverson{\seqelem{\VV_2}{\VV_1 -1}{\gnum}} \exprod \gapply{x}{(\iverson{\neg b} \cdot f)}{\gnum}$ is a shorthand for 
\[
\Sup \gnum \colon \Sup \VV \colon \iverson{\VV +1 = \VV_1} \cdot \iverson{\seqelem{\VV_2}{\VV}{\gnum}} \exprod \gapply{x}{(\iverson{\neg b} \cdot f)}{\gnum}~.
\]
Similarily, $\gproduct{\ldots}{\VV_1 -2}$ is shorthand for
\[
\Sup \VV \colon \iverson{\VV + 2 = \VV_1} \cdot \gproduct{\ldots}{\VV}~.
\]
We now show that $\pathexp{\VV_1}{\VV_2}$ indeed satisfies the specification from (\ref{eqn:pathexp_spec_1}) and (\ref{eqn:pathexp_spec_2}). We distinguish the following cases: \\ \\
\noindent
\emph{The case $\sigma(\VV_1) \not\in \Nats$.} By Theorem~\ref{thm:fo_rats_subsumes_fo_nats}, we have 
\[
\sem{\Sup \gnum \colon \iverson{\seqelem{\VV_2}{\VV_1 -1}{\gnum}} \exprod \gapply{x}{(\iverson{\neg b} \cdot f)}{\gnum}}{\sigma}{\interpret} = 0
\]
and hence $\sem{\pathexp{\VV_1}{\VV_2}}{\sigma}{\interpret} = 0$. \\ \\
\noindent
\emph{The case $\sigma(\VV_1) = 0$.} In this case, we have
\[
\sem{\iverson{\VV +1 = \VV_1}}{\sigma\statesubst{\VV}{\RR}}{\interpret\statesubst{\VV}{\RR}} \eeq 0
\]
for all $\RR \in \PosRats$ and thus $\sem{\pathexp{\VV_1}{\VV_2}}{\sigma}{\interpret} = 0$. \\ \\
\noindent
\emph{The case $\sigma(\VV_1) = 1$ and $\sigma(\VV_2) = \stateseqnum{\sigma_0}$.} We have
\begin{align*}
&\sem{\pathexp{1}{\stateseqnum{\sigma_0}}}{\sigma}{\interpret} \\
\eeq&\sem{\Sup \gnum \colon \iverson{\seqelem{\stateseqnum{\sigma_0}}{0}{\gnum}}
	\exprod \gapply{x}{(\iverson{\neg b} \cdot f)}{\gnum}}{\sigma}{\interpret}
\tag{$\sem{\iverson{\VV_1 \geq 2}}{\sigma}{\interpret} = 0$
	and $\sem{\iverson{\VV_1 < 2}}{\sigma}{\interpret} = 1$} \\
\eeq&\sem{\gapply{x}{(\iverson{\neg b} \cdot f)}{\seqnum{\sigma_0}}}{\sigma}{\interpret}
\tag{$\sem{\iverson{\seqelem{\stateseqnum{\sigma_0}}{0}{\gnum}}}{\sigma}{\interpret} = 1$ only for $\sigma(\gnum) = \seqnum{\sigma_0}$} \\
\eeq & \sem{\iverson{\neg b} \cdot \FF}{}{}(\sigma_0)
\tag{by Lemma~\ref{lem:apply_goedel}} \\
\eeq&(\iverson{\neg b} \cdot \sem{\FF}{}{})(\sigma_0)
\cdot
\prod\limits_{i=0}^{\sigma(\VV_1)-2} \wp{\ITE{b}{\cc_1}{\SKIP}}{\statepred{\sigma_{i+1}{}}{}}(\sigma_i)~.
\tag{empty product equals $1$}
\end{align*}

\noindent
\emph{The case $\sigma(\VV_1) = k \in \Nats$ with $k \geq 2$ and $\sigma(\VV_2) = \stateseqnum{\sigma_0,\ldots,\sigma_{k-1}}$.} We have
\begin{align*}
&\sem{\pathexp{k}{\stateseqnum{\sigma_0, \ldots, \sigma_{k-1}}}}{\sigma}{\interpret} \\
\eeq& \semleft{(\Sup \gnum \colon \iverson{\seqelem{\stateseqnum{\sigma_0,\ldots,\sigma_{k-1}}}{k -1}{\gnum}}
	\exprod \gapply{x}{(\iverson{\neg b} \cdot f)}{\gnum})}{\sigma}{\interpret} \\
&\quad\exprod
\gproductsymbol \big(\Sup \gnum_1 \colon \Sup \gnum_2 \colon
[\seqelem{\stateseqnum{\sigma_0,\ldots,\sigma_{k-1}}}{\vprod}{\gnum_1}\\
& \qquad \quad \wedge \seqelem{\stateseqnum{\sigma_0,\ldots,\sigma_{k-1}}}{\vprod+1}{\gnum_2}] %
\exprod
\gapply{\varseq{x}}{\gsubst{\varseq{x'}}{g}{\gnum_2}}{\gnum_1}, k - 2 \big)
\semright
\tag{$\sem{\iverson{\VV_1 \geq 2}}{\sigma}{\interpret} = 1$
	and $\sem{\iverson{\VV_1 < 2}}{\sigma}{\interpret} = 0$} \\
\eeq& \semleft{(\Sup \gnum \colon \iverson{\seqelem{\stateseqnum{\sigma_0,\ldots,\sigma_{k-1}}}{k -1}{\gnum}}
	\exprod \gapply{x}{(\iverson{\neg b} \cdot f)}{\gnum})}{\sigma}{\interpret} \semright \\
&\quad\cdot\semleft{
	\gproductsymbol \big(\Sup \gnum_1 \colon \Sup \gnum_2 \colon
	[\seqelem{\stateseqnum{\sigma_0,\ldots,\sigma_{k-1}}}{\vprod}{\gnum_1}}{\sigma}{\interpret}\\
& \qquad \quad \wedge \seqelem{\stateseqnum{\sigma_0,\ldots,\sigma_{k-1}}}{\vprod+1}{\gnum_2}] %
\exprod
\gapply{\varseq{x}}{\gsubst{\varseq{x'}}{g}{\gnum_2}}{\gnum_1}, k  -2 \big)
\semright
\tag{by Theorem~\ref{cor:unrestricted_product}}\\
\eeq& \sem{\iverson{\neg b} \cdot \FF}{}{}(\sigma_{k-1}) \\
&\quad\cdot\semleft{
	\gproductsymbol \big(\Sup \gnum_1 \colon \Sup \gnum_2 \colon
	[\seqelem{\stateseqnum{\sigma_0,\ldots,\sigma_{k-1}}}{\vprod}{\gnum_1}}{\sigma}{\interpret}\\
& \qquad \quad \wedge \seqelem{\stateseqnum{\sigma_0,\ldots,\sigma_{k-1}}}{\vprod+1}{\gnum_2}] %
\exprod
\gapply{\varseq{x}}{\gsubst{\varseq{x'}}{g}{\gnum_2}}{\gnum_1}, k -2  \big)
\semright
\tag{see reasoning for previous case} \\
\eeq& \sem{\iverson{\neg b} \cdot \FF}{}{}(\sigma_{k-1})  \\
&\quad\cdot
\prod\limits_{i=0}^{k-2 }\semleft{\Sup \gnum_1 \colon \Sup \gnum_2 \colon
	[\seqelem{\stateseqnum{\sigma_0,\ldots,\sigma_{k-1}}}{i}{\gnum_1}}{\sigma}{\interpret}\\
& \qquad \qquad\quad  \wedge \seqelem{\stateseqnum{\sigma_0,\ldots,\sigma_{k-1}}}{i+1}{\gnum_2}] %
\exprod
\gapply{\varseq{x}}{\gsubst{\varseq{x'}}{g}{\gnum_2}}{\gnum_1}
\semright
\tag{by Theorem~\ref{thm:prod_exp} and $k \geq 2$ by assumption} \\
\eeq& \sem{\iverson{\neg b} \cdot \FF}{}{}(\sigma_{k-1})  \\
&\quad\cdot
\prod\limits_{i=0}^{k-2 }\sem{ \gapply{\varseq{x}}{\gsubst{\varseq{x'}}{g}{\seqnum{\sigma_{i+1}}}}{\seqnum{\sigma_{i}}}}{\sigma}{\interpret}
\tag{$\Sup$ quantifiers enforce $\sigma(\gnum_1) = \seqnum{\sigma_i}$
	and $\sigma(\gnum_2) = \seqnum{\sigma_{i+1}}$} \\
\eeq& \sem{\iverson{\neg b} \cdot \FF}{}{}(\sigma_{k-1})  \\
&\quad\cdot
\prod\limits_{i=0}^{k-2 }\sem{ \gapply{\varseq{x}}{\FG\left[ x_0' / \sigma_{i+1}(x_0) , \ldots ,  x_n' / \sigma_{i+1}(x_n) \right] }{\seqnum{\sigma_{i}}}}{\sigma}{\interpret}
\tag{by Lemma~\ref{lem:subst_goedel}} \\
\eeq& \sem{\iverson{\neg b} \cdot \FF}{}{}(\sigma_{k-1})  \\
&\quad\cdot
\prod\limits_{i=0}^{k-2 }\sem{ \FG\left[ x_0' / \sigma_{i+1}(x_0) , \ldots ,  x_n' / \sigma_{i+1}(x_n) \right] }{}{}(\sigma_i)
\tag{by Lemma~\ref{lem:apply_goedel}}  \\
\eeq& \sem{\iverson{\neg b} \cdot \FF}{}{}(\sigma_{k-1})  \\
&\quad\cdot
\prod\limits_{i=0}^{k-2 }\wp{\ITE{\BB}{\cc_1}{\SKIP}}{\statepred{\sigma_{i+1}}{\varseq{x}}}(\sigma_i)~,
\tag{by Equation~\ref{eqn:g_express_wp_body}} 
\end{align*}
which is what we had to show. Hence, we finally obtain
\begin{align*}
& \wp{\WHILEDO{\BB}{\cc_1}}{\sem{\FF}{}{}} \\
\eeq 
& \semleft{\Sup length \colon \Sup nums \colon
	\gsumsymbol\big[\vsum, \iverson{\stateseq{x}{\vsum}{length}}}{}{} \\
&\qquad \qquad \qquad \qquad\qquad\qquad \qquad 
\exprod	\pathexp{length}{\vsum}, {nums} \big] \semright ~.
\end{align*}
This completes the proof.

\subsection{Proof of Lemma~\ref{lem:subst_goedel}}

\begin{proof}
	\label{proof:subst_goedel}
	We have
	\begin{align*}
	& \gsubst{x'}{\FF}{\gnum} \\
	\eequiv& \Sup \VV_0 \colon \ldots \colon \Sup \VV_{n-1} \colon
	\iverson{\rseqelem{\gnum}{0}{\VV_0}
		\wedge \ldots \wedge \rseqelem{\gnum}{n-1}{\VV_{n-1}}}
	\exprod
	\FF
	\left[ x_0' / \VV_0 , \ldots ,  x_{n-1}' / \VV_{n-1} \right]
	\tag{by definition} \\
	\eequiv & 
	\iverson{\rseqelem{\gnum}{0}{\sigma(x_0)}
		\wedge \ldots \wedge \rseqelem{\gnum}{n-1}{\sigma(x_{n-1})}}
	\exprod
	\FF
	\left[ x_0' / \sigma(x_{0}) , \ldots ,  x_{n-1}' / \sigma(x_{n-1}) \right]
	\tag{for $0 \leq j \leq n-1$, $\rseqelem{\gnum}{j}{m} \equiv 1$ only for $m = \sigma(x_j)$} \\
	\eequiv & \FF
	\left[ x_0' / \sigma(x_{0}) , \ldots ,  x_{n-1}' / \sigma(x_{n-1}) \right]~.
	\tag{$\iverson{\rseqelem{\gnum}{0}{\sigma(x_0)}
			\wedge \ldots \wedge \rseqelem{\gnum}{n-1}{\sigma(x_{n-1})}} \equiv 1$}
	\end{align*}
	
\end{proof}

\subsection{Proof of Lemma~\ref{lem:apply_goedel}}

\begin{proof}
	\label{proof:apply_goedel}
	
	We have
	\begin{align*}
	& \sem{\gapply{x}{\FF}{\gnum}}{\sigma}{\interpret} \\
	\eeq & \sem{\Sup \VV_0 \colon \ldots \colon \Sup \VV_{n-1} \colon
		\iverson{\rseqelem{\gnum}{0}{\VV_0}
			\wedge \ldots \wedge \rseqelem{\gnum}{n-1}{\VV_{n-1}}}
		\exprod
		\FF
		\left[ x_0 / \VV_0 , \ldots ,  x_{n-1} / \VV_{n-1} \right]}{\sigma}{\interpret}
	\tag{by definition} \\
	\eeq& \sem{\iverson{\rseqelem{\gnum}{0}{\sigma'(x_0)}
			\wedge \ldots \wedge \rseqelem{\gnum}{n-1}{\sigma'(x_{n-1})}}
		\exprod
		\FF
		\left[ x_0 / \sigma'(x_0) , \ldots ,  x_{n-1} / \sigma'(x_{n-1}) \right]}{\sigma}{\interpret}
	\tag{for $0 \leq j  \leq n-1$, $\sem{\rseqelem{\gnum}{j}{\VV_{j}}}{\sigma}{\interpret} = 1$ only for $\sigma(\VV_j) = \sigma'(x_j)$, Lemma~\ref{lem:substitution}} \\
	\eeq &\sem{\FF
		\left[ x_0 / \sigma'(x_0) , \ldots ,  x_{n-1} / \sigma'(x_{n-1}) \right]}{\sigma}{\interpret}
	\tag{$\iverson{\rseqelem{\gnum}{0}{\sigma'(x_0)}
			\wedge \ldots \wedge \rseqelem{\gnum}{n-1}{\sigma'(x_{n-1})}} \eequiv 1$} \\
	\eeq &\sem{\FF}{}{}(\sigma')~.
	\tag{$\Vars(\sem{\FF}{}{})\subseteq \FV{\FF} \subseteq  \varseq{x}$ and Lemma~\ref{lem:substitution}}
	\end{align*}

\end{proof}